\documentclass{lmcs}
\pdfoutput=1

\usepackage[utf8]{inputenc}

\usepackage{lastpage}
\lmcsdoi{17}{2}{22}
\lmcsheading{}{\pageref{LastPage}}{}{}%
{Jun.~11,~2020}{May~27,~2021}{}

\usepackage{amsfonts}
\usepackage{amssymb, latexsym}
\usepackage{amsmath}
\allowdisplaybreaks
\usepackage{amsthm}

\usepackage{empheq}

\usepackage[
colorlinks=true,       
linkcolor=blue,          
citecolor=blue,        
plainpages=false,      
pdfpagelabels,
]{hyperref}

\usepackage{cite}

\usepackage{tikz-cd}
\tikzset{close/.style={near start,outer sep=-10pt}}

\usepackage[shortlabels]{enumitem}

\makeatletter
\newcommand{\myitem}[1][]{%
\item[#1]\protected@edef\@currentlabel{#1}\ignorespaces%
}
\makeatother

\makeatletter
\def\widebreve{\mathpalette\wide@breve}
\def\wide@breve#1#2{\sbox\z@{$#1#2$}%
     \mathop{\vbox{\m@th\ialign{##\crcr
\kern0.08em\brevefill#1{0.8\wd\z@}\crcr\noalign{\nointerlineskip}%
                    $\hss#1#2\hss$\crcr}}}\limits}
\def\brevefill#1#2{$\m@th\sbox\tw@{$#1($}%
  \hss\resizebox{#2}{\wd\tw@}{\rotatebox[origin=c]{90}{\upshape(}}\hss$}
\makeatletter

\newcommand{\amp}{\mathrel{\,\&\,}}
\newcommand{\imp}{\;\rightarrow\:}
\newcommand{\defeqiv}{\stackrel{\textup{def}}{\iff}}
\newcommand{\defeql}{\stackrel{\textup{def}}{\  =\  }}

\DeclareMathOperator*{\medwedge}{{\textstyle{\bigwedge}} }
\DeclareMathOperator*{\medvee}{{\textstyle{\bigvee}}}
\newcommand{\prox}{\prec}

\newcommand{\cov}{\mathrel{\lhd}}
\newcommand{\fcov}{\mathrel{\blacktriangleleft}}
\newcommand{\fprox}{\sqsubset}
\newcommand{\fproxeq}{\sqsubseteq}
\newcommand{\cutcomp}{\mathop{\cdot}}

\newcommand{\meets}{\between}
\newcommand{\downset}{\mathord{\downarrow}}
\newcommand{\upset}{\mathord{\uparrow}}

\newcommand{\id}{\mathrm{id}}

\newcommand{\Pow}[1]{\mathsf{Pow}(#1)}
\newcommand{\Fin}[1]{\mathsf{Fin}(#1)}
\newcommand{\Sin}[1]{\mathsf{Sin}(#1)}
\newcommand{\PFin}[1]{\mathsf{Fin}^{+}\!(#1)}
\newcommand{\FFin}[1]{\Fin{\Fin{#1}}}

\newcommand{\Image}[1]{\mathrm{Im}\,#1 }
\newcommand{\Choice}[1]{\mathrm{Ch}(#1)}

\newcommand{\Nat}{\mathbb{N}}

\newcommand{\Upper}[1]{\mathrm{P_U}(#1)}
\newcommand{\UpperFunc}{\mathrm{P_U}}

\newcommand{\Lower}[1]{\mathrm{P_L}(#1)}
\newcommand{\LowerFunc}{\mathrm{P_L}}

\newcommand{\Ideals}[1]{\mathrm{Idl}(#1)}

\newcommand{\RIdeals}[1]{\mathrm{RIdl}(#1)}
\newcommand{\Karoubi}[1]{\mathsf{\mathbf{Split}(#1)}}

\newcommand{\Cat}[1]{\mathbb{#1}}
\newcommand{\AlgDom}{\mathsf{AlgDom}}

\newcommand{\AlgPxPos}{\mathsf{Pos_{App}}}

\newcommand{\AlgPxSLat}{\mathsf{JLat_{JApp}}}

\newcommand{\AlgPxSLatS}{\mathsf{JLat_{App}}}

\newcommand{\PxPos}{\mathsf{PxPos}}
\newcommand{\PxJLat}{\mathsf{PxJLat}}
\newcommand{\SPxJLat}{\mathsf{SPxJLat}}

\newcommand{\PxJLatLoc}{\mathsf{LSPxJLat}}
\newcommand{\Cont}{\mathsf{ContDom}}
\newcommand{\ContLat}{\mathsf{ContLat}}
\newcommand{\ContLatS}{\mathsf{ContLat_{Scott}}}
\newcommand{\AlgLat}{\mathsf{AlgLat}}
\newcommand{\AlgLatS}{\mathsf{AlgLat_{Scott}}}
\newcommand{\FreeSL}{\mathsf{FreeSupLat_{Scott}}}
\newcommand{\Infosys}{\mathsf{Infosys}}
\newcommand{\FCov}{\mathsf{FCov}}
\newcommand{\ContFCov}{\mathsf{ContFCov}}
\newcommand{\SContFCov}{\mathsf{SContFCov}}

\newcommand{\ContFCovLoc}{\mathsf{LSContFCov}}

\newcommand{\SEnt}{\mathsf{SEnt}}
\newcommand{\SEntsys}{\Karoubi{\SEnt}}

\newcommand{\FTop}{\mathsf{FTop}}
\newcommand{\LKFTop}{\mathsf{LKFTop}}
\newcommand{\LKFrm}{\mathsf{LKFrm}}

\newcommand{\BCov}{\mathsf{BCov}}
\newcommand{\ContBCov}{\mathsf{ContBCov}}

\newcommand{\CovtoLat}[1]{\mathrm{L}(#1)}

\newcommand{\ApproxExt}[1]{\mathrel{\widetilde{#1}}}

\newcommand{\KZ}{\mathrm{KZ}}
\newcommand{\coAlg}[1]{\mathsf{coAlg(#1)}}

\newcommand{\wb}{\mathsf{wb}}

\theoremstyle{theorem}
\newtheorem{theorem}[thm]{Theorem}
\newtheorem{proposition}[thm]{Proposition}
\newtheorem{lemma}[thm]{Lemma}
\newtheorem{corollary}[thm]{Corollary}
\theoremstyle{definition}
\newtheorem{definition}[thm]{Definition}
\newtheorem{remark}[thm]{Remark}

\newtheorem{example}[thm]{Example}
\theoremstyle{thmC}
\newtheorem{corollaryC}[thm]{Corollary}
\numberwithin{equation}{section}

\begin{document}
\title[Predicative theories of continuous lattices]{Predicative theories of continuous lattices}

\author[T.~Kawai]{Tatsuji Kawai} 
\address{Japan Advanced Institute of Science and Technology,
1-1 Asahidai, Nomi, Ishikawa 923-1292, Japan}	
\email{tatsuji.kawai@jaist.ac.jp} 

\begin{abstract}
  We introduce a notion of strong proximity join-semilattice, a
  predicative notion of continuous lattice which arises as the Karoubi
  envelop of the category of algebraic lattices.
  Strong proximity join-semilattices can be characterised
  by the coalgebras of the lower powerlocale on the wider category of
  proximity posets (also known as abstract bases or R-structures).
  Moreover, locally compact locales can be characterised in terms of
  strong proximity join-semilattices by the coalgebras of the double
  powerlocale on the category of proximity posets.
  We also provide more logical characterisation of a strong
  proximity join-semilattice, called a strong continuous finitary
  cover, which uses an entailment relation to present the underlying
  join-semilattice.
  We show that this structure naturally corresponds to the notion
  of continuous lattice in the predicative point-free topology.
  Our result makes the predicative and finitary aspect of the notion
  of continuous lattice in point-free topology more explicit.


\end{abstract}
\maketitle

\section{Introduction}\label{sec:Introduction}
A continuous lattice  \cite{gierz2003continuous} is a complete lattice
$L$ in which every element $x \in L$  is expressed as a directed
join of elements way-below $x$, where $y$ is \emph{way-below} $x$ if
\[
  x \leq \bigvee U \imp \exists z \in U \left( y \leq z \right)
\]
for every directed subset $U \subseteq L$.
The reference to an arbitrary directed subset in the
definition of the way-below relation, however,  makes the theory of continuous
lattice difficult to develop in predicative mathematics such as those of
Bishop \cite{Bishop-67} or Martin-L\"of \cite{per1984intuitionistic},
where powers of sets are not accepted as legitimate objects.  Negri
\cite{negri1998continuous,NegriContinuousDomain} has addressed this
problem in the predicative point-free topology
(known as formal topology \cite{Sambin:intuitionistic_formal_space}).
Her notion of continuous lattice assumes only the structure
on the basis of a continuous lattice together with a primitive
way-below relation on the basis (see Definition \ref{def:ContBCov}).
However, the condition on the primitive way-below relation refers to
an arbitrary element of the continuous lattice represented by such a
structure. Hence, the structure is not a predicative presentation of
a continuous lattice; rather, it should be considered as specifying
the condition such a predicative presentation, if it exists, should
satisfy. The problem of how to present continuous lattices in a
uniform way using only predicative data remains open.

The aim of this paper is to address the above predicativity problem by
introducing several finitary presentations of continuous lattices in
terms of the basis of a domain. These presentations are \emph{finitary} in
that they do not make any reference to an arbitrary element of the continuous
lattices they represent but only to finite joins of their bases.

Our idea is built on a well-known fact in domain theory: a continuous
domain can be completely recovered from a certain structure on its
basis, which can be simply described using an idempotent relation on
the basis.  Such a structure on the basis is called an abstract basis
by Abramsky and Jung~\cite[Section 2.2.6]{abramsky1994domain} and an
R-structure by Smyth~\cite{SymthEffectivelyGivenDomain}.  
In this paper, we introduce
corresponding structures for continuous lattices, called
\emph{proximity $\vee$-semilattice} and \emph{strong proximity
$\vee$-semilattice}.
Both proximity $\vee$-semilattice and its strong variant are abstract
bases enriched with a structure of join semilattice
(\emph{$\vee$-semilattice} for short).  Proximity $\vee$-semilattices
naturally arise as the Karoubi envelop  of the category of algebraic lattices and Scott
continuous functions, while strong proximity $\vee$-semilattices arise
as the Karoubi envelop of the category of algebraic lattices and
suplattice homomorphisms.  The latter also arise as the coalgebras of
the lower powerlocale on the category of proximity posets, the  dual
of the category of abstract basis. Moreover, locally compact locales
admit natural characterisation in terms of strong proximity
$\vee$-semilattice as the coalgebras of the double powerlocale on the
category of proximity posets. 

We also introduce more logical characterisation of a (strong)
proximity $\vee$-semilattice, called \emph{continuous
finitary cover} and \emph{strong continuous finitary cover}.
These structures use entailment relations enriched
with idempotent relations to present (strong) proximity
$\vee$-semilattices, and they are naturally found in the Karoubi
envelop of the category of free suplattices and Scott continuous
functions.
In particular, a strong continuous finitary cover can be regarded as a
finitary presentation of a continuous lattice in formal topology by
Negri~\cite{negri1998continuous}. We make this observation explicit by
establishing an equivalence of the category of strong continuous
finitary covers and that of continuous lattices in formal topology.
The result clarifies the predicative and finitary aspect of the
notions introduced in \cite{negri1998continuous}.

\subsubsection*{Related works}
In the context of stably compact locales, the structures analogous to
those of proximity $\vee$-semilattice and continuous finitary cover
have already appeared. Smyth~\cite{SmythStableCompactification} and
Jung and S\"underhauf~\cite{JungSunderhaufDualtyCompactOpen} gave 
predicative characterisation of stably compact locales using proximity
lattices and their strong variants, respectively.
These notions were studied further by van
Gool~\cite{vanGoolDualityCanExt}, who gave an explicit characterisation
of the category of stably compact locales by the Karoubi envelop of the
category of spectral locales (the ideal completions of distributive
lattices). On the other hand, Vickers~\cite{VickersEntailmentSystem} 
characterised the category of stably compact frames and preframe
homomorphisms as the Karoubi envelop of its subcategory of free
frames.  Building on the earlier work by
Jung, Kegelmann, and Moshier~\cite{JungEtAlMultilingual}, he also gave
a logical presentation of a proximity lattice using an entailment
relation with interpolation property. The strong variant of this notion
was investigated by Coquand and Zhang
\cite{CoquandZhangPrediativePatch} and Kawai~\cite{KawaiContEnt}.
These structures for stably compact locales are analogous to proximity
$\vee$-semilattices and continuous finitary covers and their strong
variants. The difference between the previous works stated above and
this paper is as follows: the previous works deal with the Karoubi envelop of the
category of spectral locales or that of free frames and preframe
homomorphisms; on the other hand, this paper deals with the Karoubi
envelop of the category of algebraic lattices (ideal completions of
$\vee$-semilattices) or that of free suplattices and 
Scott continuous functions.

\subsubsection*{Notations}
  We fix some notations for sets and relations which are used
  throughout the paper.
  If $S$ is a set, then $\Sin{S}$ denotes the set of singleton subsets
  of $S$;  $\Fin{S}$ denotes the set of finitely enumerable subsets of
  $S$. Here, a set $X$ is \emph{finitely
  enumerable}~\cite[Section~7]{CTSbook} (or \emph{Kuratowski
  finite}~\cite[D5.4]{ElephantII}) if there exists a surjection $f
  \colon \left\{ 0,\dots,n-1 \right\} \to X$ for some $n \in \Nat$.
  We also write $\Pow{S}$ for the collection of subsets of $S$ (usually called
  the powerset of $S$), which does not form a set predicativity unless $S = \emptyset$.

  Let $r \subseteq X \times Y$ be a relation between sets $X$ and $Y$.
  For subsets $U \subseteq X$ and $V \subseteq Y$, the forward
  image $rU \subseteq Y$ and the inverse image $r^{-}V \subseteq X$ are defined by 
  \[
    \begin{aligned}
      rU \defeql 
      \left\{ y \in Y \mid \exists x \in U \left( x \mathrel{r} y
    \right) \right\}, \\
    r^{-}V \defeql 
    \left\{ x \in X \mid \exists y \in V \left( x \mathrel{r} y
  \right) \right\}.
\end{aligned}
\]
  For $x \in X$ and $y \in Y$, we write $rx$ for $r\{x\}$ and $r^{-}y$
  for $r^{-}\{y\}$. 
  The opposite of $r$ is denoted by $r^{-}$. 
  The relational composition of $r \subseteq X \times Y$ and $s
  \subseteq Y \times Z$ is denoted by $s \circ r$.

\section{Splitting of idempotents}\label{sec:SplitIdm}
We recall the
following fact~\cite[Chapter VII,
Section 2.3]{johnstone-82}, which underlies the development of the
later sections:
\begin{enumerate}
  \item Continuous domains are the Scott continuous retracts of
    algebraic domains.
  \item Continuous lattices are the suplattice retracts of algebraic lattices.
\end{enumerate}
These allow us to characterise continuous domains (or continuous lattices) in terms of
the Karoubi envelop of the category of algebraic domains (resp.\
algebraic lattices) with suitable notions of morphisms.
 
We recall some standard notions in domain theory
\cite{gierz2003continuous}.  A \emph{dcpo} is a poset $(D,\leq)$ in
which every directed subset has a least upper bound, where a subset $U
\subseteq D$ is \emph{directed} if it is inhabited and any two
elements of $U$ have an upper bound in $U$.  A function between dcpos
$D$ and $D'$ is \emph{Scott continuous} if it preserves directed
joins.  For elements $x,y$ of a dcpo $D$, we say that $y$ is
\emph{way-below} $x$, denoted $y \ll x$, if 
\[
  x \leq \bigvee U \imp \exists z \in U \left( y \leq z\right) 
\]
for every directed subset $U \subseteq D$. A \emph{continuous domain}
is a dcpo in which every element is a directed join of elements
way-below it. 
A \emph{continuous lattice} is a continuous domain with finite joins.
Thus, every continuous lattice is a \emph{suplattice}, a poset with
all joins. A \emph{suplattice homomorphism} between continuous
lattices is a function which preserves all joins. 

\begin{proposition}
  \label{prop:SplitCont}
  Let $D$ be a continuous domain and $f \colon D \to D$ be an
    idempotent Scott continuous function.
    Let $ D_{f} =\left\{ f(a) \mid a \in D \right\}$.
    \begin{enumerate}
      \item\label{prop:SplitCont1} 
        $D_{f}$ is a continuous domain.

      \item\label{prop:SplitCont2} If $D$ is a continuous lattice,
        then so is $D_{f}$.

      \item\label{prop:SplitCont3} If $D$ is a continuous lattice and
        $f$ is a suplattice homomorphism, then $D_{f}$ is a sub-semilattice of $D$.
    \end{enumerate}
\end{proposition}
\begin{proof}
  \ref{prop:SplitCont1}.
    We recall the proof from  Johnstone~\cite[Chapter VII, Lemma
    2.3]{johnstone-82} (see also Vickers~\cite[Theorem
    3]{VickersEntailmentSystem}).
    First, the way-below relation $\ll_{f}$ of
    $D_{f}$ can be characterised in terms of the way-below relation
    $\ll$ of $D$ as follows:
    \begin{equation}
      \label{eq:WayBelowSplit}
      f(a) \ll_{f} f(b) \leftrightarrow \exists c \in D \left( f(a) \leq
      f(c) \amp c \ll f(b) \right).
    \end{equation}
    It is easy to see that the set $\downset_{\ll_{f}}
    f(a) = \left\{ f(b) \in D_{f} \mid f(b) \ll_{f} f(a) \right\}$
    is directed. Then 
    \begin{align*}
    \bigvee \downset_{\ll_{f}} f(a)
    = 
    \bigvee \left\{ f(b) \in D_{f}  \mid b \ll f(a)\right\} 
    = 
    f\left(\bigvee \left\{ b \in D  \mid b \ll f(a)\right\} \right)
    = 
    f(a),
    \end{align*}
    where the second and the third equations follow from the Scott
    continuity and the idempotency of $f$, respectively.
    \smallskip

  \noindent\ref{prop:SplitCont2}.
    If $D$ is a continuous lattice with a $\vee$-semilattice
    structure $(D, 0, \vee)$, then
     $D_{f}$ admits a $\vee$-semilattice structure as follows:
      \begin{align}
        \label{eq:FiniteJoinsInSplit}
        0_{f} &\defeql f(0), &
        f(a) \vee_{f} f(b) &\defeql f(f(a) \vee f(b)).
      \end{align}

  \noindent\ref{prop:SplitCont3}.
  Immediate from \eqref{eq:FiniteJoinsInSplit}, noting that
  $f$ is a $\vee$-semilattice homomorphism.
\end{proof}

  Let $(S,\leq)$ be a poset.  An \emph{ideal completion} of a poset
  $(S,\leq)$ is the
  collection $\Ideals{S}$ of ideals of $S$ ordered by inclusion,
  where an \emph{ideal} is a downward closed and upward directed subset of $S$. 
  The poset $(\Ideals{S}, \subseteq)$ is a continuous domain:
  directed joins of ideals are unions; the way-below relation is
  characterised by
  \begin{equation}
    \label{eq:WayBelowIdeals}
    I \ll J \leftrightarrow \exists a \in J \left( I \subseteq
    \downset_{\leq} a \right),
  \end{equation}
  where $\downset_{\leq} a \defeql \left\{ b \in S \mid b \leq a
  \right\}$.
  An \emph{algebraic
  domain} is a continuous domain which can be expressed as the ideal
  completion of a poset.  An \emph{algebraic lattice} is a continuous
  lattice which can be expressed as the ideal completion of a
  $\vee$-semilattice.  For an algebraic lattice which is expressed as $\Ideals{S}$
  for some $\vee$-semilattice $(S, 0, \vee)$, its finite joins can be
  characterised by
\begin{align}
  \label{eq:FiniteJoinsIdeals}
  0 &\defeql \left\{ 0 \right\}, &
  I \vee J  &\defeql \bigcup_{a \in I, b \in J} \downset_{\leq} (a \vee
  b).
\end{align}

\pagebreak[2]
\begin{proposition}
  \label{prop:ScotRetract}
      \leavevmode
  \begin{enumerate}
    \item\label{prop:ScotRetract1}
  Every continuous domain is a Scott continuous retract of an algebraic
  domain.
    \item\label{prop:ScotRetract3}
  Every continuous lattice is a suplattice retract of an algebraic
  lattice.
  \end{enumerate}
\end{proposition}
\begin{proof}
  Every continuous domain (or continuous lattice) is a Scott
  continuous (resp.\ suplattice) retract of its ideal completion. See
  Johnstone~\cite[Chapter~VII, Theorem~2.3]{johnstone-82}.
\end{proof}

The following construction plays a fundamental role in this paper.
\begin{defiC}[{\cite[Chapter 2, Exercise B]{FreydAbelCat}}]
  An \emph{idempotent} in a category $\Cat{C}$ is a morphism $f \colon
  A \to A$  such that $f \circ f = f$.  The \emph{Karoubi envelop} (or
  \emph{splitting of idempotents}) of $\Cat{C}$ is a category
  $\Karoubi{\Cat{C}}$ where objects are idempotents in $\Cat{C}$ and
  morphisms $h \colon (f \colon A \to A) \to (g \colon B \to B)$
  are morphisms $h \colon A \to B$ in $\Cat{C}$ such that $g \circ h =
  h = h \circ f$. 
\end{defiC}
An idempotent $f \colon A \to A$ in $\Cat{C}$ \emph{splits} if there
exists a pair of morphisms $r \colon A \to B$ and $s \colon B \to A$
such that $s \circ r = f$ and $r \circ s = \id_{B}$.
If $\Cat{C}$ is a full subcategory of $\Cat{D}$ where every idempotent
splits in $\Cat{D}$ and if every object in $\Cat{D}$ is a retract of an
object of $\Cat{C}$, then $\Cat{D}$ is equivalent to
$\Karoubi{\Cat{C}}$. 

\begin{table}[tb]
\begin{tabular}{ll}
  $\Cont$& 
    the category of continuous domains and Scott continuous
    functions\\
  $\AlgDom$& 
    the full subcategory of $\Cont$ consisting of algebraic domains\\
  $\ContLatS$&
    the full subcategory of $\Cont$ consisting of continuous
    lattices\\
  $\AlgLatS$&
    the full subcategory of $\Cont$ consisting of algebraic
    lattices \\
  $\ContLat$&
    the category of continuous lattices and suplattice
    homomorphisms\\
  $\AlgLat$&
    the full subcategory of $\ContLat$ consisting of algebraic
    lattices
\end{tabular}
\caption{Subcategories of continuous domains.}
\label{tab:CategoriesContDom}
\end{table}
\begin{theorem}
  \label{thm:ContContLatKaroubi}
Consider the subcategories of continuous domains in
Table~\ref{tab:CategoriesContDom}.
  \begin{enumerate}
    \item $\Cont$ is equivalent to $\Karoubi{\AlgDom}$.
    \item $\ContLatS$ is equivalent to $\Karoubi{\AlgLatS}$.
    \item $\ContLat$ is equivalent to $\Karoubi{\AlgLat}$.
  \end{enumerate}
\end{theorem}
\begin{proof}
By Proposition \ref{prop:SplitCont}, every idempotent in
$\Cont$, $\ContLatS$, and $\ContLat$ splits. The desired conclusion
follows from Proposition~\ref{prop:ScotRetract}.
\end{proof}

\section{Proximity \texorpdfstring{$\vee$}{join}-semilattices}\label{sec:SPSuplat}
We introduce predicative characterisations of the Karoubi envelops of the
category of algebraic domains and its subcategories of algebraic
lattices. 
The left column of Table \ref{tab:CategoriesPxPos} shows some of
the major structures introduced in this section. The categories
of these structures provide predicative characterisations of the duals of
the Karoubi envelops of the categories on the middle column as well as 
the duals of the categories on the right column (cf.\ Theorem
\ref{thm:ContContLatKaroubi}).\footnote{$\LKFrm$ in Table \ref{tab:CategoriesPxPos}
denotes the category of locally compact frames, i.e., the dual
of the category of locally compact locales (cf.\ Section \ref{sec:LocalizedPxSL}).}

\begin{table}[t]
  \renewcommand{\arraystretch}{1.2}
\begin{tabular}{ccc}
  %
  %
  Predicative characterisation
  & Karoubi envelop & Domain theoretic dual\\
  \hline
  \rule[4pt]{0pt}{14pt}
  \parbox{14em}{proximity posets \\ 
                 + approximable relations}
  & 
  $\AlgDom$ 
  & 
  $\Cont$ 
  \\
  \rule[6pt]{0pt}{14pt}
  \parbox{14em}{proximity $\vee$-semilattices \\
                 + approximable relations} 
  & 
  $\AlgLatS$
  & 
  $\ContLatS$
  \\
  \rule[6pt]{0pt}{14pt}
  \parbox{14em}{strong proximity $\vee$-semilattices \\
  + join-approximable relations} 
  & 
  $\AlgLat$
  & 
  $\ContLat$
  \\
  \rule[6pt]{0pt}{14pt}
  \parbox{14em}{localized strong prox.\
    $\vee$-semilat.\ \\ + proximity relations}
  & ---
  & $\LKFrm$
\end{tabular}
\caption{Main structures in Section~\ref{sec:SPSuplat}.}
\label{tab:CategoriesPxPos}
\end{table}
\subsection{Proximity posets and proximity 
  \texorpdfstring{$\vee$}{join}-semilattices}
\label{sec:ProxPoset}
Our predicative characterisation of continuous domains rests on an elementary
characterisation of Scott continuous functions between algebraic
domains. 

First, we recall that the ideal completion of a poset is the free dcpo on that poset.
\begin{lemma}
  \label{lem:MonotoneScottCont}
  For each poset $(S,\leq)$, there is a monotone (i.e.\ order
  preserving) function  $i_{S} \colon S \to \Ideals{S}$ 
      defined by
      \begin{equation}\label{eq:InjectionGenerator}
        i_{S}(a) \defeql \downarrow_{\leq} a
      \end{equation}
      with the following universal property:
   for any monotone function $f \colon
   S \to D$ to a dcpo $D$, there exists a unique Scott continuous function
   $\overline{f} \colon \Ideals{S} \to D$ such that
   $\overline{f} \circ i_{S} = f$.
\end{lemma}
\begin{proof}
See Vickers \cite[Proposition 9.1.2 (v)]{vickers1989topology}.  The
unique extension $\overline{f} \colon \Ideals{S} \to D$ is defined by
\begin{equation}
  \label{eq:UniquExntension}
  \overline{f}(I) \defeql \bigvee_{a \in I} f(a).
  \qedhere
\end{equation}
\end{proof}

For posets $(S,\leq)$ and $(S', \leq')$, a monotone function $f \colon
S' \to \Ideals{S}$ corresponds to a relation between $S$ and $S'$ as
characterised below.
\begin{definition}
  \label{def:ApproximableRelation}
  Let $(S,\leq)$ and $(S', \leq')$ be posets. A relation $r \subseteq
  S \times S'$ is \emph{approximable} if
  \begin{enumerate}[align=left]
    \myitem[(AppI)]\label{def:approximableI}
    $r^{-} b$
      is an ideal for each $b \in S'$,
    \myitem[(AppU)]\label{def:approximableU}
    $r a$
     is an upward closed subset for each $a \in S$.
  \end{enumerate}
\end{definition}

\begin{proposition}
  \label{prop:BijApproxScottCont}
  Let $(S,\leq)$ and $(S',\leq')$ be posets. There exists a bijective
  correspondence between approximable relations from $S$ to $S'$ and
  Scott continuous functions from $\Ideals{S'}$ to $\Ideals{S}$.
  Through this correspondence, the identity function
  on $\Ideals{S}$ corresponds to the order $\leq$ on $S$
  and the composition of two Scott continuous functions
  contravariantly corresponds
  to the relational composition of the approximable relations.
\end{proposition}
\begin{proof}
An approximable relation $r \subseteq S \times S'$
bijectively corresponds to a monotone function $f_{r} \colon S' \to
\Ideals{S}$ defined by
\begin{equation}
  \label{eq:ApproxToScott}
  f_{r}(b) \defeql r^{-} b,
\end{equation}
and hence to a Scott continuous function $\overline{f_{r}} \colon
\Ideals{S'} \to \Ideals{S}$ defined in \eqref{eq:UniquExntension}.
%
%
Conversely, each Scott continuous function $f \colon \Ideals{S'} \to \Ideals{S}$
determines an approximable relation $r_{f} \subseteq S \times S'$ by 
\begin{equation}
  \label{eq:ScottToApprox}
  a \mathrel{r_{f}} b \defeqiv a \in f(i_{S'}(b)).
\end{equation}
It is straightforward to check that the above correspondence is bijective.
The second statement is also straightforward to check.
\end{proof}

Let $\AlgPxPos$ be the category in which objects are posets and
morphisms are approximable relations
between posets: the identity on a poset is its underlying order;
the composition of two approximable relations is the relational composition.

\begin{proposition}\label{prop:AlgPxPos}
  $\AlgPxPos$ is dually equivalent to $\AlgDom$.
\end{proposition}
\begin{proof}
  Immediate from Proposition~\ref{prop:BijApproxScottCont}.
\end{proof}

By Theorem \ref{thm:ContContLatKaroubi},
we have the following characterisation of the category of continuous
domains.
\begin{theorem}
  \label{prop:SplitAlgPxPos}
   $\Karoubi{\AlgPxPos}$ is dually equivalent to $\Cont$. 
\end{theorem}

%
%
The objects and
morphisms of $\Karoubi{\AlgPxPos}$ can be explicitly described  as follows.
\begin{definition}
A \emph{proximity poset} is a structure $(S, \leq, \prox)$
where $(S, \leq)$ is a poset and $\prox$ is a relation on $S$
satisfying 
\begin{enumerate}
  \item
  $\downset_{\prox} a \defeql \left\{ b \in S \mid b \prox a \right\}$
    is a rounded ideal of $S$,

  \item
  $\upset_{\prox} a \defeql \left\{ b \in S \mid b
    \succ a \right\}$
    is a rounded upward closed subset $S$.
\end{enumerate}
Here, an ideal $I$ is \emph{rounded} if 
\[
  a \in I \leftrightarrow \exists b \succ a \left( b \in I\right).
\]
Similarly, an upward closed subset $U \subseteq S$ is \emph{rounded} if
\[
  a \in U \leftrightarrow \exists b \prox a \left( b \in U\right).
\]
We write $(S, \prox)$ or simply $S$ for a proximity poset, leaving the
underlying order of $S$ implicit.
\end{definition}

\begin{definition}
  \label{def:ApproximableRelationRounded}
  Let $(S,\prox)$ and  $(S',\prox')$ be proximity posets.
  A relation $r \subseteq S \times S'$ is \emph{approximable} if
  $r$ is an approximable relation between the underlying posets of $S$ and $S'$
  such that ${\mathrel{r} \circ \prox} = r = {\prox' \circ
    \mathrel{r}}$. 
  Equivalently, a relation $r \subseteq S \times S'$ is approximable if 
  \begin{enumerate}
    \item
      $r^{-} b$ is a rounded ideal for each $b \in S'$,
    \item 
      $r a$ is a rounded upward closed subset for each $a \in S$.
    \end{enumerate}
\end{definition}

Henceforth, we write $\PxPos$ for $\Karoubi{\AlgPxPos}$: objects
of $\PxPos$ are proximity posets and morphisms between them are
approximable relations. The identity on a proximity poset $(S, \prox)$
is $\prox$; the composition of two approximable relations is the relational
composition. Note that $\AlgPxPos$ is a full subcategory of $\PxPos$,
where each poset $(S, \leq)$ is identified with a proximity poset with
$\leq$ as its idempotent relation. Thus, the terminology
\emph{approximable relation} for morphisms of $\PxPos$ is consistent
with that of $\AlgPxPos$.

\begin{remark}
  \label{rem:ConnectionToRstructure}
  \leavevmode
\begin{enumerate}
  \item 
  The notion of  proximity poset is essentially equivalent to that of
  \emph{abstract basis}~\cite[Definition~2.2.20]{abramsky1994domain}
  or \emph{R-structure}~\cite{SymthEffectivelyGivenDomain}, 
  %
  %
  where the morphisms between abstract bases are also characterised as
  approximable relations~\cite[Definition~2.2.27]{abramsky1994domain}.
  An abstract basis lacks the underlying poset structure on $S$. In this
  paper, we include the poset structure to stress the fact that every continuous
  domain is a retract of an algebraic domain, which is represented by
  that poset structure. 

  \item 
  Another closely related notion is that of \emph{infosys} and
  approximable mapping by Vickers~\cite[Definition~2.1 and
  Definition~2.18]{Infosys}.  Constructively, however, the notion of
  infosys seems to be more general than that of continuous domain,
  and hence does not seem to be equivalent to proximity poset.
  Nevertheless, some of the conjectures Vickers has made in the
  context of infosys are proved to be correct for proximity posets;
  see Section~\ref{sec:LowerPowerLoc} and
  Section~\ref{sec:DoublePowerLoc}.
\end{enumerate}
\end{remark}

To obtain a similar characterisation of continuous lattices,
consider a full subcategory $\AlgPxSLatS$ of  $\AlgPxPos$ consisting
of $\vee$-semilattices. Then, the dual equivalence in Proposition~\ref{prop:AlgPxPos}
restricts to a dual equivalence between
$\AlgPxSLatS$ and $\AlgLatS$. 
Hence, by Theorem \ref{thm:ContContLatKaroubi}, 
we have the following characterisation of the category of continuous
lattices and Scott continuous functions (cf.\ Table~\ref{tab:CategoriesPxPos}).
\begin{theorem}
  \label{thm:SplitAlgPxSLatS}
$\Karoubi{\AlgPxSLatS}$ is dually equivalent to $\ContLatS$. 
\end{theorem}

%
%
The objects of $\Karoubi{\AlgPxSLatS}$ can be explicitly described as follows.
\begin{definition}
A \emph{proximity $\vee$-semilattice} is a proximity poset whose
underlying poset is a $\vee$-semilattice.
We write $(S, 0, \vee, \prox)$  or simply  $S$  for a proximity
$\vee$-semilattice, where $(S, 0, \vee)$ is a $\vee$-semilattice and $(S,
\prox)$ is a proximity poset.
\end{definition}

Henceforth, we write $\PxJLat$ for $\Karoubi{\AlgPxSLatS}$, which
forms a full subcategory of $\PxPos$ consisting of 
proximity $\vee$-semilattices.

\begin{remark}
It is well known that $\ContLatS$ is equivalent to the category of
injective $\mathrm{T}_{0}$-spaces (cf.\ Scott~\cite{ScottContLatt}).
Thus, $\PxJLat$ can be seen as a predicative characterisation of the
latter category. 
\end{remark}

%
%
\subsection{Representation of continuous domains}
By Theorem~\ref{prop:SplitAlgPxPos}, each proximity poset $(S,
\prox)$ represents a continuous domain. In the case of continuous
domain, the equivalence in Theorem \ref{thm:ContContLatKaroubi} is
mediated by the canonical splitting of idempotents of $\AlgDom$ by the
image factorisation described in Proposition~\ref{prop:SplitCont}.
This suggests that the continuous domain represented by $(S,\prox)$ arises as
the image of $\Ideals{S}$ under the idempotent Scott continuous
function on $\Ideals{S}$ induced by $\prox$ (cf.\
Proposition~\ref{prop:BijApproxScottCont}).  This image, which we
denote by
$\RIdeals{S}$, is the collection of ideals of the form 
\[
  \downset_{\prox}I
  \defeql \prox I = \left\{ a \in S \mid \exists b \in I \left ( a \prox b \right)
  \right\}
\]
for some $I \in \Ideals{S}$. By the idempotency of $\prox$, these
ideals are exactly the rounded ideals of $S$ (hence the notation
$\RIdeals{S}$).
By \eqref{eq:WayBelowSplit} and \eqref{eq:WayBelowIdeals},
the way-below relation in $\RIdeals{S}$ can be characterised by
\begin{equation*}
  I \ll J \leftrightarrow \exists a \in J \left( I \subseteq
  \downset_{\prox} a \right).
\end{equation*}
Furthermore, if $(S,0,\vee, \prox)$ is a proximity $\vee$-semilattice, then 
$\RIdeals{S}$ has finite joins by
Proposition~\ref{thm:SplitAlgPxSLatS}.
By \eqref{eq:FiniteJoinsInSplit} and \eqref{eq:FiniteJoinsIdeals},
these joins can be characterised by
\begin{align}
  0 &\defeql \downset_{\prox} 0,
  &
  I \vee J &\defeql
  \bigcup_{a \in I, b \in J} \downset_{\prox}\! \left(a \vee b \right).
\end{align}

By Theorem~\ref{prop:SplitAlgPxPos}, each approximable
relation between proximity posets $S$ and $S'$ represents
a Scott continuous function between $\RIdeals{S'}$ and $\RIdeals{S}$.
This is analogous to the case of posets;
Proposition~\ref{prop:BijApproxScottCont} and
Lemma~\ref{lem:MonotoneScottCont}, which are stated in terms of posets
and their ideals, can be restated in terms of proximity posets and
their rounded ideals as follows: each approximable
relation $r \colon S \to S'$ bijectively corresponds
to a Scott continuous function $\overline{f_{r}} \colon \RIdeals{S'} \to
\RIdeals{S}$ given by 
\begin{equation}
  \label{eq:ApproxToScottPxPos}
  \overline{f_{r}}(I) \defeql r^{-}I,
\end{equation}
and this correspondence preserves identity and composition of
$\PxPos$. 
This can be established as in~Proposition~\ref{prop:BijApproxScottCont}
using the Lemma~\ref{lem:DcpoInterpretationScottCont} below in place of
Lemma~\ref{lem:MonotoneScottCont}.
Here, a \emph{dcpo interpretation} of a proximity poset $(S, \prox)$ in
a dcpo $D$ is a monotone function $f \colon S \to D$ such that $f(a) =
\bigvee_{b \prox a} f(b)$. Note that $\downset_{\prox} a$ is directed
so that $\bigvee_{b \prox a} f(b)$ exists.
\begin{lemma}
  \label{lem:DcpoInterpretationScottCont}
  For each proximity poset $(S,\prox)$, there is a dcpo interpretation
  $i_{S} \colon S \to \RIdeals{S}$ 
      defined by
      \begin{equation*}
        i_{S}(a) \defeql \downarrow_{\prox} a
      \end{equation*}
      with the following universal property:
   for any dcpo interpretation  $f \colon
   S \to D$ in a dcpo $D$, there exists a unique Scott continuous function
   $\overline{f} \colon \RIdeals{S} \to D$ such that
   $\overline{f} \circ i_{S} = f$.
\end{lemma}
\begin{proof}
  Similar to that of Lemma \ref{lem:MonotoneScottCont}.
\end{proof}
Then, each approximable relation $r \colon (S,\prox) \to (S',\prox')$
bijectively corresponds to a dcpo interpretation $f_{r} \colon S' \to
\RIdeals{S}$ given by \eqref{eq:ApproxToScott}, and hence to a Scott
continuous function $\overline{f_{r}} \colon \RIdeals{S'} \to
\RIdeals{S}$ given by \eqref{eq:ApproxToScottPxPos}.

\subsection{Strong proximity \texorpdfstring{$\vee$}{join}-semilattices}
We introduce another characterisation of a continuous lattice, which is
obtained by 
splitting idempotents of the category of algebraic lattices and
suplattice homomorphisms. This naturally leads to a predicative
characterisation of the category of continuous lattices and suplattice
homomorphisms.

First, we give a predicative characterisation of suplattice
homomorphisms between algebraic lattices.
The following is analogous to Lemma \ref{lem:MonotoneScottCont}.
\begin{lemma}
  \label{lem:JoinPreservingSupLatticeHom}
  For each $\vee$-semilattice  $(S,0, \vee)$, there is a
  join-preserving function $i_{S} \colon S \to \Ideals{S}$ 
   defined  \eqref{eq:InjectionGenerator}
      with the following universal property:
   for any join-preserving function $f \colon
   S \to L$ to a suplattice $L$, there exists a unique suplattice
   homomorphism
   $\overline{f} \colon \Ideals{S} \to L$ such that
   $\overline{f} \circ i_{S} = f$.
\end{lemma}
\begin{proof}
See Vickers \cite[Proposition 9.1.5 (ii)(iv)]{vickers1989topology}.  The
unique extension $\overline{f} \colon \Ideals{S} \to L$ is again given
by  \eqref{eq:UniquExntension}.
\end{proof}

For $\vee$-semilattices $(S,0,\vee)$ and $(S', 0', \vee')$, 
$\vee$-semilattice homomorphisms $f \colon S' \to \Ideals{S}$ correspond to
the class of approximable relations from $S$ to $S'$ as defined below.
\begin{definition}
  \label{def:JPreservingApproximableRelation}
  Let  $(S,0,\vee)$ and $(S',0',\vee')$ be $\vee$-semilattices.
  An approximable relation $r \colon S \to S'$ is 
   \emph{join-preserving} if
  \begin{enumerate}[align=left]
    \myitem[(App$0$)]\label{def:approximable0} $a \mathrel{r} 0' \imp a
    = 0$,
    \myitem[(App$\vee$)]\label{def:approximableJ} $a \mathrel{r} (b
    \vee' c) \imp
    \exists b', c' \in S   
    \left( a \leq b' \vee c' 
    \amp b' \mathrel{r} b \amp c' \mathrel{r} c \right)$.
  \end{enumerate}
  We call join-preserving approximable relations simply as
  \emph{join-approximable relations}.
\end{definition}
\begin{remark}
  The two conditions in
  Definition~\ref{def:JPreservingApproximableRelation} correspond to
  the condition (\emph{G-$\vee$}) in Jung and
  S\"underhauf~\cite[Definition~25]{JungSunderhaufDualtyCompactOpen},
  which is also called \emph{join-approximable} by van
  Gool~\cite[Definition 1.9 (4)]{vanGoolDualityCanExt}.
\end{remark}

Proposition \ref{prop:BijApproxScottCont} 
restricts to 
a bijective correspondence between join-approx\-imable
relations and suplattice homomorphisms. 
\begin{proposition}
  \label{prop:BijVeeApproxSupLatticeHom}
  Let $(S, 0, \vee)$ and $(S',0', \vee')$ be $\vee$-semilattices.
  There exists a bijective
  correspondence between join-approximable relations from $S$ to $S'$ and
  suplattice homomorphisms from $\Ideals{S'}$ to $\Ideals{S}$.
  Through this correspondence, the identity function
  on $\Ideals{S}$ corresponds to the order $\leq$ on $S$
  and the composition of suplattice homomorphisms 
  contravariantly corresponds
  to the relational composition of the join-approximable relations.
\end{proposition}
\begin{proof}
  We show that the correspondence described in the proof of
  Proposition~\ref{prop:BijApproxScottCont} restricts to
  join-approximable relations and suplattice homomorphisms.
  First, if $r \colon (S, 0, \vee) \to (S',0', \vee')$ is a
  join-approximable  relation, the function $f_{r} \colon S'
  \to \Ideals{S}$ given by \eqref{eq:ApproxToScott} preserves finite
  joins by the characterisation of finite joins of ideals in
  \eqref{eq:FiniteJoinsIdeals}.  Thus, the resulting Scott continuous
  function $\overline{f_{r}} \colon \Ideals{S'} \to \Ideals{S}$ is a
  suplattice homomorphism by
  Lemma~\ref{lem:JoinPreservingSupLatticeHom}. 
 
  Conversely, if $f \colon \Ideals{S'} \to \Ideals{S}$ is a suplattice
  homomorphism, then the approximable relation $r_{f} \colon S \to S'$
  given by  \eqref{eq:ScottToApprox} is join-preserving again by the
  characterisation \eqref{eq:FiniteJoinsIdeals}.
\end{proof}

Let $\AlgPxSLat$ be a subcategory of $\AlgPxPos$ consisting of
$\vee$-semilattices and join-approximable relations between them.
By Proposition \ref{prop:BijVeeApproxSupLatticeHom}, we have the following.
\begin{proposition}\label{prop:AlgPxSLat}
  $\AlgPxSLat$ is dually equivalent to $\AlgLat$.
\end{proposition}
Thus, by Theorem \ref{thm:ContContLatKaroubi}, we have the following
characterisation of the category of continuous lattices.
\begin{theorem}
  \label{prop:SPxJLatEquivContLat}
  $\Karoubi{\AlgPxSLat}$ is dually equivalent to
$\ContLat$.
\end{theorem}
The objects of $\Karoubi{\AlgPxSLat}$ can be explicitly described as follows.
\begin{definition}
A \emph{strong proximity $\vee$-semilattice} is a proximity
$\vee$-semilattice $(S, 0, \vee, \prox)$ satisfying
  \begin{enumerate}
    \item
    $a \prox 0 \imp a = 0$,
    \item
    $a \prox (b \vee c) 
    \imp
    \exists b',c' \in S
    \left( a \leq b' \vee c' \amp b' \prox b \amp c' \prox c \right)$.
  \end{enumerate}
\end{definition}

Henceforth, we write $\SPxJLat$ for $\Karoubi{\AlgPxSLat}$, which forms a
subcategory of $\PxJLat$ consisting of strong proximity
$\vee$-semilattices and join-approximable  relations
(in the sense of Definition \ref{def:ApproximableRelationRounded} and
Definition \ref{def:JPreservingApproximableRelation}).

Note that the objects of $\SPxJLat$
and $\PxJLat$ are essentially the same as they both represent continuous
lattices. Proposition~\ref{prop:PxJLatToSPxJLat} below makes this explicit (cf.\
Kawai~\cite[Theorem 5.8]{KawaiContEnt}). We need the following
construction for its proof.
\begin{definition}
  For a relation $r \subseteq S \times S'$, its \emph{lower extension} $r_{L}
  \subseteq  \Fin{S} \times \Fin{S'}$
  is defined by
  \[
    A \mathrel{r_{L}} B
    \defeqiv
    \forall a \in A \, \exists b \in B \left( a \mathrel{r} b \right).
  \]
\end{definition}

\begin{proposition}
  \label{prop:PxJLatToSPxJLat}
  For each proximity $\vee$-semilattice $S$, there exists a strong
  proximity $\vee$-semilattice $S'$ which is isomorphic to $S$ in
  $\PxJLat$.
\end{proposition}
\begin{proof}
  Given a proximity $\vee$-semilattice $S = (S, 0, \vee, \prox)$,
  define a preorder $\leq^{\vee}$ on $\Fin{S}$
  by 
  \[
    A \leq^{\vee} B \defeqiv \forall C \prox_{L} A \, \exists D
    \prox_{L} B \left( \medvee C \prox \medvee D \right).
  \]
  Let $S' = (\Fin{S}, \leq^{\vee})$ be the poset reflection of
  $\leq^{\vee}$, and define a $\vee$-semilattice structure $(S', 0^{\vee},
  \vee^{\vee})$ on $S'$ by 
  \begin{align*}
    0^{\vee} &\defeql \emptyset, &
   A  \vee^{\vee} B &\defeql A \cup B.
  \end{align*}
  It is easy to see that the above operations are well-defined with
  respect to the
  equality on $S'$. Moreover, define a relation $\prox^{\vee}$ on
  $S'$ by
  \[
    A \prox^{\vee} B  \defeqiv \exists C \prox_L B \left( A
    \leq^{\vee} C \right). 
  \]
  Then, $S' = (S', 0^{\vee},
  \vee^{\vee}, \prox^{\vee})$ is a strong proximity
  $\vee$-semilattice. Indeed, since $\prox_{L}$ is idempotent and ${\prox_{L}}
  \subseteq {\prox^{\vee}} \subseteq {\leq^{\vee}}$, 
  the relation $\prox^{\vee}$ is idempotent, 
  which clearly satisfies \ref{def:approximableI}. As for the property
  \ref{def:approximableU}, suppose $A \prox^{\vee} B$ and 
  $B \leq^{\vee} C$. By the former, there exists  $B' \prox_{L} B$ such that $A
  \leq^{\vee} B'$, so by the latter, there exists  $C' \prox_{L} C$ such
  that $\bigvee B' \prox \bigvee C'$. Thus, there exists $D$ such that 
  $C' \prox_{L} D \prox_{L} C$. Then, one can easily see that $A
  \leq^{\vee} D$, so that $A \prox^{\vee} C$.
  To see that $S'$ is strong, we have $A \prox^{\vee} \emptyset$ 
  implies $A \leq^{\vee} \emptyset$, so $S'$ satisfies
  \ref{def:approximable0}. Next, suppose that $A \prox^{\vee} B
  \vee^{\vee} C$. Then, there exists $D \prox_{L} B \cup C$
  such that $A \leq^{\vee} D$. Split $D$ into $D_{B}$ and 
  $D_{C}$ such that $D = D_{B} \cup D_{C}$, $D_{B} \prox_{L} B$, and
  $D_{C} \prox_{L} C$. Then, $A \leq^{\vee} D_{B} \vee^{\vee} D_{C}$,
  $D_{B} \prox^{\vee} B$, and $D_{C} \prox^{\vee} C$. Hence $S'$
  satisfies~\ref{def:approximableJ}.
  
  Now, define relations $r \subseteq \Fin{S} \times S$
  and $s \subseteq S \times \Fin{S}$ by 
  \begin{align*}
    A \mathrel{r} a &\defeqiv A \prox^{\vee} \left\{ a \right\}, &
    a \mathrel{s} A \defeqiv \exists B \prox_{L} A \left( a \prox
    \medvee B \right).
  \end{align*}
  We show that these are approximable relations between $S'$ and $S$
  which are inverse to each other.  The property
  \ref{def:approximableI} for $s$ follows from the fact that $0 \prox
  0$ and $a \prox b \amp a' \prox b' \imp a \vee a' \prox b \vee b'$. 
  From the definition of $\leq^{\vee}$, it is
  easy to see that $s$ also satisfies \ref{def:approximableU}.
  Moreover, we have ${\mathrel{s} \circ \prox} = s$ by the idempotency
  of $\prox$, and ${s} \subseteq {\prox_{L} \circ \mathrel{s}}
  \subseteq {\prox^{\vee} \circ \mathrel{s}} \subseteq {\leq^{\vee}
  \circ \mathrel{s}} = {\mathrel{s}}$. Thus, $s$ is an approximable
  relation from $S$ and $S'$. The fact that $r$ is approximable is also
  easy to check.
  To see that $r$ and $s$ are inverse to each other,
  first, we have  ${s \circ r} \subseteq  {\prox^{\vee}}$, noting that $a
  \mathrel{s} A$ implies $\left\{ a \right\} \leq^{\vee} A$.
  Conversely, if $A \prox^{\vee} B$, then there exists $C \prox_{L} B$
  such that $A \prox^{\vee} C$. Then, $A \mathrel{r} \bigvee C
   \mathrel{s} B$, and so $A \mathrel{(s \circ r)} B$.
  Thus, ${s \circ r}  =  {\prox^{\vee}}$.
  Next, we have ${\prox} \subseteq  {r \circ s}$ by the idempotency of
  $\prox$. Conversely, if $a \mathrel{s} A \mathrel{r} b$, 
  there exists $B \prox_L A$ such that $a \prox \bigvee B$. Since
  $A \mathrel{r} b$ implies $A \leq^{\vee} \left\{ b \right\}$, there exists $C
  \prox_{L} \left\{ b \right\}$ such that $\bigvee B \prox \bigvee C$.
  Then, $a \prox
    b$, and thus ${r \circ s} = {\prox}$.
\end{proof}

\subsection{Algebraic theory of continuous lattices}
\label{sec:LowerPowerLoc}
We give yet another predicative characterisation of continuous
lattices in terms of algebras of the lower powerlocale on the category
of continuous domains. Specifically, from the localic point of view,
the category of continuous domains and continuous functions can be
regarded as a subcategory of locales whose
frames are Scott topologies of continuous domains.
This is the view of the category $\Infosys$ of infosys~\cite{Infosys},
a classical dual of $\PxPos$.  In this view, the construction in
Definition \ref{def:LowerPower} below corresponds to that of the lower
powerlocale on the Scott topologies of continuous
domains~\cite[Theorem~4.3~(ii)]{Infosys}. In the case of infosys, the
construction forms a monad on $\Infosys$, and Vickers conjectured that
the category of algebras for the monad is equivalent to
$\ContLat$~\cite[Section~5.1]{Infosys}.  In classical domain
theory, the corresponding characterisation of continuous lattices has
been given by Schalk~\cite[Section~6.2.2]{schalk1994algebras}.

In what follows, we work in the dual context of $\PxPos$, where the
construction in Definition~\ref{def:LowerPower} (also called the lower
powerlocale in this paper) forms a comonad on $\PxPos$. Then, the main result of this
subsection says that the category of strong proximity
$\vee$-semilattices is equivalent to that of the coalgebras of the
lower powerlocale.

\begin{defiC}[{\cite[Definition~4.1]{Infosys}}]
  \label{def:LowerPower}
  Let $(S,\prox)$ be a proximity poset. The \emph{lower powerlocale}
  $\Lower{S}$
  of $(S,\prox)$ is a proximity poset $((\Fin{S}, \leq_{L}), \prox_{L})$ where
  $(\Fin{S}, \leq_{L})$ denotes the poset reflection of the preorder
  $\leq_{L}$.
\end{defiC}
  
\begin{proposition}\label{prop:LowerIsUniversal}
  Let $(S,\prox)$ be a proximity poset.
  \begin{enumerate}
    \item\label{prop:LowerIsUniversal1} $\Lower{S}$ is a strong
      proximity $\vee$-semilattice.
    \item\label{prop:LowerIsUniversal2} 
      A relation $\varepsilon^{L}_{S} \subseteq \Fin{S} \times S$
      defined by
      \begin{equation} \label{eq:Lowercounit}
        A \mathrel{\varepsilon^{L}_{S}} a \defeqiv A \prox_{L} \left\{
        a \right\}
      \end{equation}
      is an approximable relation from $\Lower{S}$ to $(S,\prox)$ with
      the following universal property: for any strong proximity
      $\vee$-semilattice
      $(S',0',\vee', \prox')$ and an approximable relation $r \colon
      S' \to S$, there exists a unique join-approximable 
      relation $\overline{r} \colon S' \to \Lower{S}$ such that
      $\varepsilon^{L}_{S} \circ \overline{r} = r$.
  \end{enumerate}
\end{proposition}
\begin{proof}
  \ref{prop:LowerIsUniversal1}. The poset $(\Fin{S}, \leq_{L})$ has
  finite joins characterised by
  \begin{align*}
    0 &\defeql \emptyset,
    &
    A \vee B &\defeql A \cup B.
  \end{align*}
  It is straightforward to show that $\prox_{L}$ is an idempotent approximable
  relation on ${(\Fin{S}, \leq_{L})}$ which satisfies \ref{def:approximable0}
  and \ref{def:approximableJ} with respect to  the above operations.
  \smallskip
   
  \noindent\ref{prop:LowerIsUniversal2}.
  We clearly have
  ${\mathrel{\varepsilon^{L}_{S}} \circ \prox_{L}} 
  = {\varepsilon^{L}_{S}}$. 
  By \ref{def:approximableI} for $\prox$, we also have
  ${\prox \circ \mathrel{\varepsilon^{L}_{S}}}
  = {\varepsilon^{L}_{S}}$.
  Since  $(\Fin{S}, \leq_{L})$ has finite joins computed by
  unions, $\varepsilon^{L}_{S}$ satisfies \ref{def:approximableI}. 
  Moreover, the property \ref{def:approximableU}
  follows from the corresponding property of $\prox$.
  Thus, $\varepsilon^{L}_{S}$ is an approximable relation.
  Next, given an
  approximable relation $r \colon S' \to S$ from a strong proximity
  $\vee$-semilattice
  $(S',0',\vee',\prox')$, define a relation $\overline{r}
  \subseteq S' \times \Fin{S}$ by
  \[
    b \mathrel{\overline{r}} A \defeqiv \exists B \in \Fin{S'}
    \left( b \leq' \medvee B \amp B \mathrel{r_{L}} A \right).
  \]
  It is straightforward to check that $\overline{r}$ is a
  join-approximable  relation from $S'$ to $\Lower{S}$ such
  that
  $\varepsilon^{L}_{S} \circ \overline{r} = r$. For the uniqueness of
  $\overline{r}$,
  suppose that $u \colon S' \to \Lower{S}$ is a join-approximable
  relation such that ${\mathrel{\varepsilon^{L}_{S}} \circ \mathrel{u}} = r$.
  Suppose $b \mathrel{u} A$. Since ${\prox_{L} \circ \mathrel{u}} = u$, there
  exists $A' \in \Fin{S}$ such that $b \mathrel{u} A' \prox_{L}A$.
  Since $A' = \bigcup \left\{ \left\{ a' \right\} \mid a' \in A'
\right\}$  and $u$ is join-preserving, there exists $B \in
  \Fin{S'}$ such that $b \leq' \medvee B$ and $B \mathrel{u_{L}}
  \left\{ \left\{ a' \right\} \mid a' \in A' \right\}$. Then, $B
  \mathrel{r_{L}}A$ by ${\mathrel{\varepsilon^{L}_{S}} \circ
  \mathrel{u}} \subseteq r$. Hence $b
  \mathrel{\overline{r}} A$.  Conversely, suppose $b
  \mathrel{\overline{r}} A$. Then, there exists $B \in \Fin{S'}$ such
  that $b \leq' \medvee B$ and $B \mathrel{r_{L}} A$. Thus, for each
  $b' \in B$, there exists $a \in A$ such that $b'
  \mathrel{(\varepsilon^{L}_{S} \circ u)} a$.
  Then, $b' \mathrel{u}
  \left\{ a \right\}$, and since $\left\{ a \right\} \leq_{L} A$, we
  have $b' \mathrel{u} A$ by \ref{def:approximableU}. Hence,
  $b \mathrel{u} A$ by \ref{def:approximableI}.
\end{proof}
By Proposition \ref{prop:LowerIsUniversal}, the construction $\Lower{S}$ determines a right adjoint to the forgetful functor
from $\SPxJLat$ to $\PxPos$.
The functor
$\LowerFunc \colon \PxPos \to \SPxJLat$ sends an approximable relation
$r \colon S \to S'$ to 
a join-approximable  relation $\Lower{r} \colon \Lower{S}
\to \Lower{S'}$ defined by
\[
  \Lower{r} \defeql r_{L}.
\]
The counit $\varepsilon^{L}_{S} \colon \Lower{S} \to S$ of the adjunction
is given by \eqref{eq:Lowercounit},
and the unit $\eta^{L}_{S} \colon S \to \Lower{S}$ at a strong proximity
$\vee$-semilattice $(S,0,\vee, \prox)$ is given by 
\begin{equation} \label{eq:Lowerunit}
  a \mathrel{\eta^{L}_{S}} A \defeqiv a \prox \medvee A.
\end{equation}
The adjunction induces a comonad $\langle \LowerFunc, \varepsilon^{L},
\nu^{L} \rangle$ on $\PxPos$ with a co-multiplication
$\nu^{L}_{S} \defeql
\eta_{\Lower{S}}^{L}$, which by definition satisfies
\[
  A \mathrel{\nu^{L}_{S}} \mathcal{U}
  \leftrightarrow
  A \prox_{L} \bigcup \mathcal{U}.
\]
Let $\coAlg{\LowerFunc}$ denote the category of $\LowerFunc$-coalgebras
and coalgebra homomorphisms, and let  $K \colon \SPxJLat \to
\coAlg{\LowerFunc}$ be the comparison functor. Note that $K$ sends each strong
proximity
$\vee$-semilattice $S$ to a $\LowerFunc$-coalgebra
$\eta^{L}_{S} \colon S \to \Lower{S}$ and it is an identity map on
morphisms (cf.\ \cite[Chapter VI, Section 3]{Category_at_work}).
\begin{lemma}\label{lem:ProxSupLatLowerCoalg}
 The functor $K \colon \SPxJLat \to \coAlg{\LowerFunc}$ 
 is full and faithful.
\end{lemma}
\begin{proof}
  $K$ is obviously faithful. To see that $K$ is full,
  let $r \colon S \to S'$ be a $\LowerFunc$-coalgebra homomorphism
  between strong proximity $\vee$-semilattices $(S, 0,\vee, \prox)$ and $(S', 0',
  \vee', \prox')$.  Then, for any $a \in S$ and $B \in \Fin{S'}$, we
  have
  \[
    \begin{aligned}
      a \mathrel{r} {\medvee}' B
      &\iff \exists b \in S' \bigl( 
      a \mathrel{r} b \prox' {\medvee}' B\bigr)
      && (\text{by } {\mathrel{r}} = {\prox' \circ \mathrel{r}}) \\
      &\iff a \mathrel{(\eta^{L}_{S'} \circ r)} B \\
      &\iff a \mathrel{(\Lower{r} \circ \eta^{L}_{S})} B
      && (\text{$r$ is a homomorphism}) \\
      &\iff \exists A \in \Fin{S}
      \left( a \prox \medvee A \amp A \mathrel{r_{L}} B \right) \\
      &\iff \exists A \in \Fin{S}
      \left( a \leq \medvee A \amp A \mathrel{(\mathrel{r_{L}} \circ
      \prox_{L})} B \right)
      && (\text{$S$ is strong}) \\
      &\iff \exists A \in \Fin{S}
      \left( a \leq \medvee A \amp A \mathrel{r_{L}} B \right).
      && (\text{by }{\mathrel{r}} = {\mathrel{r}  \circ  \prox})
    \end{aligned}
  \]
  Thus, $r$ is join-preserving. Hence $K$ is full.
\end{proof}
It remains to show that 
$K$ is essentially surjective. This will become clear after 
we recall the fact  that the comonad $\LowerFunc$ is of Kock--Z\"oberlein type
(see Escard\'o \cite[Section 4.1]{escardo1998properly}). The
dual notion is that of $\KZ$-monad \cite{KOCKKZmonad}.
\begin{definition} 
  Let $\langle T, \varepsilon, \nu \rangle$ be a comonad on a poset enriched
  category  $\mathbb{C}$, where $T$ preserves the order on
  homsets. Then, $T$ is called a \emph{$\KZ$-comonad (co$\KZ$-comonad)} if
  $\varepsilon_{TX} \leq T \varepsilon_{X}$ (resp.\
  $T\varepsilon_{X} \leq \varepsilon_{TX}$) for each object $X$ of
  $\mathbb{C}$.
\end{definition}
\begin{proposition}\label{prop:KZ}
  Let $\langle T, \varepsilon, \nu \rangle$ be a comonad on a poset enriched
  category  $\mathbb{C}$, where $T$ preserves the order on
  homsets. Then, the following are equivalent:
  \begin{enumerate}
    \item\label{prop:KZ1} $T$ is a $\KZ$-comonad.
    \item\label{prop:KZ2} $\alpha \colon X \to TX$ is a $T$-coalgebra
      if and only if $\alpha \dashv \varepsilon_{X}$ and 
      $\varepsilon_{X} \circ \alpha = \id_{X}$,
      where
      \[
      \alpha \dashv \varepsilon_{X} \defeqiv
      \id_{X} \leq \varepsilon_{X} \circ \alpha 
      \amp
      \alpha \circ \varepsilon_{X} \leq \id_{TX}.\footnote{As usual
      in the context of poset enriched categories, the situation 
      $\alpha \dashv \varepsilon_{X}$
      is
      called an \emph{adjunction}, where  $\alpha$ and
      $\varepsilon_{X}$ are the \emph{left adjoint} and
      the \emph{right adjoint}, respectively.}
      \]
  \end{enumerate}
\end{proposition}
\begin{proof}
  See Escard\'o \cite[Lemma 4.1.1]{escardo1998properly} for a proof
  for $\KZ$-monads.
\end{proof}
We have a similar characterisation for co$\KZ$-comonads, which is
obtained by replacing
item \ref{prop:KZ2} of Proposition~\ref{prop:KZ}  with the following:
\begin{equation}
  \label{eq:coKZcomonad}
  \text{$\alpha \colon X \to TX$ is a $T$-coalgebra
  if and only if $\varepsilon_{X} \dashv \alpha$ and
  $\varepsilon_{X} \circ \alpha = \id_{X}$.}
\end{equation}
By the uniqueness of left and right adjoints, for a (co)$\KZ$-comonad, each object
admits at most one coalgebra structure.  The following
characterisation of coalgebras and isomorphisms between them is useful.
\begin{corollaryC}[{\cite[Corollary 4.2.3]{escardo1998properly}}]
\label{cor:KZAlg}
  Let $\langle T, \varepsilon, \nu \rangle$ be a (co)$\KZ$-comonad on
  a poset enriched category $\mathbb{C}$.
  \begin{enumerate}
    \item \label{cor:KZAlg1}
      The following are equivalent for an object $X$:
      \begin{enumerate}
        \item\label{cor:KZAlg1a} $X$ admits a $T$-coalgebra structure.
        \item\label{cor:KZAlg1b} $X$ is a retract of $TX$.
        \item\label{cor:KZAlg1c} There exists  $\alpha \colon X \to TX$ such 
          that $\varepsilon_{X} \circ \alpha = \id_{X}$.
      \end{enumerate}

    \item \label{cor:KZAlg2} Any isomorphism in $\mathbb{C}$ 
      between the underlying objects of $T$-coalgebras is an
      isomorphism of the coalgebras.
  \end{enumerate}
\end{corollaryC}
\begin{proof}
  \noindent \ref{cor:KZAlg1}. 
  See Escard\'o~\cite[Corollary~4.2.3]{escardo1998properly}.
  \smallskip

  \noindent \ref{cor:KZAlg2}. 
  We only give a proof for $\KZ$-comonads.
  Let $\alpha \colon X \to TX$ and 
    $\beta\colon Y \to TY$ be $T$-coalgebras, and let $f \colon X \to Y$
    be an isomorphism in $\mathbb{C}$ with an inverse $g \colon Y \to
    X$.  It suffices to show that $f$ is a coalgebra homomorphism.
    Since
    \[
      \varepsilon_{X} \circ Tg \circ \beta \circ f  = g \circ
      \varepsilon_{Y} \circ \beta \circ f = g \circ f = \id_{X},
    \]
    $Tg \circ \beta \circ f$ is a $T$-coalgebra on $X$ by \ref{cor:KZAlg1c}.
    Since a coalgebra structure on $X$ is unique, we must have 
    $Tg \circ \beta \circ f
    = \alpha$. Thus, $\beta \circ f = Tf \circ \alpha$.
\end{proof}

Each homset of $\PxPos$ is ordered by the inclusion of graphs of
approximable relations, and the functor $\LowerFunc$ clearly preserves
this order. 
\begin{proposition}
  \label{prop:LowerIscoKZ}
The comonad $\langle \LowerFunc, \varepsilon^{L}, \nu^{L} \rangle$ is a
co$\KZ$-comonad on $\PxPos$.
\end{proposition}
\begin{proof}
  Let $(S, \prox)$ be a proximity poset. For each $\mathcal{U} \in
  \FFin{S}$ and $A \in \Fin{S}$, we have
  \begin{align*}
    \mathcal{U} \mathrel{\Lower{\varepsilon^{L}_{S}}}  A 
    &\iff
    \mathcal{U} \mathrel{(\prox_{L})_{L}}  \left\{ \left\{ a
    \right\} \mid a \in A \right\}\\
    &\;\,\Longrightarrow
    \mathcal{U} \mathrel{(\prox_{L})_{L}}  \left\{ A \right\}\\
    &\iff
    \mathcal{U} \mathrel{\varepsilon^{L}_{\Lower{S}}}  A.
  \end{align*}
  Thus  $\Lower{\varepsilon^{L}_{S}} \leq
  \varepsilon^{L}_{\Lower{S}}$.
\end{proof}

By Proposition \ref{prop:LowerIscoKZ} and Corollary
\ref{cor:KZAlg}~\eqref{cor:KZAlg1}, $\LowerFunc$-coalgebras are
precisely the retracts of free $\LowerFunc$-coalgebras in the category
of proximity posets. In view of Proposition~\ref{prop:LowerIsUniversal}
\eqref{prop:LowerIsUniversal1}, Theorem~\ref{prop:SPxJLatEquivContLat}, and
Proposition~\ref{prop:SplitCont}~\eqref{prop:SplitCont2}, each
$\LowerFunc$-coalgebra represents a continuous lattice.
Hence, the functor $K \colon \SPxJLat \to \coAlg{\LowerFunc}$ must be
essentially surjective, which we now make explicit.
\begin{lemma}
  \label{lem:LowerCoalgEssSurj}
  For each $\LowerFunc$-coalgebra $\alpha \colon S \to \Lower{S}$,
  there exists a strong proximity $\vee$-semilattice $S'$ which is isomorphic
  to $S$ in $\PxPos$.
\end{lemma}
\begin{proof}
  Let $\alpha \colon S \to \Lower{S}$ be a $\LowerFunc$-coalgebra  
  on a proximity poset $(S, \prox)$. Let $\prox_{\alpha} \defeql
  \alpha \circ \varepsilon^{L}_{S}$, and put $S' = (\Fin{S},
  \prox_{\alpha})$ where $\Fin{S}$ is regarded as a free
  $\vee$-semilattice over $S$. Then, it is easy to see that
  $S'$ is a proximity $\vee$-semilattice. Moreover, 
  $\varepsilon^{L}_{X}$ and $\alpha$ are approximable relations from 
  $S'$ to $S$ and $S$ to $S'$, respectively, and they are inverse to each
  other. On the other hand, by Proposition~\ref{prop:PxJLatToSPxJLat}, there
  is a strong proximity $\vee$-semilattice $S''$ which is isomorphic
  to $S'$ in $\PxJLat$. Then, $S''$ is isomorphic to $S$ in $\PxPos$.
\end{proof}
By Lemma \ref{lem:LowerCoalgEssSurj} and
Corollary~\ref{cor:KZAlg}~\eqref{cor:KZAlg2}, the comparison functor
$K \colon \SPxJLat \to \coAlg{\LowerFunc}$ is essentially surjective.
Since $K$ is full and faithful (cf.\
Lemma~\ref{lem:ProxSupLatLowerCoalg}), $K$ determines an equivalence
of the categories.%
\footnote{This equivalence does not require
  the axiom of choice because an explicit description of a
  quasi-inverse of $K$ can be obtained from the proof
  of Lemma \ref{lem:LowerCoalgEssSurj}.}

\begin{theorem}
  \label{thm:LowerCoAlgEquivPxSupLat}
    $\SPxJLat$ is equivalent to $\coAlg{\LowerFunc}$.
\end{theorem}

By Theorem~\ref{prop:SPxJLatEquivContLat} and
Theorem~\ref{thm:LowerCoAlgEquivPxSupLat},  
we have the following characterisation of continuous lattices.
\begin{theorem}
  \label{thm:LowerEquivContLat}
  $\coAlg{\LowerFunc}$ is dually equivalent to $\ContLat$.
\end{theorem}

\subsection{Localized strong proximity 
  \texorpdfstring{$\vee$}{join}-semilattices}\label{sec:LocalizedPxSL}
We characterise locally compact locales in terms of strong proximity
$\vee$-semi\-lattices.  Recall that a \emph{frame} is a poset $(X,
\wedge, \bigvee)$ with finite meets $\wedge$ and joins $\bigvee$ for
all subsets of $X$ where finite meets distribute over all joins. A
homomorphism between frames $X$ and $Y$ is a function $f \colon X \to
Y$ which preserves finite meets and all joins.  The \emph{category of
locales} is the opposite of the category of frames and frame
homomorphisms. A  locale is \emph{locally compact} if it is a
continuous lattice (see Johnstone~\cite[Chapter VII, Section
4]{johnstone-82}).

The following structure characterises locally compact locales.
\begin{definition}
  A strong proximity $\vee$-semilattice $(S, 0, \vee, \prox)$ is  
   \emph{localized} if 
  \begin{equation}
    \label{eq:Localized}
    a \prox b \leq c \vee d \imp \exists a_{1} \in \left(b
    \downarrow_{\prox} c \right)
     \exists a_{2} \in \left( b \downarrow_{\prox} d \right) 
     \left( a \prox a_{1} \vee a_{2} \right),
  \end{equation}
  where $b \downarrow_{\prox} c 
  \defeql \downset_{\prox}b \cap \downset_{\prox}c$.
\end{definition}
Since $S$ is assumed to be strong, the antecedent of \eqref{eq:Localized}
can be equivalently stated as $a \prox b \prox c \vee d$. 

\begin{lemma}
  \label{lem:EquivalenceLocalization}
  For strong proximity $\vee$-semilattices,
  the condition \eqref{eq:Localized} is equivalent to
  \begin{equation}
    \label{eq:LocalizedGeneral}
    a \prox a'\leq \medvee A \amp a' \leq \medvee B 
    \imp \exists C \in \Fin{A \downarrow_{\prox} B} \left( a \prox
    \medvee C \right),
  \end{equation}
  where $A \downarrow_{\prox} B \defeql \downset_{\prox} A \cap
  \downset_{\prox} B$.
\end{lemma}
\begin{proof}
  First, \eqref{eq:Localized} implies
  \begin{equation}
    \label{lem:EquivalenceLocalization1}
    a \prox a' \leq \medvee A \imp \exists C \in 
    \Fin{a' \downarrow_{\prox} A} \left( a \prox \medvee C \right),
  \end{equation}
  where $a' \downarrow_{\prox} A \defeql \{a'\} \downarrow_{\prox} A$.
  This can be proved
  by induction on the size of $A$. Note that the antecedent
  of \eqref{lem:EquivalenceLocalization1} can be equivalently stated
  as $a \prox a' \prox \medvee A $.
  Now, suppose
  $a \prox a' \leq \medvee A$ and $a' \leq \medvee B$.
  By \eqref{lem:EquivalenceLocalization1},
  there exists $C \in \Fin{a' \downarrow_{\prox} A}
  \subseteq  \Fin{\medvee B \downarrow_{\prox} A}$
  such that $a \prox \medvee C$,
  and since $S$ is strong, there exists $D \prox_{L} C$ such that
  $a \leq \bigvee D$.
%
%
  For each  $d \in D$, there exist $c \in \downset_{\prox} A$ and
  $E_{d} \in \Fin{c \downarrow_{\prox} B}$ such that $d \prox
  \bigvee E_{d}$ by
  \eqref{lem:EquivalenceLocalization1}.
  Thus, $E_{d} \in \Fin{A \downarrow_{\prox} B}$ for each $d \in D$.
  Then, by putting $E = \bigcup_{d \in D} E_{d}$, we have
  $E \in \Fin{A \downarrow_{\prox} B}$
  and $\medvee D \prox \medvee E$, and so $a \prox \medvee E$.

%
%
  Conversely, assume that \eqref{eq:LocalizedGeneral} holds,
  and let $a \prox b \leq c \vee d$. By letting
  $a' = b$, $A = \left\{ b \right\}$, and
  $B = \left\{ c,d \right\}$ in \eqref{eq:LocalizedGeneral},
  we find $C \in \Fin{b \downarrow_{\prox} \left\{ c,d \right\}}$ 
  such that $a \prox \bigvee C$.
  Split $C$ into $C_{c}$ and $C_{d}$ such that $C = C_{c} \cup C_{d}$,
  $C_{c} \in \Fin{b \downarrow_{\prox} c}$, and 
  $C_{d} \in \Fin{b \downarrow_{\prox} d}$. By putting $a_{1} = \bigvee
  C_{c}$ and $a_{2} = \bigvee C_{d}$, we have  $a \prox a_{1} \vee a_{2}$,
  $a_{1} \in (b \downarrow_{\prox} c)$, and 
  $a_{2} \in (b \downarrow_{\prox} d)$.
\end{proof}

As in the case of \eqref{eq:Localized}, the antecedent of
\eqref{eq:LocalizedGeneral} can be equivalently stated as $a \prox a' \prox
\bigvee A$ and $a' \prox \bigvee B$. 
%
%
By simple induction, one can show that \eqref{eq:LocalizedGeneral}
is further equivalent to 
  \begin{equation}
    \label{eq:LocalizedGeneralFinite}
    a \prox a' \amp \forall i < n \left( a' \leq \medvee A_{i} \right) 
    \imp \exists C \in \Fin{\bigcap_{i < n} \downset_{\prox} A_{i}}
    \left( a \prox \medvee C \right)
  \end{equation}
  for finitely many $A_{0},\dots,A_{n-1} \in \Fin{S}$.\footnote{In the case
  $n = 0$, we assume $\bigcap_{i < n} \downset_{\prox} A_{i} = S$.}

The following proposition says that localized strong proximity
$\vee$-semilattices capture the notion of locally compact locale.
\begin{proposition}
  \label{prop:LocalizedFrame}
  A strong proximity $\vee$-semilattice $S$ is localized if and
  only if the collection $\RIdeals{S}$ of rounded ideals of $S$ is a
  frame.
\end{proposition}
\begin{proof}
Let $(S, 0, \vee, \prox)$ be a strong proximity $\vee$-semilattice. Since $\RIdeals{S}$
has all joins, $\RIdeals{S}$ has finite meets
characterised by
\begin{equation}\label{eq:FiniteMeets}
  1 \defeql \downarrow_{\prox} S,
  \qquad \qquad
  I \wedge J \defeql \downarrow_{\prox} \left(  I \cap J\right).
\end{equation}
%
%
Since $\RIdeals{S}$ is a continuous lattice, finite meets distribute
over directed joins (cf.\ Johnstone~\cite[Chapter~VII,
Lemma~4.1]{johnstone-82}). Thus, it suffices to show that the
condition~\eqref{eq:Localized} is equivalent to
the distributivity of finite meets over finite joins, i.e.,
\begin{equation}\label{eq:Distributive}
  I \wedge (J \vee K)  = (I \wedge J) \vee (I \wedge K) 
\end{equation}
for all $I,J,K \in \RIdeals{S}$. 

First, suppose that $S$ is localized.
Let $a \in I \wedge (J \vee K)$. 
%
%
Since $S$ is strong and $J$ and $K$ are rounded, we have 
$J \vee K 
= \bigcup_{c \in J, d \in K} \downset_{\prox} (c \vee d) 
= \bigcup_{c \in J, d \in K} \downset_{\leq} (c \vee d)$,
and so $I \wedge (J \vee K)= \downset_{\prox} (I \cap \bigcup_{c \in J,
d \in K} \downset_{\leq} (c \vee d))$. 
Thus, there exist $b \in I$,
$c \in J$, and $d \in K$ such that $a \prox b \leq c \vee d$.
Then, by \eqref{eq:Localized}, there exist $a_{1} \in b \downarrow_{\prox} c$
and $a_{2} \in b \downarrow_{\prox} d$ such that $a \prox a_{1} \vee
a_{2}$, and since $S$ is strong, there exist $a_{1}'$ and $a_{2}'$ such that
$a_{1}' \prox a_{1}$, $a_{2}' \prox a_{2}$ and $a \leq a_{1}' \vee
a_{2}'$. Then $a \in (I \wedge J) \vee (I \wedge K)$. 
 
Conversely, assume that \eqref{eq:Distributive}
holds for all $I,J,K \in \RIdeals{S}$, and let  $a \prox b \leq c \vee d$.
Choose $a' \in S$ such that $a \prox a' \prox b$. Then, 
$a' \in {\downarrow_{\prox} b} \cap \left({\downarrow_{\prox}c}
\vee {\downarrow_{\prox}d}  \right)$, and so
$a \in {\downarrow_{\prox} b} \wedge \left({\downarrow_{\prox}c}
\vee {\downarrow_{\prox}d}  \right)$. Thus
$a \in \left( {\downarrow_{\prox}b} \wedge
{\downarrow_{\prox}c} \right) \vee \left(
{\downarrow_{\prox}b} \wedge {\downarrow_{\prox}d} \right)$
by~\eqref{eq:Distributive}. Then, there exist
$a_{1} \in b \downarrow_{\prox} c$ and $a_{2} \in b \downarrow_{\prox}
d$ such that $a \prox a_{1} \vee a_{2}$. 
\end{proof}


%
%
A continuous lattice has finite meets, so it is a continuous meet
semilattice, i.e., continuous domain which has finite meets.
In classical domain theory, the category of continuous meet
semilattices and Scott continuous meet semilattice homomorphisms is
equivalent to the category of algebras of the upper powerdomains of
continuous domains~(cf.\ Schalk~\cite[Section~7.2.5]{schalk1994algebras}).
In the point-free setting, Vickers~\cite[Section 5.1]{Infosys}
conjectured that this domain theoretic characterisation should hold
for the upper powerlocale on the category of infosys. In what follows,
we confirm his conjecture in the dual context of $\PxPos$ by showing
that the category of coalgebras of the upper powerlocale on $\PxPos$
is dually equivalent to that of continuous meet semilattices and Scott
continuous meet semilattice homomorphisms.

%
%
In the context of infosys~\cite{Infosys}, the construction of the
upper powerlocale is given as in Definition~\ref{def:UpperPower} below,
which corresponds to the upper powerlocale on the Scott
topologies of continuous domains~\cite[Theorem~4.3~(iii)]{Infosys}.
\begin{definition}
  \label{def:UpperExt}
  For a relation $r \subseteq S \times S'$, 
   its \emph{upper extension}
   $r_{U} \subseteq  \Fin{S} \times \Fin{S'}$ 
  is defined by
  \[
    A \mathrel{r_{U}} B
    \defeqiv
    \forall b \in B \, \exists a \in A \left(  a \mathrel{r} b \right).
  \]
\end{definition}
\begin{defiC}[{\cite[Definition~4.1]{Infosys}}]
  \label{def:UpperPower}
  Let $(S,\prox)$ be a proximity poset. The \emph{upper powerlocale}
  $\Upper{S}$ of $(S,\prox)$ is a proximity poset $((\Fin{S},
  \leq_{U}), \prox_{U})$ where $(\Fin{S}, \leq_{U})$ denotes the poset reflection of the preorder
  $\leq_{U}$.
\end{defiC}
Note that $\Upper{S}$ is indeed a proximity poset: the only
non-trivial property to be checked is that $\downset_{\prox_{U}} A$ is
directed for each $A \in \Fin{S}$.  To see this, let $B,C \in
\downset_{\prox_{U}} A$.  For each $a \in A$, there exist $b_{a} \in
B$ and $c_{a} \in C$ such that $b_{a} \prox a$ and $c_{a} \prox a$.
Thus, there exists $d_{a} \prox a$ such that $b_{a} \leq d_{a}$ and
$c_{a} \leq d_{a}$. Put $D = \left\{ d_{a} \mid a \in A \right\}$.
Then, $D \prox_{U} A$, $B \leq_{U} D$, and  $C \leq_{U} D$.
Similarly, one can show that $\downset_{\prox_{U}} A$ is inhabited.

The construction $\Upper{S}$ gives rise to a functor $\UpperFunc
\colon \PxPos \to \PxPos$,
which is defined on morphisms as follows:
\[
  \Upper{r} \defeql r_{U}.
\]
%
%
As in the previous paragraph, one can show that  $\UpperFunc$ is
well-defined on morphisms, i.e., if $r \colon (S, \prox) \to (S', \prox')$
is an approximable relation, then $r_{U}$ is an approximable relation
from $\Upper{S}$ to $\Upper{S'}$.
 
There are approximable relations $\varepsilon^{U}_{S} \colon
\Upper{S} \to S$ and $\nu^{U}_{S} \colon
\Upper{S} \to \Upper{\Upper{S}}$ defined by
  \begin{align*}
    A \mathrel{\varepsilon^{U}_{S}} a
    &\defeqiv A \prox_{U} \left\{ a \right\}, \\
    A \mathrel{\nu^{U}_{S}} \mathcal{U}
    &\defeqiv  A \prox_{U} \bigcup \mathcal{U}.
  \end{align*}
It is routine to show that $\langle \UpperFunc, \varepsilon^{U},
\nu^{U} \rangle$ is a $\KZ$-comonad on $\PxPos$. 

\begin{proposition}
  \label{prop:AlgUpper}
  A proximity poset $(S,\prox)$ is a $\UpperFunc$-coalgebra if and only if
  $\RIdeals{S}$ has finite meets.
\end{proposition}
\begin{proof}
  Suppose that $S$ has a $\UpperFunc$-coalgebra
  structure $\alpha \colon S \to \Upper{S}$.
  Put $\top = \alpha^{-}\emptyset$. Since $\varepsilon^{U}_{S} \circ
  \alpha = \id_{S}$, for any $a,b \in S$ such that
  $b \prox a$, there exists $A \in \Fin{S}$ such that $b
  \mathrel{\alpha} A \prox_{U} \left\{ a \right\}$.
  %
  %
  Since $A \leq_{U} \emptyset$ and $\alpha$ is approximable, we have $b \in \top$.
  Thus, ${\downarrow_{\prox}a} \subseteq \top$ for all
  $a \in S$, which implies that $\top$ is the greatest element of $\RIdeals{S}$.
  Next, we show that $\RIdeals{S}$ has binary meets.
  %
  %
  To this end, it suffices to show that
  $\downset_{\prox} a$ and $\downset_{\prox}b$
  have a meet for all $a, b  \in S$; for then we have
  $I \land J = \bigcup_{a \in I, b \in J} \downset_{\prox} a
  \land \downset_{\prox} b$ for all $I,J \in \RIdeals{S}$. 
  So let $a, b \in S$, and put $a \wedge b =
  \alpha^{-}\left\{ a,b \right\}$, which is in $\RIdeals{S}$ because
  $\alpha$ is approximable. For each $c \in a \wedge b$, there
  exists $C \in \Fin{S}$ such that $c
  \mathrel{\alpha} C  \prox_{U} \left\{ a,b \right\}$.
  %
  %
  Then, $C \prox_{U} \left\{ a \right\}$ and
  $C \prox_{U} \left\{ b \right\}$,
  so $C \mathrel{\varepsilon^{U}_{S}} a$
  and $ C \mathrel{\varepsilon^{U}_{S}} b$.
  Since $\varepsilon^{U}_{S} \circ \alpha =
  \id_{S}$, we have $c \in \downset_{\prox} a \cap
  \downset_{\prox} b$. Thus, $a \wedge b$ is a lower bound
  of $\downset_{\prox}a$ and $\downset_{\prox}b$.
  Let $I \in \RIdeals{S}$ such that $I \subseteq \downset_{\prox} a \cap
  \downset_{\prox} b$,
  %
  %
  and  let $c \in I$ and
  $c' \prox c$. Then 
  $\left\{ c \right\} \prox_{U} \left\{ a,b \right\}$, and 
  since $\varepsilon^{U}_{S} \circ \alpha =
  \id_{S}$, there exists $C \prox_{U} \left\{ c \right\}$ such that
  $c' \mathrel{\alpha} C$. Then $C \prox_{U} \left\{ a,b \right\}$, and
  since $\alpha$ is approximable, we have $c' \in a \land b$. Thus, $I =
  \medvee_{c \in I} \downset_{\prox}c \subseteq a \wedge
  b$. Therefore, $a \wedge b$ is a meet
  of $\downset_{\prox} a$ and $\downset_{\prox} b$.

  Conversely, suppose that $\RIdeals{S}$ has finite meets. Define a relation
  $\alpha \subseteq S \times \Fin{S}$ by 
  \begin{equation}
    \label{eq:UpperCoalgStruct}
    a \mathrel{\alpha} A
    \defeqiv
    \exists b \in S 
    \left( a \prox b \amp \left\{ b  \right\} \prox_{U} A \right).
  \end{equation}
  %
  %
  Note that $\alpha^{-}A = \medwedge_{a \in A}
  \downset_{\prox}a$, which shows that $\alpha$ satisfies
  \ref{def:approximableI}. From \eqref{eq:UpperCoalgStruct},
  it is also easy to see that $\alpha$ satisfies \ref{def:approximableU}
  and that
  ${\mathrel{\alpha} \circ \prox}
  = {\mathrel{\alpha}}
  = {\prox_{U} \circ \mathrel{\alpha}}$. 
  Thus, $\alpha$ is an approximable relation from $S$ to $\Upper{S} $.
  Then, we clearly have  $\varepsilon^{U}_{S} \circ \alpha =
  \id_{S}$. Hence, $\alpha$ is a $\UpperFunc$-coalgebra by
  Corollary~\ref{cor:KZAlg}~\eqref{cor:KZAlg1}.
\end{proof}
\begin{remark}
  \label{rem:StructUpper}
  Since $\UpperFunc$ is a $\KZ$-comonad, a $\UpperFunc$-coalgebra
  structure on a proximity poset, if it exists, is unique. Thus, it is
  always characterised by \eqref{eq:UpperCoalgStruct}.
\end{remark}

Next, we give an intrinsic characterisation of homomorphisms between
$\UpperFunc$-coalgebras.
\begin{defiC}[{\cite[Definition 3.6]{Infosys}}]
  \label{def:LawsonApproxmable}
  An approximable relation $r \colon {(S, \prox)} \to (S', \prox')$
  between proximity posets is \emph{Lawson approximable}
  if
  \begin{enumerate}
    \item  $a \prox a' \imp \exists b \in S' \left( a
      \mathrel{r} b \right)$,
    \item  $a \prox a'  \mathrel{r} b \amp a' \mathrel{r} c
      \imp \exists d \in b \downarrow_{\prox'}c \left( a
      \mathrel{r} d \right)$.
  \end{enumerate}
\end{defiC}
%
%
In Vickers~\cite[Definition 3.6]{Infosys}, Lawson approximable
relation is defined by 
\begin{equation}
  \label{def:VickesLawsonApproxmable}
  a \prox a' \amp \left\{ a' \right\} \mathrel{r_{U}} B \imp
  \exists b \in S' \left( a \mathrel{r} b \amp \left\{ b  \right\}
  \prox_{U}' B \right)
\end{equation}
for each $a,a' \in S$ and $B \in \Fin{S'}$. By induction on the size
of $B$, one can show that \eqref{def:VickesLawsonApproxmable} is
equivalent to the two conditions in Definition~\ref{def:LawsonApproxmable}.

The following is noted by Vickers~\cite[Section 5.1]{Infosys}.
We give a proof for the sake of completeness.
\begin{proposition}
  \label{prop:CharUpperCoalgHom}
  Let $\alpha \colon S \to \Upper{S}$ and 
  $\beta \colon S' \to \Upper{S'}$ be $\UpperFunc$-coalgebras on
  proximity posets $(S, \prox)$ and $(S', \prox')$. 
  For any approximable relation $r \colon S \to S'$,
  the  following are equivalent:
  \begin{enumerate}
    \item\label{prop:CharUpperCoalgHom1}
      $r$ is a $\UpperFunc$-coalgebra homomorphism.
    \item\label{prop:CharUpperCoalgHom2} $r$ is Lawson approximable.
    \item\label{prop:CharUpperCoalgHom3}
      The Scott continuous function $\overline{f_{r}} \colon
      \RIdeals{S'} \to \RIdeals{S}$ given by
      \eqref{eq:ApproxToScottPxPos} preserves finite meets.
  \end{enumerate}
\end{proposition}
\begin{proof}
  Before getting down to the proof, note that $\alpha$ and $\beta$ are
  given by \eqref{eq:UpperCoalgStruct} (cf.\ Remark~\ref{rem:StructUpper}).
  \smallskip

  \noindent(\ref{prop:CharUpperCoalgHom1} $\to$ \ref{prop:CharUpperCoalgHom2})
  Suppose that $r$ is a $\UpperFunc$-coalgebra homomorphism.  Let $a
  \prox a'$. Then $a \mathrel{\alpha} \emptyset$. Since $\emptyset
  \mathrel{\Upper{r}} \emptyset$, there exists $b \in S'$ such that $a
  \mathrel{r} b$ and $b \mathrel{\beta} \emptyset$. Next, suppose $a
  \prox a' \mathrel{r} b$ and $a' \mathrel{r} c$. Then $a
  \mathrel{(\Upper{r} \circ \alpha)} \left\{ b,c \right\}$. Thus,
  there exists $d \in S'$ such that $a \mathrel{r} d$ and $d
  \mathrel{\beta} \left\{ b,c \right\}$. Then $d  \in b
  \downarrow_{\prox'} c$.
  \smallskip
  
  \noindent(\ref{prop:CharUpperCoalgHom2} $\to$ \ref{prop:CharUpperCoalgHom1})
  Suppose that $r$ is Lawson approximable, and
  let $a \mathrel{(\beta \circ r)} B$. Then, there exist
  $b,b' \in S'$ such that 
%
%
  $a \mathrel{r} b \prox' b'$ and $ \left\{ b' \right\} \prox'_{U} B$.
  Thus, there exists $a' \in S$ such
  that $a \prox a'$ and  $\left\{ a' \right\} \mathrel{r_{U}} B$.
  Then $a \mathrel{(\Upper{r} \circ \alpha)} B$.
  Conversely, suppose $a \mathrel{(\Upper{r} \circ \alpha)} B$.
  Then, there exist $a' \in S$ and $A \in \Fin{S}$ such that
  $a \prox a'$ and $\left\{ a' \right\} \prox_{U} A
  \mathrel{r_{U}} B$, so that $\left\{ a' \right\} \mathrel{r_{U}} B$.
%
%
  Since $r$ is Lawson approximable, there exists $b \in S'$ such
  that $a \mathrel{r} b$ and $\left\{ b \right\} \prox_{U}'  B$ by
  \eqref{def:VickesLawsonApproxmable}.
  Then, $a \mathrel{(\beta \circ r)} B$.
  \smallskip

  The equivalence (\ref{prop:CharUpperCoalgHom2} $\leftrightarrow
  $ \ref{prop:CharUpperCoalgHom3}) is also straightforward to check.
\end{proof}
\begin{theorem}
  \label{thm:EquivPUCoAlgContMSLat}
  The category of $\UpperFunc$-coalgebras over $\PxPos$ is dually
  equivalent to the category of continuous meet semilattices and Scott
  continuous meet semilattice homomorphisms.
\end{theorem}
\begin{proof}
  By Theorem~\ref{prop:SplitAlgPxPos}, Proposition~\ref{prop:AlgUpper}, and Proposition
  \ref{prop:CharUpperCoalgHom}.
\end{proof}
\begin{remark}
  As of this writing, we still lack an intrinsic characterisation of continuous meet
  semilattices in terms of proximity posets. As far as we know,
  $\UpperFunc$-coalgebras provide the best predicative
  characterisation of continuous meet semilattices so far.
\end{remark}


\begin{definition}
  \label{def:ProxRel}
  A join-preserving Lawson approximable relations between localized
  strong proximity $\vee$-semilattices is called a \emph{proximity relation}.
\end{definition}
Let $\PxJLatLoc$ be the category of localized strong proximity
$\vee$-semilattices and
proximity relations.

\begin{theorem}
  \label{thm:LocProxSuplattEquiLKLoc}
  $\PxJLatLoc$ is equivalent to the category of locally
 compact locales.
\end{theorem}
\begin{proof}
  By Theorem~\ref{prop:SPxJLatEquivContLat} and Proposition \ref{prop:LocalizedFrame}, the class of localized
strong proximity $\vee$-semilattices  characterises locally compact
locales among continuous lattices.
Since every strong proximity $\vee$-semilattice is a $\UpperFunc$-coalgebra by
Proposition \ref{prop:AlgUpper}, proximity relations
between localized strong proximity $\vee$-semilattices characterise
locale maps between the corresponding locally compact locales by
Proposition~\ref{prop:CharUpperCoalgHom}.
\end{proof}

\subsection{Algebraic theory of locally compact locales}
\label{sec:DoublePowerLoc}
We give another predicative characterisation of locally compact
locales in terms of algebras of the double powerlocale on the category
of continuous domains.
Specifically, building on the result of Section~\ref{sec:LocalizedPxSL},
we characterise localized strong proximity
$\vee$-semilattices by the coalgebras of the double powerlocale on
$\PxPos$.  This algebraic characterisation of locally compact locales
has been conjectured by Vickers in the dual context of
infosys~\cite[Section 5.1]{Infosys}.

\subsubsection{Finite subsets}\label{sec:FiniteSubsets}
We first recall some properties of finitely enumerable subsets from Vickers
\cite[Section~4]{VickersEntailmentSystem}.
Let $S$ be a set. For each $\mathcal{U} \in \FFin{S}$, define $\mathcal{U}^{*} \in
\FFin{S}$ inductively by
\begin{align*}
  \emptyset^{*}
  &\defeql \left\{ \emptyset \right\},
  &
  \left( \mathcal{U} \cup \left\{ A \right\} \right)^{*}
  &\defeql
  \left\{ B \cup C \mid B \in \mathcal{U}^{*} \amp C \in \PFin{A}\right\},
\end{align*}
where $\PFin{A}$ denotes the set of inhabited finitely enumerable subsets
of $A$.\footnote{
In Vickers~\cite[Section 4]{VickersEntailmentSystem},
the set $\mathcal{U}^{*}$ is equal to
  $
  \left\{ \Image{\gamma} \mid \gamma \in \Choice{\mathcal{U}}\right\},
  $
where $\Choice{\mathcal{U}}$ is the set of choices of $\mathcal{U}$
and $\Image{\gamma}$ is the image of a choice $\gamma$; see Definition
12 and Definition 13, and the proof of Proposition~14 in
\cite{VickersEntailmentSystem}.}
Note that for each $B \in \mathcal{U}^{*}$, we have $B \meets A$
for all $A \in \mathcal{U}$, where
\[
  U \meets V \defeqiv \exists a \in S\left( a \in U \cap V \right)
\]
for $U,V \subseteq S$.

\begin{example}
  For $a,b,c \in S$, 
  we have
  \begin{align*}
    \left\{ \left\{ a \right\} \left\{ b \right\} \right\}^{*}
    &=
    \left\{ \left\{ a,b \right\} \right\}, \\
    \left\{ \left\{ a, b \right\}, \left\{ c \right\}\right\}^{*} 
    &=
    \left\{ A \cup \left\{ c \right\} \mid A \in \PFin{\left\{ a,b \right\}} \right\} 
    =
    \left\{ 
    \left\{ a,b,c \right\}, 
    \left\{ a,c \right\}, 
    \left\{ b,c \right\} \right\}, \\
    \left\{ \left\{ a, b \right\}, \left\{ c \right\}\right\}^{**} 
    &=
    \left\{ A \cup B \cup C \mid
      A \in \PFin{\left\{ a,c \right\}}, 
      B \in \PFin{\left\{ b,c \right\}}, 
      C \in \PFin{\left\{ a,b,c \right\}} \right\}  \\
    &=
    \left\{ \left\{ a,b,c \right\}, 
    \left\{ a,b \right\}, 
    \left\{ a,c \right\},
    \left\{ b,c \right\}, 
    \left\{ c \right\} \right\}.
  \end{align*}
\end{example}
\begin{lemma}
  \label{lem:BasicFFin}
  Let $S$ be a set. For each\, $\mathcal{U} \in \FFin{S}$ and $U
  \subseteq S$, we have
  \begin{enumerate}
    \item\label{lem:BasicFFin1} $  
       \forall C \in \mathcal{U}  \left( U \meets C \right)
        \imp \exists B \in \mathcal{U}^{*} \left( B \subseteq U \right)$,

    \item\label{lem:BasicFFin2} $
      \forall C \in
      \mathcal{U}^{*} \left( U \meets C\right)
      \imp \exists B \in \mathcal{U}     \left( B \subseteq U\right)$,

    \item\label{lem:BasicFFin3} $\forall A \in \mathcal{U}^{**} 
      \exists B \in \mathcal{U} \left( B \subseteq A  \right)$
      and\, 
      $\forall B \in\mathcal{U} \, \exists A \in
      \mathcal{U}^{**} \left( A \subseteq B\right)$.
  \end{enumerate}
\end{lemma}
\begin{proof}
  Items \ref{lem:BasicFFin1} and \ref{lem:BasicFFin2} correspond to
  Proposition~15 and Lemma~16 in
  Vickers~\cite{VickersEntailmentSystem}, respectively.
  Item \ref{lem:BasicFFin3} follows from \ref{lem:BasicFFin1} and
  \ref{lem:BasicFFin2} since $\forall C \in \mathcal{U}^{*}
  \left( C \meets A\right) $ for all $A \in \mathcal{U}^{**}$ and
  $\forall C \in \mathcal{U}^{*}\left(  C \meets B \right)$ for all $B
  \in \mathcal{U}$.
\end{proof}

\begin{lemma}
  \label{lem:UpperLower}
  For any relation $r \subseteq S \times S'$, we have
  \[
    \mathcal{U} \mathrel{\left( r_{L} \right)_{U}} \mathcal{V} \leftrightarrow
    \mathcal{U}^{*} \mathrel{\left( r_{U} \right)_{L}} \mathcal{V}^{*} 
  \]
  for all\, $\mathcal{U} \in \FFin{S}$ and $\mathcal{V} \in
  \FFin{S'}$.
\end{lemma}
\begin{proof}
  First, suppose $\mathcal{U} \mathrel{( r_{L})_{U}} \mathcal{V}$,
  and let $A' \in \mathcal{U}^{*}$. 
  For each $B \in \mathcal{V}$, there exists $A \in \mathcal{U}$ such
  that $A \mathrel{r_{L}} B$. Since $A \meets A'$, we have $r A'
  \meets B$.  Thus, by 
  Lemma \ref{lem:BasicFFin}~\eqref{lem:BasicFFin1}, there
  exists $B' \in \mathcal{V}^{*}$ such that $B' \subseteq r A'$, i.e.,
  $A' \mathrel{r}_{U} B'$.
  Thus, $\mathcal{U}^{*} \mathrel{(r_{U})_{L}} \mathcal{V}^{*}$.
 
  Conversely, suppose 
  $\mathcal{U}^{*} \mathrel{\left( r_{U} \right)_{L}} \mathcal{V}^{*}$.
  Then, $\mathcal{V}^{*} \mathrel{\left( (r^{-})_{L} \right)_{U}} \mathcal{U}^{*}
  $,
  so by the proof of the converse, we have
  $\mathcal{V}^{**} \mathrel{\left( (r^{-})_{U} \right)_{L}} \mathcal{U}^{**}$,
  i.e., $\mathcal{U}^{**} \mathrel{\left( r_{L} \right)_{U}} \mathcal{V}^{**}$.
  Since $\mathcal{U} \subseteq_{U} \mathcal{U}^{**}$
  and $\mathcal{V}^{**} \subseteq_{U} \mathcal{V}$ by 
  Lemma \ref{lem:BasicFFin}~\eqref{lem:BasicFFin3}, we have
  $\mathcal{U} \mathrel{( r_{L})_{U}} \mathcal{V}$.
\end{proof}

The following is also useful later. 
\begin{lemma}
  \label{lem:StarSingleton}
  Let $S$ be a set. For each $A \in \Fin{S}$ and $\mathcal{U} \in
  \Fin{\Fin{S}}$, we have 
  \begin{enumerate}
    \item \label{lem:StarSingleton1}
      $\left\{ A \right\}^{*} \subseteq_U \Sin{A}
      \subseteq_U \left\{ A \right\}^{*}$,

    \item \label{lem:StarSingleton2}
      $\left\{ \Sin{A} \mid A \in \mathcal{U} \right\}^{*}
      \subseteq_U
      \left\{ \Sin{B} \mid B \in \mathcal{U}^{*} \right\}
      \subseteq_U
      \left\{ \Sin{A} \mid A \in \mathcal{U} \right\}^{*}$.
  \end{enumerate}
\end{lemma}
\begin{proof}
  \noindent \ref{lem:StarSingleton1}.
  This is immediate from $\left\{ A \right\}^{*} = \PFin{A}$.

  \smallskip
  \noindent \ref{lem:StarSingleton2}.
  For the first inclusion, let 
  $B \in \mathcal{U}^{*}$. Then, $B \meets A$ for all $A \in
  \mathcal{U}$, i.e., $\Sin{B} \meets \Sin{A}$ for all $A \in
  \mathcal{U}$. By Lemma~\ref{lem:BasicFFin}~\eqref{lem:BasicFFin1},
  there exists $C \in \left\{ \Sin{A} \mid A \in \mathcal{U} \right\}^{*}$
  such that $C \subseteq \Sin{B}$.
  Next,
  let $C \in \left\{ \Sin{A} \mid A \in \mathcal{U} \right\}^{*}$.
  Then, $C \meets \Sin{A}$ for each $A \in \mathcal{U}$, i.e.,
  $\left\{ a \in S \mid \left\{ a \right\} \in C \right\} \meets A$
  for each $A \in \mathcal{U}$. By
  Lemma~\ref{lem:BasicFFin}~\eqref{lem:BasicFFin1}, there exists $B
  \in \mathcal{U}^{*}$ such that $B \subseteq \left\{ a \in S \mid \left\{
  a \right\} \in C \right\}$. Then $\Sin{B} \subseteq C$.
\end{proof}

\subsubsection{Double powerlocale}
As is well known in locale theory \cite{DoublePowLocExp}, the
two compositions $\UpperFunc \circ \LowerFunc$ and $\LowerFunc \circ
\UpperFunc$ of the upper and lower powerlocales are naturally isomorphic. Indeed, for each proximity poset
$(S,\prox)$, there is an approximable relation $\sigma_{S} \colon
\Upper{\Lower{S}} \to
\Lower{\Upper{S}}$ defined by
\begin{equation}
  \label{eq:DistLaw}
  \mathcal{U} \mathrel{\sigma_{S}} \mathcal{V} \defeqiv
  \mathcal{U} \mathrel{(\prox_{L})_{U}} \mathcal{V}^{*}.
\end{equation}
For any approximable relation $r \colon (S, \prox) \to (S', \prox')$,
we have
\begin{align*}
  &\mathcal{U} \mathrel{(\sigma_{S'} \circ \Upper{\Lower{r}})}
  \mathcal{V}\\
  &\iff  
  \mathcal{U} \mathrel{((\prox' \circ \mathrel{r})_{L})_{U}}
  \mathcal{V}^{*}\\
  &\iff  
  \mathcal{U} \mathrel{((\mathrel{r} \circ \prox)_{L})_{U}} \mathcal{V}^{*}
  && (\text{$r$ is approximable})\\
  &\iff  \exists \mathcal{W} 
  \left( \mathcal{U} \mathrel{(\prox_{L})_{U}} \mathcal{W}^{**}
  \amp \mathcal{W}^{**} \mathrel{(r_{L})_{U}} \mathcal{V}^{*}\right) 
  && (\text{by Lemma \ref{lem:BasicFFin}~\eqref{lem:BasicFFin3}})
  \\
  &\iff  \exists \mathcal{W} 
  \left( \mathcal{U} \mathrel{(\prox_{L})_{U}} \mathcal{W}^{**}
  \amp \mathcal{W}^{*} \mathrel{(r_{U})_{L}} \mathcal{V}\right) 
  && (\text{by Lemma \ref{lem:UpperLower}})
  \\
  &\iff  
  \mathcal{U} \mathrel{(\Lower{\Upper{r}} \circ \sigma_{S})}
  \mathcal{V},
  && (\text{by Lemma \ref{lem:BasicFFin}~\eqref{lem:BasicFFin3}})
\end{align*}
where $\mathcal{W}$ ranges over $\Fin{\Fin{S}}$.
Thus, $\sigma$ is natural in $S$. 
Moreover, there is an approximable relation 
$\tau_{S} \colon \Lower{\Upper{S}} \to \Upper{\Lower{S}}$ 
defined by
\begin{equation}
  \mathcal{V} \mathrel{\tau_{S}} \mathcal{U} \defeqiv
  \mathcal{V} \mathrel{(\prox_{U})_{L}} \mathcal{U}^{*}.
  \label{eq:DistLawInverse}
\end{equation}
We have
\begin{align*}
  &\mathcal{U} \mathrel{(\tau_{S} \circ \sigma_{S})} \mathcal{V} \\
  &\iff  \exists \mathcal{W} 
  \left( \mathcal{U} \mathrel{(\prox_{L})_{U}} \mathcal{W}^{*}
  \amp \mathcal{W} \mathrel{(\prox_{U})_{L}} \mathcal{V}^{*}\right) \\
  &\iff  \exists \mathcal{W} 
  \left( \mathcal{U} \mathrel{(\prox_{L})_{U}} \mathcal{W}^{*}
  \amp \mathcal{W}^{*} \mathrel{(\prox_{L})_{U}} \mathcal{V}^{**}\right)
  && \text{(by Lemma \ref{lem:UpperLower})} \\
  &\iff  \exists \mathcal{W} 
  \left( \mathcal{U} \mathrel{(\prox_{L})_{U}} \mathcal{W}^{*}
  \amp \mathcal{W}^{*} \mathrel{(\prox_{L})_{U}} \mathcal{V}\right)
  && \text{(by Lemma~\ref{lem:BasicFFin}~\eqref{lem:BasicFFin3})} \\
  &\iff \mathcal{U} \mathrel{(\prox_{L})_{U}}\mathcal{V},
  && \text{(by Lemma \ref{lem:BasicFFin}~\eqref{lem:BasicFFin3})}
\end{align*}
where $\mathcal{W}$ ranges over $\Fin{\Fin{S}}$.
Thus, $\tau_{S} \circ \sigma_{S} = \id_{\Upper{{\Lower{S}}}}$.
Similarly, we have $\sigma_{S} \circ \tau_{S} =
\id_{\Lower{{\Upper{S}}}}$. Hence $\sigma$ and $\tau$ are inverse to each other.

We show that $\sigma \colon \UpperFunc \circ \LowerFunc
\to \LowerFunc \circ \UpperFunc$ satisfies the distributive law of
comonad, which is the dual of the distributive law of
monad~\cite{BeckDistributiveLaw}.
\begin{definition}
  \label{def:DistributiveLaw}
Let $\langle K, \varepsilon^{K}, \nu^{K} \rangle$ and 
$\langle T, \varepsilon^{T}, \nu^{T} \rangle$ be comonads
on a category~$\mathbb{C}$.
A \emph{distributive law} of $K$ over $T$ is a natural transformation
$\sigma \colon T \circ K \to K \circ T$ which makes the following diagrams
commute:
  \[
    \begin{tikzcd}
      [sep=large, row sep=huge]
      {}   & T \circ K
      \ar[d, "\sigma"]
      \ar[dl, "\varepsilon^{T}K"']
      \ar[dr, "T\varepsilon^{K}"]
      & {}
      \\
      K
      & K \circ T 
      \ar[l, "K \varepsilon^{T}"]
      \ar[ul, phantom, "(1)", pos=0.22]
      \ar[ur, phantom, "(2)", pos=0.22]
      \ar[r, "\varepsilon^{K}T"']
      & T
    \end{tikzcd}
    \quad
    \begin{tikzcd}
      T \circ K \circ K
      \ar[d, "\sigma K"']
      & T \circ K 
      \ar[l, "T \nu^{K}"']
      \ar[dd, "\sigma"]
      \ar[r, "\nu^{T}K"]
      & T \circ T \circ K
      \ar[d, "T\sigma"]\\
      K \circ T \circ K
      \ar[d, "K \sigma"']
      \ar[rr, phantom, "(3)", pos=0.2]
      \ar[rr, phantom, "(4)", pos=0.8]
      & 
      & 
      T \circ K \circ T
      \ar[d, "\sigma T"]\\
      K \circ K \circ T
      & 
      \ar[l, "\nu^{K}T"]
      K \circ T 
      \ar[r, "K \nu^{T}"']
      &
      K \circ T \circ T
    \end{tikzcd}
  \]
  In this situation, $T \circ K$ is a comonad with a counit $\varepsilon$ and a co-multiplication
  $\nu$ defined by
  \begin{equation}
    \label{eq:CompositeMonad}
  \begin{aligned}
    \varepsilon &=
    T \circ K \xrightarrow{T \varepsilon^{K}}
    T \xrightarrow{\varepsilon^{T}}
    \id_{\mathbb{C}},\\
    \nu &= 
    T \circ K
    \xrightarrow{T\nu^{K}} 
    T \circ K \circ K
    \xrightarrow{\nu^{T}K\circ K} 
    T \circ T \circ K \circ K
    \xrightarrow{T \sigma K} 
    T \circ K \circ T \circ K.
  \end{aligned}
  \end{equation}
\end{definition}

\begin{proposition}
  \label{prop:LowerUpperDistributiveLaw}
  $\sigma$ given by \eqref{eq:DistLaw} is a distributive law of
$\LowerFunc$ over $\UpperFunc$.
\end{proposition}
\begin{proof}
  The commutativity of the diagrams (1) and (2) are
  easy to check. 
  We show that the diagram (3) commutes for $T = \UpperFunc$ and $K = \LowerFunc$. 
  Fix a proximity poset $(S,\prox)$. It suffices to show
  \begin{align}
  \label{eq:DistrLaw_3_1}
  \Lower{\tau_{S}} \circ \nu^{L}_{\Upper{S}} \circ
      \sigma_{S} 
      &\leq \sigma_{\Lower{S}} \circ \Upper{\nu^{L}_{S}},\\
  \label{eq:DistrLaw_3_2}
  \sigma_{\Lower{S}} \circ \Upper{\nu^{L}_{S}} \circ \tau_{S}
  &\leq  \Lower{\tau_{S}} \circ \nu^{L}_{\Upper{S}}.
  \end{align}
  In the proof below, we identify each proximity poset with its
  underlying set.

  First, to see that \eqref{eq:DistrLaw_3_1} holds,
  let $\mathcal{X} \in \Upper{\Lower{S}}$ and
  $\mathbb{U} \in \Lower{\Upper{\Lower{S}}}$, and suppose
  $\mathcal{X} \mathrel{(\Lower{\tau_{S}} \circ \nu^{L}_{\Upper{S}} \circ
  \sigma_{S})} \mathbb{U}$. 
  Then, there exist $\mathcal{Y} \in \Lower{\Upper{S}}$ 
  and $\mathbb{V} \in \Lower{\Lower{\Upper{S}}}$ such that
  $\mathcal{X} \mathrel{\sigma_{S}} 
  \mathcal{Y} \mathrel{\nu^{L}_{\Upper{S}}} \mathbb{V} 
  \mathrel{\Lower{\tau_{S}}} \mathbb{U}$.
  Thus
  \begin{enumerate}
    \item\label{eq:DistrLaw_3_1_1} $\mathcal{X}
      \mathrel{(\prox_{L})_{U}} \mathcal{Y}^{*}$,

    \item\label{eq:DistrLaw_3_1_2} $\Sin{\mathcal{Y}} \mathrel{\left( \left( \prox_{U}
      \right)_{L} \right)_{L}} \mathbb{V}$,

    \item\label{eq:DistrLaw_3_1_3} $\mathbb{V} \mathrel{\left( \left( \prox_{U}
      \right)_{L} \right)_{L}} \left\{ \mathcal{U}^{*} \mid
      \mathcal{U} \in \mathbb{U} \right\}$.
  \end{enumerate}
  By \ref{eq:DistrLaw_3_1_2}
  and \ref{eq:DistrLaw_3_1_3}, we have 
  $\Sin{\mathcal{Y}} 
  \mathrel{\left( \left( \prox_{U} \right)_{L} \right)_{L}}
  \left\{ \mathcal{U}^{*} \mid \mathcal{U} \in \mathbb{U} \right\}$.
  Then
  \begin{align*}
  &\Sin{\mathcal{Y}} 
  \mathrel{\left( \left( \prox_{U} \right)_{L} \right)_{L}}
  \left\{ \mathcal{U}^{*} \mid \mathcal{U} \in \mathbb{U} \right\}\\
  &\iff
  \left\{ \left\{ Y \right\}^{**} \mid Y \in \mathcal{Y}\right\} 
  \mathrel{\left( \left( \prox_{U} \right)_{L} \right)_{L}}
  \left\{ \mathcal{U}^{*} \mid \mathcal{U} \in \mathbb{U} \right\}
  && \text{(by  Lemma \ref{lem:BasicFFin}~\eqref{lem:BasicFFin3})} \\
  &\iff
  \left\{ \left\{ Y \right\}^{*} \mid Y \in \mathcal{Y}\right\} 
  \mathrel{\left( \left( \prox_{L} \right)_{U} \right)_{L}}
  \mathbb{U}
  && \text{(by  Lemma \ref{lem:UpperLower})} \\
  &\iff
  \left\{ \Sin{Y} \mid Y \in \mathcal{Y}\right\} 
  \mathrel{\left( \left( \prox_{L} \right)_{U} \right)_{L}}
  \mathbb{U}
  && \text{(by Lemma \ref{lem:StarSingleton}~\eqref{lem:StarSingleton1})} \\
  &\iff
  \left\{ \Sin{Y} \mid Y \in \mathcal{Y}\right\}^{*}
  \mathrel{\left( \left( \prox_{L} \right)_{L} \right)_{U}}
  \mathbb{U}^{*}
  && \text{(by Lemma \ref{lem:UpperLower})} \\
  &\iff
  \left\{ \Sin{Y} \mid Y \in \mathcal{Y}^{*}\right\}
  \mathrel{\left( \left( \prox_{L} \right)_{L} \right)_{U}}
  \mathbb{U}^{*}.
  && \text{(by Lemma \ref{lem:StarSingleton}~\eqref{lem:StarSingleton2})}
  \end{align*}
  On the other hand, \ref{eq:DistrLaw_3_1_1} implies
  $\left\{ \Sin{X} \mid X \in \mathrel{X}\right\}
  \mathrel{\left( \left( \prox_{L} \right)_{L} \right)_{U}}
  \left\{ \Sin{Y} \mid Y \in \mathcal{Y}^{*}\right\}$. Thus, by putting
  $\mathbb{W} = \left\{ \Sin{Y} \mid Y \in \mathcal{Y}^{*}\right\}$,
  we have
  $\mathcal{X} \mathrel{\Upper{\nu^{L}_{S}}} \mathbb{W}$
  and $\mathbb{W} \mathrel{\sigma_{\Lower S}} \mathbb{U}$. Hence
  $\mathcal{X} \mathrel{(\sigma_{\Lower S} \circ
  \Upper{\nu^{L}_{S}})} \mathbb{U}$.

  Next, we verify \eqref{eq:DistrLaw_3_2}. Let $\mathcal{Y} \in
  \Lower{\Upper{S}}$ and
  $\mathbb{U} \in \Lower{\Upper{\Lower{S}}}$, and suppose 
  $\mathcal{Y} \mathrel{(\sigma_{\Lower{S}} \circ \Upper{\nu^{L}_{S}} \circ
  \tau_{S})} \mathbb{U}$. Then, there exist $\mathcal{X} \in
  \Upper{\Lower{S}}$ and $\mathbb{W} \in
  \Upper{\Lower{\Lower{S}}}$ such that $\mathcal{Y}
  \mathrel{\tau_{S}} \mathcal{X} \mathrel{\Upper{\nu^{L}_{S}}} \mathbb{W}
  \mathrel{\sigma_{\Lower{S}}} \mathbb{U}$.
  Thus
  \begin{enumerate}[resume]
    \item\label{eq:DistrLaw_3_1_4} 
      $\mathcal{Y} \mathrel{(\prox_{U})_{L}} \mathcal{X}^{*}$,
    \item\label{eq:DistrLaw_3_1_5} 
      $\left\{ \Sin{X} \mid X \in \mathcal{X} \right\}
      \mathrel{\left( \left( \prox_{L} \right)_{L} \right)_{U}} 
      \mathbb{W}$,
    \item\label{eq:DistrLaw_3_1_6}
      $\mathbb{W} \mathrel{\left( \left( \prox_{L}
      \right)_{L} \right)_{U}} \mathbb{U}^{*}$.
  \end{enumerate}
  By \ref{eq:DistrLaw_3_1_5} and  \ref{eq:DistrLaw_3_1_6}, 
  we have 
  $\left\{ \Sin{X} \mid X \in \mathcal{X} \right\}
      \mathrel{\left( \left( \prox_{L} \right)_{L} \right)_{U}} 
  \mathbb{U}^{*}$. Then
  \begin{align*}
  &\left\{ \Sin{X} \mid X \in \mathcal{X} \right\}
      \mathrel{\left( \left( \prox_{L} \right)_{L} \right)_{U}}
      \mathbb{U}^{*}\\
  &\iff
  \left\{ \Sin{X} \mid X \in \mathcal{X} \right\}^{**}
      \mathrel{\left( \left( \prox_{L} \right)_{L} \right)_{U}}
      \mathbb{U}^{*}
  && \text{(by  Lemma \ref{lem:BasicFFin}~\eqref{lem:BasicFFin3})} \\
  &\iff
  \left\{ \Sin{X} \mid X \in \mathcal{X} \right\}^{*}
      \mathrel{\left( \left( \prox_{L} \right)_{U} \right)_{L}} \mathbb{U}
      && \text{(by Lemma \ref{lem:UpperLower})} \\
  &\iff
  \left\{ \Sin{X} \mid X \in \mathcal{X}^{*} \right\}
      \mathrel{\left( \left( \prox_{L} \right)_{U} \right)_{L}} \mathbb{U}
  && \text{(by Lemma \ref{lem:StarSingleton}~\eqref{lem:StarSingleton2})} \\
  &\iff
  \left\{ \left\{ X \right\}^{*} \mid X \in \mathcal{X}^{*} \right\}
      \mathrel{\left( \left( \prox_{L} \right)_{U} \right)_{L}} \mathbb{U}
  && \text{(by Lemma \ref{lem:StarSingleton}~\eqref{lem:StarSingleton1})} \\
  &\iff
  \left\{ \left\{ X \right\}^{**} \mid X \in \mathcal{X}^{*} \right\}
      \mathrel{\left( \left( \prox_{U} \right)_{L} \right)_{L}} 
      \left\{ \mathcal{U}^{*} \mid \mathcal{U} \in \mathbb{U} \right\}
      && \text{(by Lemma \ref{lem:UpperLower})} \\
  &\iff
  \Sin{\mathcal{X}^{*}}
      \mathrel{\left( \left( \prox_{U} \right)_{L} \right)_{L}} 
      \left\{ \mathcal{U}^{*} \mid \mathcal{U} \in \mathbb{U} \right\}.
        && \text{(by  Lemma \ref{lem:BasicFFin}~\eqref{lem:BasicFFin3})}
  \end{align*}
  On the other hand, \ref{eq:DistrLaw_3_1_4} implies
  $\Sin{\mathcal{Y}} 
  \mathrel{\left( \left( \prox_{U} \right)_{L} \right)_{L}} 
  \Sin{\mathcal{X}^{*}}$. Thus, $\mathcal{Y} \mathrel{\nu^{L}_{\Upper{S}}}
  \Sin{\mathcal{X}^{*}}$ and $\Sin{\mathcal{X}^{*}}
  \mathrel{\Lower{\tau_{S}}} \mathbb{U}$.
  Hence $\mathcal{Y} \mathrel{(\Lower{\tau_{S}} \circ
  \nu^{L}_{\Upper{S}})} \mathbb{U}$.

  The commutativity of (4) can be proved similarly.
\end{proof}

Thus, $\UpperFunc \circ \LowerFunc$ and $\LowerFunc \circ \UpperFunc$
give rise to equivalent comonads on $\PxPos$.
\begin{definition}
  The \emph{double powerlocale} on $\PxPos$ is
  the composition $\UpperFunc \circ \LowerFunc$ (or equivalently the
  composition $\LowerFunc \circ \UpperFunc $).%
  \footnote{
  In this paper, we choose $\UpperFunc \circ \LowerFunc$ as the double powerlocale.}
\end{definition}

Next, we recall some properties of distributive laws to obtain
convenient characterisations of coalgebras of double powerlocales
and homomorphisms between them.  Fix comonads
$\langle T, \varepsilon^{T}, \nu^{T} \rangle$ and $\langle K,
\varepsilon^{K}, \nu^{K} \rangle$ on a category $\mathbb{C}$ and a
distributive law $\sigma \colon T \circ K \to K \circ T$ of $K$ over
$T$. Let $\langle H, \varepsilon^{H},\nu^{H} \rangle$ be the composite
comonad $T \circ K$ where 
$\varepsilon^{H}$ and $\nu^{H}$ are given by
\eqref{eq:CompositeMonad}.
\begin{lemma}\label{prop:HAlgIsTKAlg}
  If $\alpha \colon X \to H X$ is an $H$-coalgebra, then
  \begin{align*}
      \alpha_{T} &\defeql T \varepsilon^{K}_{X} \circ \alpha,
      &
      \alpha_{K} &\defeql \varepsilon^{T}_{KX} \circ \alpha 
  \end{align*}
  are $T$-coalgebra and $K$-coalgebra, respectively, and make the
  following diagram commute:
\begin{equation*}
  \begin{tikzcd}
    TK X \arrow[rr,"\sigma_{X}"] &  & KT X \\
    T X \arrow[u,"T\alpha_{K}"] & & K X \arrow[u,"K\alpha_{T}"'] \\
    & \arrow[ul,"\alpha_{T}"]X \arrow[ur,"\alpha_{K}"']&
  \end{tikzcd}
\end{equation*}
Moreover, $\alpha = T\alpha_{K} \circ \alpha_{T}$.
\end{lemma}
\begin{proof}
  By direct calculations, one can show
  $\varepsilon^{T}_{X} \circ  \alpha_{T} = \id_{X}$ and 
  $\nu^{T}_{X} \circ \alpha_{T} =  T \alpha_{T} \circ \alpha_{T}$,
  and similarly for $\alpha_{K}$.
Moreover,
\begin{align*}
  T \alpha_{K} \circ \alpha_{T}
  &= T(\varepsilon^{T}_{KX} \circ \alpha) \circ T \varepsilon^{K}_{X} \circ \alpha
  \\
  &= T \varepsilon^{T}_{KX} \circ T \varepsilon^{K}_{TKX}\circ TK \alpha \circ \alpha 
  && \text{(by the naturality of $\varepsilon^{K}$)}
  \\
  &= T \varepsilon^{T}_{KX} \circ T \varepsilon^{K}_{TKX}\circ
  \left(T \sigma_{KX} \circ \nu^{T}_{KKX} \circ T \nu^{K}_{X} \right)
  \circ \alpha 
  && \text{($\alpha$ is an $H$-coalgebra)}
  \\
  &= T \varepsilon^{T}_{KX} \circ TT \varepsilon^{K}_{KX}\circ
  \nu^{T}_{KKX} \circ T \nu^{K}_{X} \circ \alpha 
  && \text{(by the diagram (2))}
  \\
  &= T \varepsilon^{T}_{KX} \circ 
  \nu^{T}_{KX} \circ T \varepsilon^{K}_{KX}\circ T \nu^{K}_{X} \circ \alpha
  && \text{(by the naturality of $\nu^{T}$)}
  \\
  &= \alpha.
  && \text{($K$ and $T$ are comonads)}
\end{align*}
Similarly, we have $K \alpha_{T} \circ \alpha_{K} = \sigma_{X} \circ
\alpha$. 
\end{proof}
\begin{lemma}\label{prop:CharHalg}
  Let $\alpha \colon X \to T X$  and $\beta \colon X
  \to K X$ 
  be a $T$-coalgebra and a $K$-coalgebra, respectively.
  Then, $\gamma \defeql T \beta \circ \alpha$ is an $H$-coalgebra
  if and only if 
  $\sigma_{X} \circ T \beta \circ \alpha = K \alpha \circ \beta$.
\end{lemma}
\begin{proof}
  ($\Rightarrow$)
  This follows from Lemma \ref{prop:HAlgIsTKAlg} because
  \begin{align*}
    \gamma_{T} &= T \varepsilon^{K}_{X} \circ T \beta \circ \alpha  = \alpha,\\
    \gamma_{K} &= \varepsilon^{T}_{KX} \circ T \beta \circ \alpha
    = \beta \circ \varepsilon^{T}_{X} \circ \alpha = \beta.
  \end{align*}

  \noindent($\Leftarrow$)
  If
  $\sigma_{X} \circ T \beta \circ \alpha = K \alpha \circ \beta$,
  then
  \begin{align*}
    \varepsilon^{H}_{X} \circ \gamma 
    &= \varepsilon^{H}_{X} \circ T\beta \circ \alpha\\
    &= \varepsilon^{T}_{X} \circ T \varepsilon^{K}_{X} \circ T \beta
    \circ \alpha  \\
    &= \id_{X},
    && \text{($\alpha$ and $\beta$ are coalgebras)}
    \\[.5em]
     \nu^{H}_{X} \circ \gamma
    &= \nu^{H}_{X} \circ T\beta \circ \alpha\\
    &= T\sigma_{KX} \circ \nu^{T}_{KKX} \circ T\nu^{K}_{X}
     \circ T\beta \circ \alpha
     && \text{(by \eqref{eq:CompositeMonad})}
     \\
    &= T\sigma_{KX} \circ \nu^{T}_{KKX} \circ T(K\beta \circ \beta)
     \circ \alpha
    && \text{($\beta$ is a $K$-coalgebra)}
     \\
    &= T\sigma_{KX} \circ TT(K\beta \circ \beta) \circ \nu^{T}_{X}
     \circ \alpha
  && \text{(by the naturality of $\nu^{T}$)}
     \\
    &= T(KT\beta \circ \sigma_{X}) \circ TT \beta \circ \nu^{T}_{X}
     \circ \alpha
  && \text{(by the naturality of $\sigma$)}
     \\
    &= T(KT\beta \circ \sigma_{X}) \circ TT \beta \circ T\alpha \circ \alpha
    && \text{($\alpha$ is a $T$-coalgebra)}
    \\
    &= TKT\beta \circ TK\alpha \circ T\beta \circ \alpha
    && \text{(by $\sigma_{X} \circ T \beta \circ \alpha = K \alpha \circ \beta$)}
    \\
    &= H\gamma \circ \gamma.
  \end{align*}
  Thus, $\gamma$ is an $H$-coalgebra.
\end{proof}

Homomorphisms  between $H$-coalgebras can be characterised by their
underlying homomorphisms of $T$-coalgebras and $K$-coalgebras.
In the following lemma, we use the notation from Lemma \ref{prop:HAlgIsTKAlg}.
\begin{lemma}\label{prop:CharHalgHom}
  Let $\alpha \colon X \to H X$ and $\beta \colon Y \to HY$ be
  $H$-coalgebras. Then $f \colon X \to Y$
  is an $H$-coalgebra homomorphism from $\alpha$ to $\beta$ if and only if it is
  a $T$-coalgebra homomorphism from $\alpha_{T}$ to $\beta_{T}$
  and a $K$-coalgebra homomorphism from $\alpha_{K}$ to~$\beta_{K}$.
\end{lemma}
\begin{proof}
  ($\Rightarrow$)
  Immediate from the naturality of $\varepsilon^{K}$
  and $\varepsilon^{T}$.
  \smallskip

  \noindent($\Leftarrow$)
  If $f$ is a
  homomorphism of the underlying $T$-coalgebra and $K$-coalgebra structures of $\alpha$ and $\beta$, 
  then
  \[
    \beta \circ f
    =   (T \beta_{K} \circ \beta_{T})  \circ f
    =  T\beta_{K} \circ T f  \circ \alpha_{T}
    =  T(K f \circ \alpha_{K})  \circ \alpha_{T}
    =  H f \circ \alpha.
    \qedhere
  \]
\end{proof}

We come back to the context of the distributive law of $\LowerFunc$
over $\UpperFunc$ in $\PxPos$.
Note that by the last part of Lemma \ref{prop:HAlgIsTKAlg}, a double
powerlocale structure on a proximity poset, if it exists, is
completely determined by the underlying $\LowerFunc$-coalgebra and
$\UpperFunc$-coalgebra structures. Since $\LowerFunc$ and $\UpperFunc$
are (co)$\KZ$-comonads, these coalgebra structures, if they exist, are
unique.  Thus, each proximity poset admits at most one double
powerlocale coalgebra structure.
\begin{proposition}
  \label{prop:LocalizedDoubleCoalg}
  A strong proximity $\vee$-semilattice is localized if and only if it is a
  coalgebra of the double powerlocale over $\PxPos$.
\end{proposition}
\begin{proof}
  Let $(S,0,\vee,\prox)$ be a strong proximity $\vee$-semilattice.
  Before getting down to the proof, note that
  by Theorem \ref{thm:LowerCoAlgEquivPxSupLat} and 
  Proposition~\ref{prop:AlgUpper}, $S$ is both a $\LowerFunc$-coalgebra and 
  a $\UpperFunc$-coalgebra, and their coalgebra structures 
   $\alpha \colon S \to \Lower{S}$ and $\beta \colon S \to
  \Upper{S}$ are given by \eqref{eq:Lowerunit}
  and \eqref{eq:UpperCoalgStruct}, respectively.
  Then, by Lemma~\ref{prop:LowerUpperDistributiveLaw} and 
  Lemma~\ref{prop:CharHalg}, $S$ is a coalgebra of
  the double powerlocale
  if and only if 
  \begin{equation}
    \label{eq:EquivDoublePowerCoalgebra}
    \sigma_{S} \circ \Upper{\alpha} \circ \beta = \Lower{\beta} \circ
    \alpha.
  \end{equation}

  \noindent($\Rightarrow$) Suppose that $S$ is localized.
  It suffices to show \eqref{eq:EquivDoublePowerCoalgebra}, or
  equivalently
  \begin{align}
    \label{eq:DistLaw1}
    \tau_{S} \circ \Lower{\beta} \circ \alpha
    &\leq 
    \Upper{\alpha} \circ \beta, \\
    \label{eq:DistLaw2}
    \sigma_{S} \circ \Upper{\alpha} \circ \beta
    &\leq
    \Lower{\beta} \circ \alpha.
  \end{align}
  First, to see that \eqref{eq:DistLaw1} holds, suppose
  $a \mathrel{(\tau_{S} \circ \Lower{\beta} \circ \alpha)}
  \mathcal{U}$. Then, there exist $A \in \Fin{S}$ and
  $\mathcal{V} \in \FFin{S}$ such that
  \begin{enumerate}
    \item \label{eq:DistLaw1_1}
      $a \prox \medvee A$, 
    \item \label{eq:DistLaw1_2}
      $A \mathrel{\beta_{L}} \mathcal{V}$,
    \item \label{eq:DistLaw1_3}
      $\mathcal{V} \left( \prox_{U} \right)_{L} \mathcal{U}^{*}$.
  \end{enumerate}
  Then, \ref{eq:DistLaw1_2} implies $\Sin{A}
  \mathrel{\left(\prox_{U}\right)_{L}} \mathcal{V}$, so by
  \ref{eq:DistLaw1_3}, we have
  $\Sin{A} \mathrel{\left(\prox_{U}\right)_{L}} \mathcal{U}^{*}$.
  Then, by Lemma~\ref{lem:UpperLower} and
  Lemma~\ref{lem:StarSingleton}~\eqref{lem:StarSingleton1}, we have
  $ \left\{ A  \right\} \mathrel{\left(\prox_{L} \right)_{U}} \mathcal{U}$.
  By putting $B = \left\{ \medvee A \right\}$, we have
  $a \mathrel{\beta} B$ 
  and $B \mathrel{\Upper{\alpha}} \mathcal{U}$.
  Thus, $a \mathrel{\left( \Upper{\alpha} \circ \beta  \right)} \mathcal{U}$.

  Next, to see that \eqref{eq:DistLaw2} holds, suppose
  $a \mathrel{(\sigma_{S} \circ \Upper{\alpha} \circ \beta)}
  \mathcal{V}$. Then, there exist $a' \in S$, $B \in \Fin{S}$, and
  $\mathcal{U} \in \FFin{S}$ such that
  \begin{enumerate}[resume]
    \item \label{eq:DistLaw1_4}
      $a \prox a'$ and $\left\{ a' \right\} \prox_{U} B$, 
    \item \label{eq:DistLaw1_5}
      $B \mathrel{\alpha_{U}} \mathcal{U}$,
    \item \label{eq:DistLaw1_6}
      $\mathcal{U} \left( \prox_{L} \right)_{U} \mathcal{V}^{*}$.
  \end{enumerate}
  By  \ref{eq:DistLaw1_4} and \ref{eq:DistLaw1_5}, we have $a' \prox
  \medvee C$ for all $C \in \mathcal{U}$. 
  Since $S$ is localized,
  there exists $A \in \Fin{S}$ such that $a \prox \medvee A$ and $A
  \prox_{L} C$ for all $C \in \mathcal{U}$ by \eqref{eq:LocalizedGeneralFinite},
  i.e., $\left\{ A \right\} \mathrel{(\prox_{L})_{U}} \mathcal{U}$. Then, by
  Lemma~\ref{lem:UpperLower}, we have
  $\left\{ A \right\}^{*} \mathrel{(\prox_{U})_{L}} \mathcal{U}^{*}$, so that
  $\Sin{A} \mathrel{(\prox_{U})_{L}} \mathcal{U}^{*}$ by
  Lemma~\ref{lem:StarSingleton}~\eqref{lem:StarSingleton1}.
  On the other hand, \ref{eq:DistLaw1_6} implies
  $\mathcal{U}^{*} \mathrel{\left( \prox_{U} \right)_{L}} \mathcal{V}$
  by Lemma~\ref{lem:BasicFFin}~\eqref{lem:BasicFFin3} and Lemma~\ref{lem:UpperLower}.
  Thus, $\Sin{A} \mathrel{\left( \prox_{U} \right)_{L}} \mathcal{V}$.
  Moreover, since $S$ is strong, there exists $A' \in
  \Fin{S}$ such that $a \prox \medvee A'$ and $A' \prox_{L} A$. Then,
  $a \mathrel{\alpha} A'$ and $A' \mathrel{\Lower{\beta}}
  \mathcal{V}$, and hence $a \mathrel{\left(\Lower{\beta} \circ \alpha
\right)} \mathcal{V}$.
  \smallskip

  \noindent($\Leftarrow$)
  Suppose that $\alpha$ and $\beta$ satisfy \eqref{eq:EquivDoublePowerCoalgebra}, and
  let $a \prox b \leq c \vee d$. Put $\mathcal{U} = \left\{ \left\{ b
  \right\}, \left\{ c,d \right\} \right\}$, and choose $a'$ such that 
  $a \prox a' \prox b$. Then 
  $a \mathrel{\beta} \left\{ a' \right\} \mathrel{\Upper{\alpha}}
  \mathcal{U}$,
  so by the opposite of \eqref{eq:DistLaw1}, we have
  $a \mathrel{(\tau_{S} \circ \Lower{\beta} \circ \alpha)}
  \mathcal{U}$. Thus, there exist $A \in \Fin{S}$ and
  $\mathcal{V} \in \FFin{S}$ such that 
  \begin{enumerate}[resume]
    \item\label{eq:DistrLaw_3_1_9}
      $a \prox \medvee A$,
    \item\label{eq:DistrLaw_3_1_10}
      $\Sin{A} \mathrel{(\prox_{U})_{L}} \mathcal{V}$,
    \item\label{eq:DistrLaw_3_1_11}
      $\mathcal{V} \mathrel{(\prox_{U})_{L}} \mathcal{U}^{*}$.
  \end{enumerate}
  From  \ref{eq:DistrLaw_3_1_10}, \ref{eq:DistrLaw_3_1_11}, and Lemma
  \ref{lem:StarSingleton}~\eqref{lem:StarSingleton1},
  we have $\left\{ A \right\}^{*} \mathrel{(\prox_{U})_{L}} 
  \mathcal{U}^{*}$. Then
  $\left\{ A \right\} \mathrel{(\prox_{L})_{U}} \mathcal{U}$ 
  by Lemma \ref{lem:UpperLower}, or equivalently
  $A \prox_{L} \left\{ b \right\}$ and $A \prox_{L} \left\{ c,d \right\}$.
  Split $A$ into $C$ and $D$ so that $A = C \cup
  D$, $C \subseteq b \downarrow_{\prox} c$
  and $D \subseteq b \downarrow_{\prox} d$. 
  By putting $a_{1} = \medvee C$ and $a_{2} = \medvee D$ and using
  \ref{eq:DistrLaw_3_1_9} , we have
  $a \prox a_{1} \vee a_{2}$, $a_{1} \in b \downarrow_{\prox} c$ and
  $a_{2} \in b \downarrow_{\prox} d$. Hence, $S$ is localized.
\end{proof}

Proposition \ref{prop:LocalizedDoubleCoalg} also implies that the
property of being localized is invariant under isomorphisms
of strong proximity $\vee$-semilattices.

We are now ready to prove the main result of this subsection.
\begin{theorem}
  \label{thm:LocalizedPxSupLatEquivDoubleCoalg}
  $\PxJLatLoc$ is equivalent to the category of coalgebras of the double
  powerlocale over $\PxPos$.
\end{theorem}
\begin{proof}
  By Lemma \ref{lem:ProxSupLatLowerCoalg},
  Lemma \ref{prop:CharUpperCoalgHom},
  Lemma \ref{prop:CharHalgHom}, and
  Proposition \ref{prop:LocalizedDoubleCoalg},
  the category $\PxJLatLoc$ embeds into that of coalgebras of the double
  powerlocale. To see that the embedding is essentially surjective,
  let $S$ be a double powerlocale coalgebra. By Lemma~\ref{lem:LowerCoalgEssSurj},
  there is a strong proximity $\vee$-semilattice $S'$ which is
  isomorphic to $S$, so by Corollary \ref{cor:KZAlg}~\eqref{cor:KZAlg2},
  the isomorphism induces an isomorphism of the underlying
  $\UpperFunc$-coalgebras and $\LowerFunc$-coalgebras of $S$ and $S'$.
  Then, it is straightforward to show that the $\UpperFunc$-coalgebra
  and $\LowerFunc$-coalgebra structures on $S'$ satisfies
  the condition of Lemma~\ref{prop:CharHalg}. Thus, $S'$ is localized by
  Proposition \ref{prop:LocalizedDoubleCoalg}, which is isomorphic to
  $S$ as a double powerlocale coalgebra by
  Lemma~\ref{prop:CharHalgHom}.
\end{proof}

By Theorem \ref{thm:LocProxSuplattEquiLKLoc} and Theorem
\ref{thm:LocalizedPxSupLatEquivDoubleCoalg}, we have the following
characterisation of locally compact locales.
\begin{theorem}
  \label{thm:DoubleEquivLKLoc}
The category of coalgebras of the double powerlocale over $\PxPos$ is
equivalent to the category of locally compact locales.
\end{theorem}

\section{Continuous finitary covers}\label{sec:PresentingJoinProximitySemilattices}
We introduce a notion of continuous finitary cover, which provides a
logical characterisation of proximity $\vee$-semilattice.  
A continuous finitary cover can be thought of as a presentation of a
proximity $\vee$-semilattice where the underlying $\vee$-semilattice
is presented by generators and relations. The structure provides us
with a flexible way of constructing proximity $\vee$-semilattices
by generators and relations, and directly working on the
presentation.  Moreover, its strong variant provides a predicative and
even finitary alternative to the notion of continuous lattice in
formal topology~\cite{negri1998continuous}.
The left column of Table \ref{tab:CategoriesContFCov} shows some of
the major structures introduced in this section, with the
corresponding structures in domain theory and formal topology
on the middle and the right columns.

The development of this section can be seen as a suplattice analogue
of those of entailment systems~\cite{VickersEntailmentSystem} and
continuous entailment relations~\cite{KawaiContEnt}.
Similar to~\cite[Section~5]{VickersEntailmentSystem}, the underlying idea of
this section is that $\ContLatS$ can be described as the Karoubi 
envelop of its subcategory of free suplattices (cf.\ Section~\ref{sec:SEntsys}).

\begin{table}[t]
  \renewcommand{\arraystretch}{1.2}
\begin{tabular}{ccc}
  %
  %
  Predicative characterisation
  & Domain theoretic dual & Formal topology\\
  \hline
  \rule[4pt]{0pt}{14pt}
  \parbox{12em}{continuous finitary covers \\
                 + approximable maps} 
  & 
  $\ContLatS$
  & 
  ---
  \\
  \rule[6pt]{0pt}{14pt}
  \parbox{12em}{strong cont.\ fin.\ covers \\
  + join-approximable maps} 
  & 
  $\ContLat$
  & 
  \parbox{12em}{continuous basic covers \\
  + basic cover maps}
  \\
  \rule[6pt]{0pt}{14pt}
  \parbox{12em}{localized str.\ cont.\ fin.\
    cov.\ \\ + proximity maps}
  & $\LKFrm$
  & \parbox{12em}{locally compact formal top.\ \\
  + formal topology maps}
\end{tabular}
\caption{Main structures in Section~\ref{sec:PresentingJoinProximitySemilattices}.}
\label{tab:CategoriesContFCov}
\end{table}

\subsection{Semi-entailment systems}
\label{sec:SEntsys}
We begin with an observation that every continuous lattice
is a Scott continuous retract of the free suplattice over that lattice.
Recall that the \emph{free $\vee$-semilattice} over a set $S$ 
can be represented by $\Fin{S}$ ordered by inclusion, where
joins are computed by unions.
The \emph{free suplattice} over $S$ is the
ideal completion of $\Fin{S}$, or equivalently the power set of
$S$ (see Johnstone~\cite[Section C1.1, Lemma
1.1.3]{ElephantII}).

\begin{proposition}
  \label{prop:ContLatRetractFreeSLat}
  Every continuous lattice is a Scott continuous retract of a free
  suplattice.
\end{proposition}
\begin{proof}
  Let $L$ be a continuous lattice. Define
  functions
    $f  \colon  L \to \Ideals{\Fin{L}}$ and $g \colon \Ideals{\Fin{L}}
    \to L$ by 
  \begin{align*}
    f(a) &\defeql  \left\{ A \in \Fin{L} \mid \medvee A \ll a
  \right\}, &
    g(I) &\defeql \bigvee_{A \in I} \medvee A.
  \end{align*}
  It is easy to see that $f$ and $g$ are Scott continuous
  and $g \circ f = \id_{L}$.
\end{proof}
Let $\FreeSL$ be the full subcategory of $\ContLatS$ consisting of
free suplattices. Since every idempotent splits in $\ContLatS$
(cf.\ Proposition~\ref{prop:SplitCont}~\eqref{prop:SplitCont2}), we
have the following.
\begin{theorem}
  \label{thm:KaroubiFreeSL}
   $\Karoubi{\FreeSL}$ is equivalent to $\ContLatS$.
\end{theorem}

The dual of $\FreeSL$ can be identified with a full subcategory
of $\AlgPxPos$ consisting of free $\vee$-semilattices.
Then, the morphisms of this subcategory can be characterised in terms of generators
of free semilattices as follows.
\begin{definition}
  \label{def:UpperCut}
  Let $S$ and $S'$ be sets.
  A relation $r \subseteq S \times \Fin{S'}$ is \emph{upper} if
  \[
    a \mathrel{r} B \imp a \mathrel{r} (B \cup B').
  \]
  Given two upper relations
  $r \subseteq S \times \Fin{S'}$
  and $s \subseteq S' \times \Fin{S''}$,
  their \emph{cut composition} $s \cutcomp r \subseteq S \times
  \Fin{S''}$ is defined by
  \begin{align}
    \label{eq:CutComposition}
    a \mathrel{(s \cutcomp r)} C 
    &\defeqiv
    \exists B \in \Fin{S'} \left( a \mathrel{r} B \amp B
    \ApproxExt{s} C \right),\\
    \shortintertext{where}
    \label{eq:ApproxExt}
    B \ApproxExt{s} C &\defeqiv \forall b \in B \left( b
    \mathrel{s} C \right).
  \end{align}
\end{definition}
The following is analogous to
\cite[Proposition~25]{VickersEntailmentSystem}.
\begin{proposition}
  \label{prop:BijectCorrespondenceFreeSL}
  Let $\Fin{S}$ and $\Fin{S'}$ be free $\vee$-semilattices
  over sets $S$ and~$S'$.
  Then, there exists a bijective
  correspondence between approximable relations from $\Fin{S}$
  to $\Fin{S'}$ and upper relations from $S$ to $\Fin{S'}$.
  Via this correspondence, the identities and compositions
  of\, $\AlgPxPos$ correspond to the membership relation $\in$
  and cut compositions.
\end{proposition}
\begin{proof}
  The correspondence is as follows.
  An approximable relation $r \subseteq \Fin{S} \times \Fin{S'}$
  corresponds to an upper relation $\hat{r} \subseteq S \times
  \Fin{S'}$ defined by 
  \[
    a \mathrel{\hat{r}} B \defeqiv \left\{ a \right\}
    \mathrel{r} B.
  \]
  Conversely, an upper relation $r \subseteq S \times \Fin{S'}$ corresponds
  to an approximable relation ${\ApproxExt{r}} \subseteq \Fin{S} \times
  \Fin{S'}$ defined by \eqref{eq:ApproxExt}.
  Since finite joins in a free $\vee$-semilattice are given by unions, the
  above correspondence is well-defined and bijective.
  The second part is straightforward to check.
\end{proof}

Analogous to~\cite[Definition 27]{VickersEntailmentSystem}, let
$\SEnt$ denote the category in which objects are sets and morphisms
between sets $S$ and $S'$ are upper relations from $S$ to $\Fin{S'}$:
the identity on a set $S$ is the membership relation $\in$; the
composition of two upper relations is the cut composition.

By Proposition \ref{prop:BijectCorrespondenceFreeSL} and
Proposition~\ref{prop:AlgPxPos}, we have the following.
\begin{proposition}
  $\SEnt$ is dually equivalent to $\FreeSL$.
\end{proposition}
Thus, by Theorem~\ref{thm:KaroubiFreeSL}, we have another characterisation of
$\ContLatS$.
\begin{theorem}
  \label{thm:EquivSEntsysContLatS}
  $\Karoubi{\SEnt}$ is dually equivalent to $\ContLatS$.
\end{theorem}
The objects and morphisms of $\Karoubi{\SEnt}$ are 
called \emph{semi-entailment systems}
and \emph{approximable maps}, respectively.
For the record, we put down their explicit descriptions.
\begin{definition}
  A \emph{semi-entailment system} is a structure
  $(S, \ll)$, where ${\ll} \subseteq S \times \Fin{S}$ is an upper
  relation such that $\ll \cutcomp \ll$. 

  Let $(S, \ll)$ and $(S', \ll')$ be semi-entailment systems. An 
  \emph{approximable map} from $(S, \ll)$ to $(S', \ll')$
  is a relation $r \subseteq S \times \Fin{S'}$ such that 
  ${\ll \cutcomp r} = r = {\ll' \cutcomp r}$.
\end{definition}
Note that every approximable map is an upper relation.

In $\Karoubi{\SEnt}$, the identity on a semi-entailment system $(S,
\ll)$ is $\ll$; the composition of two approximable maps is the cut
composition.

\subsection{Continuous finitary covers}
We introduce another description of $\SEntsys$
which is closely related to $\PxJLat$.
First, we characterise a full subcategory of $\SEntsys$ which
corresponds to $\AlgLatS$.
As in Vickers~\cite[Section 6.1]{VickersEntailmentSystem}, 
we say that a semi-entailment system ${(S, \ll)}$ is \emph{reflexive} if
  \[
    a \in A \imp a \ll A.
  \]
Reflexive semi-entailment systems are
also known as finitary covers, or single
conclusion entailment relations~\cite{SchusterRinaldiWesselElimDisj}.
\begin{definition}
  A \emph{finitary cover} is a pair $(S, \fcov)$ where $S$ is a set 
  and $\fcov$
  is a relation between $S$ and $\Fin{S}$ such that
  \begin{align*}
    &\frac{a \in A}{a \fcov A} 
    &
    &\frac{a \fcov A}{a \fcov A \cup B} 
    &
    &\frac{a \fcov {A} \cup \left\{ b \right\} 
  \quad b \fcov A}{a \fcov A}.
  \end{align*}
We write $\FCov$ for the full subcategory of $\SEntsys$ consisting of
finitary covers. 
\end{definition}

Each finitary cover $(S, \fcov)$ determines a
$\vee$-semilattice $\CovtoLat{S,\fcov}$, which is the poset reflection of
$(\Fin{S}, \ApproxExt{\fcov})$ with finite joins computed by unions.%
\footnote{Specifically, two elements $A,B \in \Fin{S}$ are in the same
equivalence class if $A \ApproxExt{\fcov} B$ and $B \ApproxExt{\fcov}
A$.}
Also, each approximable map $r \colon  {(S,\fcov)} \to (S',\fcov')$
determines an approximable relation $\ApproxExt{r} \colon
\CovtoLat{S,\fcov} \to \CovtoLat{S',\fcov'}$ between the corresponding
$\vee$-semilattices. Indeed, since 
${\ApproxExt{\fcov'} \circ \ApproxExt{r}}
= {\ApproxExt{\fcov' \cutcomp r}} 
= {\ApproxExt{r}}$
and 
${\ApproxExt{r} \circ \ApproxExt{\fcov}}
= {\ApproxExt{r}}$, the relation $\ApproxExt{r}$ is well-defined and
satisfies \ref{def:approximableU}; since joins of $\CovtoLat{S,\fcov}$  are computed by unions, $\ApproxExt{r}$ also satisfies \ref{def:approximableI}.
It is easy to see that these assignments determine a full and faithful functor $F
\colon \FCov \to \AlgPxSLatS$.
\begin{proposition}
  \label{prop:EquivFCovPxJLat}
  The functor $F \colon \FCov \to \AlgPxSLatS$ is essentially
  surjective. Hence, $F$ determines an equivalence of $\FCov$ and
  $\AlgPxSLatS$.
\end{proposition}
\begin{proof}
  For each $\vee$-semilattice
  $(S,0,\vee)$, define a finitary cover $(S, \fcov_{\vee})$ by
  \begin{equation}
    \label{eq:CovFromPxSL}
    a \fcov_{\vee} A \defeqiv a \leq \medvee A.
  \end{equation}
  Define relations $r \subseteq S \times \Fin{S}$
  and $s \subseteq \Fin{S} \times S$ by
  \begin{align*}
    a \mathrel{r} A &\defeqiv a \leq \medvee A,
    & 
    A \mathrel{s} a &\defeqiv \medvee A  \leq a.
  \end{align*}
  It is straightforward to show that $r$ and $s$ are approximable
  relations between $(S, 0, \vee)$ and $\CovtoLat{S, \fcov_{\vee}}$
  and that they are inverse to each other.
\end{proof}
\begin{remark}
  \label{rem:QuasiInverseOfF}
  From the explicit constructions of the finitary
  cover $(S, \fcov_{\vee})$ and isomorphisms $r$ and $s$
  in the above proof,
  we can define a quasi-inverse $G \colon
  \AlgPxSLatS \to \FCov$ of $F$ as follows:
  \begin{align*}
    G(S,0,\vee) &\defeql (S, \fcov_{\vee}), \\
    a \mathrel{G(r)} A &\defeqiv a \mathrel{r} \medvee A.
  \end{align*}
\end{remark}
Since $\PxJLat = \Karoubi{\AlgPxSLatS}$ by definition, we have the following.
\begin{corollary}
  \label{cor:EquivContFCovPxJLat}
   $\Karoubi{\FCov}$ is equivalent to $\PxJLat$. 
\end{corollary}
The objects and morphisms of $\Karoubi{\FCov}$ 
can be explicitly described as follows.
\begin{definition}
  A \emph{continuous finitary cover} is a structure $(S, \fcov, \ll)$,
  where $(S, \fcov)$ is a finitary cover and ${\ll} \subseteq S \times
  \Fin{S}$ is a relation such that ${\ll \cutcomp \ll} = {\ll}$ and 
  ${\fcov \cutcomp \ll} = {\ll} = {\ll \cutcomp \fcov}$.

  Let $(S, \fcov, \ll)$ and $(S', \fcov', \ll')$ be continuous finitary
  covers. An \emph{approximable map} from $(S, \fcov, \ll)$ to $(S',
  \fcov', \ll')$ is a relation $r \subseteq S \times \Fin{S'}$ such that 
  ${\ll \cutcomp r} = r = {\ll' \cutcomp r}$.
\end{definition}

Henceforth, we write $\ContFCov$ for $\Karoubi{\FCov}$, which consists
of continuous finitary covers and approximable maps between them: the
identity on a continuous finitary cover $(S, \fcov, \ll)$ is $\ll$;
the composition of two approximable maps is the cut composition.

The equivalence of $\ContFCov$ and $\PxJLat$ in Corollary \ref{cor:EquivContFCovPxJLat}
is induced by the
functor $F \colon \FCov \to \AlgPxSLatS$ and its quasi-inverse $G \colon
\AlgPxSLatS \to \FCov$. Specifically, we have a pair of
functors $\overline{F} \colon \ContFCov \to \PxJLat$ 
and $\overline{G} \colon \PxJLat \to \ContFCov$ which act on
morphisms as $F$ and $G$, respectively, and on objects as follows:
\begin{align}
  \label{OverlineF}
  \overline{F}(S, \fcov, \ll) &\defeql (\CovtoLat{S, \fcov},
  \widetilde{\ll}),\\
  \notag
  \overline{G}(S, 0, \vee, \prox) &\defeql (S, \fcov_{\vee},
  \ll_{\vee}),\\
\shortintertext{where}
  \label{eq:llVee}
  a \mathrel{\ll_{\vee}} A 
  &\defeqiv a \prox \medvee A.
\end{align}
Thus, by Theorem~\ref{thm:SplitAlgPxSLatS}, we have the following.
\begin{theorem}
  \label{thm:ContFCovEquivContLatS}
  $\ContFCov$ is equivalent to $\ContLatS$.
\end{theorem}

On the other hand, 
each continuous finitary cover $(S, \fcov, \ll)$ determines  a
semi-entailment system $(S, \ll)$, and approximable maps
between two continuous finitary covers are precisely the approximable maps
between the corresponding semi-entailment systems. 
Conversely, each semi-entailment system $(S, \ll)$ can be regarded
as a continuous finitary cover $(S, \in, \ll)$. 
\begin{proposition}
  \label{prop:EquivContFCovSEntsys}
  $\ContFCov$ is equivalent to $\SEntsys$.
\end{proposition}
\begin{proof}
  As noted above, we have a pair of functors
  $P \colon \ContFCov \to \SEntsys$
  and 
  $Q \colon  \SEntsys \to \ContFCov$ whose actions on objects 
  are given by 
\begin{align*}
  P(S, \fcov, \ll) &\defeql (S, \ll),\\
  Q(S, \ll) &\defeql (S, \in, \ll),
\end{align*}
  and which are identity on morphisms.
  Obviously, $P \circ Q$ is an identity. Moreover,
  we have an approximable map $\ll \colon QP(S, \fcov, \ll) \to (S,
  \fcov, \ll)$, which is clearly isomorphic and natural in $(S, \fcov, \ll)$.
\end{proof}
In fact, Proposition~\ref{prop:EquivContFCovSEntsys} follows from
Theorem \ref{thm:EquivSEntsysContLatS} and
Theorem~\ref{thm:ContFCovEquivContLatS}. The point of this
proposition, however, is that its proof does not rely on the
impredicative notion of $\ContLatS$.

\subsection{Strong continuous finitary covers}
We characterise a subcategory of $\ContFCov$ which corresponds
to $\SPxJLat$.

\begin{definition}
  \label{def:FinCov}
  A \emph{strong continuous finitary cover}
  is a finitary cover $(S,\fcov)$ 
  together with an
  idempotent relation $\fprox$ on $S$ such that
  \begin{equation}
    \label{eq:ContFCov}
    \exists b \in S \left( a \fprox b \fcov A\right)  \leftrightarrow
    \exists B \in \Fin{S} \left(  a \fcov B \fprox_{L}  A \right).
  \end{equation}
  We write $(S,\fcov, \fprox)$ or simply $S$ for a strong continuous finitary cover.
\end{definition}
The condition \eqref{eq:ContFCov} is equivalent to
\[
  {\fcov \cutcomp {\fprox_{\exists}}}
  =
  {{\fprox_{\exists}} \cutcomp \fcov},
\]
where ${\fprox_{\exists}} \subseteq S \times \Fin{S}$ is defined by
  $
  a \fprox_{\exists} B \defeqiv \exists b \in B \left( a \fprox b
  \right). 
  $
Put
\begin{equation}
  \label{eq:ProxContFCovOnesided}
  {\ll_{\fcov}}
  \defeql {\fcov \cutcomp \fprox_{\exists}}
  = {{\fprox_{\exists}} \cutcomp \fcov}.
\end{equation}
Then, we have
\[
  \begin{aligned}
    {\ll_{\fcov} \cutcomp \ll_{\fcov}}
    &= {\fcov \cutcomp \fprox_{\exists} \cutcomp \fcov \cutcomp
    \fprox_{\exists}}
    = {\fprox_{\exists} \cutcomp \fcov \cutcomp \fcov \cutcomp
    \fprox_{\exists}}\\
    &= {\fprox_{\exists} \cutcomp \fcov \cutcomp \fprox_{\exists}}
    = {\fprox_{\exists} \cutcomp \fprox_{\exists} \cutcomp \fcov}
    = {\fprox_{\exists} \cutcomp \fcov}
    = {\ll_{\fcov}}.
  \end{aligned}
\]
Similarly,
we have
${\ll_{\fcov} \cutcomp \fcov}  
= {\ll_{\fcov}}
= {\ll_{\fcov} \cutcomp \fcov}$.
Thus, $(S, \fcov, \ll_{\fcov})$ is a continuous finitary cover.
With the above identification, strong continuous finitary covers form a full subcategory
of $\ContFCov$. 
In particular, each strong continuous finitary cover $(S, \fcov,
\fprox)$ determines a proximity $\vee$-semilattice $(\CovtoLat{S,
\fcov}, \widetilde{\ll_{\fcov}})$ as in  \eqref{OverlineF}.
\begin{lemma}
  \label{lem:SContFCovToSPxSLat}
  $(\CovtoLat{S, \fcov}, \widetilde{\ll_{\fcov}})$ is a 
  strong proximity $\vee$-semilattice.
\end{lemma}
\begin{proof}
  We show that $\ApproxExt{\ll_{\fcov}}$ satisfies
  \ref{def:approximable0} and \ref{def:approximableJ}.  For
  \ref{def:approximable0}, suppose $A \mathrel{\widetilde{\ll_{\fcov}}}
  \emptyset$. Then $A \ApproxExt{\fcov} \emptyset$ by
  \eqref{eq:ContFCov}, and so $A = 0$ in $\CovtoLat{S,\fcov}$.  For
  \ref{def:approximableJ}, suppose $A \mathrel{\widetilde{\ll_{\fcov}}}
  B \cup C$.
  Then, there exist $B',C' \in \Fin{S}$ such that $B'
  \mathrel{\widetilde{\fprox_{\exists}}} B$, $C'
  \mathrel{\widetilde{\fprox_{\exists}}} C$, and
  $A \ApproxExt{\fcov} B' \cup C'$.  Since $\fcov$ is reflexive, we have  $B'
  \mathrel{\widetilde{\ll_{\fcov}}} B$
  and $C' \mathrel{\widetilde{\ll_{\fcov}}} C$.
\end{proof}
In the opposite direction, each strong proximity $\vee$-semilattice
$(S,0,\vee, \prox)$ determines a strong continuous finitary cover
$(S, \fcov_{\vee}, \prox)$ where $\fcov_{\vee}$ is given
by~\eqref{eq:CovFromPxSL}.
Moreover, we have ${\ll_{\fcov_{\vee}}} = {\ll_{\vee}}$ where
$\ll_{\vee}$ is given by  \eqref{eq:llVee}. Thus, we have the
following.
\begin{proposition}
  \label{prop:EquivContFCovPxJLatRestrictToSContFCov}
The equivalence of
$\ContFCov$ and $\PxJLat$ restricts to their respective full  subcategories of strong
continuous finitary covers and strong proximity $\vee$-semilattices.
\end{proposition}

Next, we characterise approximable maps between strong
continuous finitary covers which correspond to join-approximable
relations.

\begin{definition}
  \label{def:ApproxRel}
  Let $(S, \fcov, \fprox)$ and $(S', \fcov', \fprox')$ be strong continuous
  finitary covers. An approximable map $r \colon (S,
  \fcov,\ll_{\fcov}) \to (S', \fcov', \ll_{\fcov'})$
  is  \emph{join-preserving} if
  \begin{equation}
      \label{def:approximableJoin}
      a \mathrel{r} B \imp \exists A \in \Fin{S}
      \left( a \fcov A \amp \forall a' \in A \, \exists b \in B \left(a' \mathrel{r} \left\{ b \right\}\right) \right).
  \end{equation}
  We call join-preserving approximable maps simply as
  \emph{join-approximable maps}. 
\end{definition}

\begin{lemma}\label{lem:ApproxToVmap}
  Let $(S, \fcov, \fprox)$ and $(S', \fcov', \fprox')$ be strong
  continuous finitary covers.
  An approximable map $r \colon (S, \fcov, \ll_{\fcov}) \to (S',
  \fcov', \ll_{\fcov'})$ is join-preserving if and only if
  the corresponding approximable relation $\ApproxExt{r} \colon
  (\CovtoLat{S, \fcov}, \widetilde{\ll_{\fcov}}) \to
  {(\CovtoLat{S', \fcov'}, \widetilde{\ll_{\fcov'}})}$
  is join-preserving.
\end{lemma}
\begin{proof}
  Suppose that $r$ is join-approximable. We must show
  \ref{def:approximable0}
  and
  \ref{def:approximableJ}.
  For \ref{def:approximable0}, suppose $A \ApproxExt{r} \emptyset$.
  Then \eqref{def:approximableJoin} implies $A \ApproxExt{\fcov} \emptyset$.
  For \ref{def:approximableJ}, suppose $A \ApproxExt{r} B \cup C$.
  By \eqref{def:approximableJoin}, for each $a \in A$, there  exist
  $B_{a}, C_{a} \in \Fin{S}$ such that $a \fcov
  B_{a} \cup C_{a}$ and that $\forall b' \in B_{a} \exists b
  \in B \left( b' \mathrel{r} \left\{ b \right\}\right)$
  and $\forall c' \in C_{a} \exists c
  \in C \left( c' \mathrel{r} \left\{ c \right\}\right)$.
  Put
  $B' = \bigcup_{a \in A}B_{a} $
  and 
  $C' = \bigcup_{a \in A}C_{a} $. Then, $A \ApproxExt{\fcov} B' \cup
  C'$,
  $B' \ApproxExt{r} B$, and $C' \ApproxExt{r} C$.
   
  Conversely, suppose that $\ApproxExt{r}$ is join-approximable, and
  let $a \mathrel{r} B$. Then, $\left\{ a \right\}
  \mathrel{\ApproxExt{r}} \bigcup_{b \in B} \left\{ b \right\}$. Since
  $\ApproxExt{r}$ is join-preserving, 
  there exist $A_{0},\dots,A_{n-1} \in \Fin{S}$ such that
  $a \fcov \bigcup_{i < n} A_{i}$
  and for each $i< n$, there exists $b\in B$ such that $A_{i}
  \mathrel{\widetilde{r}} \left\{ b \right\}$.  
\end{proof}

Let $\SContFCov$ be the category in which objects are strong continuous
finitary covers and morphisms are join-approximable
maps between underlying continuous finitary covers
of strong continuous finitary covers.
Then, by Proposition~\ref{prop:EquivContFCovPxJLatRestrictToSContFCov}
and Lemma~\ref{lem:ApproxToVmap}, the functor $\overline{F} \colon
\ContFCov \to \PxJLat$ restricts to a full and faithful functor
from $\SContFCov$ to $\SPxJLat$.
The following lemma implies that this restriction of $\overline{F}$ is essentially surjective.
\begin{lemma}
  \label{lem:IsoJoinPreserving}
  Every isomorphism in $\PxJLat$ between strong proximity
  $\vee$-semi\-lattices is 
  join-preserving.
\end{lemma}
\begin{proof}
  Let $(S,0, \vee, \prox)$ and $(S', 0', \vee', \prox')$
  be strong proximity $\vee$-semilattices, and $r \colon S \to S'$ and
  $s \colon S' \to S$ be approximable relations  such that $s \circ r = {\prox}$
  and $r \circ s = {\prox'}$.
  We show that $r$ is join-preserving.
  For \ref{def:approximable0}, suppose $a \mathrel{r} 0'$. 
  Since $0' \mathrel{s} 0$ by \ref{def:approximableI}, we have
  $a \mathrel{(s \circ r)}0$. Thus, $a \prox 0$ and so $a = 0$ 
  by the strength of $S$. For \ref{def:approximableJ}, suppose $a
  \mathrel{r} (b \vee' c)$. Since $r$ is approximable, there exists $d \prox' (b \vee' c)$ such that $a
  \mathrel{r} d$. Since $S'$ is strong, there exist $b' \prox' b$
  and $c' \prox' c$ such that  $d \leq (b' \vee' c')$. Thus, there exist
  $a_{b}, a_{c} \in S$ such that $b' \mathrel{s} a_{b} \mathrel{r} b$
  and $c' \mathrel{s} a_{c} \mathrel{r} c$. Then $a \mathrel{r} (b' \vee'
  c') \mathrel{s} (a_{b} \vee a_{c})$, and so $a \prox a_{b} \vee
  a_{c}$.  Since $S$ is strong, there exist $a',a'' \in S$ such that
  $a \leq a' \vee a''$, $a' \prox a_{b}$, and $a'' \prox a_{c}$. 
  Then $a' \mathrel{r} b$ and $a'' \mathrel{r} c$.
\end{proof}
\begin{theorem}
  \label{thm:EquivSContFCovSPxJLat}
  $\SContFCov$ is equivalent to $\SPxJLat$.
\end{theorem}

\subsection{Localized strong continuous finitary covers}\label{sec:LocalizedContFCov}
We characterise a subcategory of strong continuous finitary covers
which corresponds to $\PxJLatLoc$, the category of localized strong
proximity $\vee$-semilattices and proximity relations. 

In what follows, notations $a \downarrow_\fprox A$ and $A
\downarrow_\fprox B$ are defined similarly as in $a \downarrow_\prox A$ and $A
\downarrow_\prox B$, respectively.
\begin{definition}
  A strong continuous finitary cover $(S,\fcov, \fprox)$ is  \emph{localized} if 
  \[
    b \fprox a \fcov A \imp \exists B \in \Fin{a \downarrow_{\fprox} A} \left( b \fcov B\right).
  \]
\end{definition}
The following characterisation is more convenient.
\begin{lemma}
  \label{lem:ContFCovlocalized}
  A strong continuous finitary cover $(S,\fcov, \fprox)$ is localized if 
  and only if
  \[
    b \fprox a \fcov A \amp a \fcov B\imp 
    \exists C \in \Fin{A \downarrow_{\fprox} B} \left( b \fcov C
    \right).
  \]
\end{lemma}
\begin{proof}
  ``If'' part follows from $a \fcov \left\{ a \right\}$.  Conversely,
  suppose $b \fprox a \fcov A$ and $a \fcov B$.
  Choose $b' \in S$  such that $b \fprox b' \fprox a$. Since $S$ is
  localized, there exists $C \in \Fin{a \downarrow_{\fprox} A}$ such
  that $b' \fcov C$, and since $S$ is strong, there exists
  $D \fprox_{L} C$ such that $b \fcov D$.  Then, for each $d \in D$,
  there exists 
  $c \in a \downarrow_{\fprox} A$
  such that
  $d \prox c$, and since $S$ is strong, we
  find $B' \fprox_{L} B$ such
  that $c \fcov B'$.  Since $S$ is localized, there exists $E_{d} \in
  \Fin{c \downarrow_{\fprox} B'} \subseteq \Fin{A \downarrow_{\fprox}
B}$ such that  $d \fcov E_{d}$. Then, by putting $E = \bigcup_{d \in D}
E_{d}$, we have $b \fcov E \in \Fin{A \downarrow_{\fprox} B}$.
\end{proof}

\begin{lemma}
  \label{prop:ContFCovLocalizedProxSuplatLocalized}
  For any localized strong continuous finitary cover $(S,\fcov, \fprox)$,
   the corresponding strong proximity $\vee$-semilattice $(\CovtoLat{S,\fcov},
   \ApproxExt{\ll_{\fcov}})$ is localized.
\end{lemma}
\begin{proof}
  Suppose that $(S,\fcov, \fprox)$ is localized, and let
  $A \ApproxExt{\ll_{\fcov}} B \ApproxExt{\fcov} C \cup D$. 
  There exists $B'$ such that $A \fprox_{L} B' \ApproxExt{\fcov}
  B$ by \eqref{eq:ProxContFCovOnesided}, and so  $B' \ApproxExt{\fcov}
  C \cup D$.
  Write $A = \left\{ a_{0},\dots,a_{n-1} \right\}$. For each
  $i < n$, there exist $e_{i} \in S$ and $b_{i} \in B'$ such that $a_{i} \fprox
  e_{i} \fprox b_{i}$. Since $S$ is localized, there exists $A_{i}
  \in  \Fin{ C \cup D  \downarrow_{\fprox} b_{i}}$ such that $e_{i} \fcov A_{i}$.
  Split $A_{i}$ into $C_{i}$ and  $D_{i}$ so that $A_{i} =
  C_{i} \cup D_{i}$,
  $C_{i} \subseteq C \downarrow_{\fprox} b_{i}$, and 
  $D_{i} \subseteq D \downarrow_{\fprox} b_{i}$.
  Put $E = \left\{ e_{i} \mid i < n \right\}$, 
  $C' = \bigcup_{i < n}C_{i}$, and 
  $D' = \bigcup_{i < n}D_{i}$.
  Then, $A \fprox_{L} E \ApproxExt{\fcov} C' \cup D'$,
  $C' \subseteq C \downarrow_{\fprox} B'$, and 
  $D' \subseteq D \downarrow_{\fprox} B'$. Thus
  $A \ApproxExt{\ll_{\fcov}} C' \cup D'$,
  $C' \in C \downarrow_{\ApproxExt{\ll_{\fcov}}} B$, and $D' \in D
  \downarrow_{\ApproxExt{\ll_{\fcov}}} B$.
  Hence $(\CovtoLat{S,\fcov}, \ApproxExt{\ll_{\fcov}})$ is localized.
\end{proof}
\begin{definition}
  \label{def:ConFCovLawson}
  An approximable map $r \colon (S, \fcov, \fprox) \to (S',\fcov',
  \fprox')$ between strong continuous finitary covers  is 
  \emph{Lawson approximable} if
  \begin{enumerate}
    \item\label{def:ConFCovLawson1}
  $a \fprox a'\imp \exists B \in \Fin{S'} \left( a
      \mathrel{r} B \right)$,
    \item\label{def:ConFCovLawson2}
  $a \fprox a'  \mathrel{r} B \amp a' \mathrel{r} C
  \imp \exists D \in \Fin{S'} \left(D \ApproxExt{\ll_{\fcov'}} B \amp D
  \ApproxExt{\ll_{\fcov'}} C \amp
      a \mathrel{r} D \right)$.
  \end{enumerate}
\end{definition}
\begin{remark}
  \label{rem:ConFCovLawson}
  If $S'$ is localized, then the
  condition \ref{def:ConFCovLawson2} can be strengthened to
  \[
    a \fprox  a' \mathrel{r} B \amp a' \mathrel{r} C
    \imp \exists D \in \Fin{B \downarrow_{\fprox'} C} \left( a \mathrel{r} D \right).
  \]
  If $S$ and $S'$ are both localized and $r$ is
  join-preserving, then the
  condition \ref{def:ConFCovLawson2} can be simplified further to
  \begin{equation}
    \label{eq:LawsonLocalized}
    a \fprox  a' \mathrel{r} \left\{ b \right\} \amp a' \mathrel{r} \left\{ c \right\}
    \imp \exists D \in \Fin{b \downarrow_{\fprox'} c} \left( a \mathrel{r} D \right).
  \end{equation}
\end{remark}
\begin{lemma}
  \label{lem:EquivalenceOfLawson}
  An approximable map $r \colon (S, \fcov, \fprox) \to (S',\fcov', \fprox')$
  is Lawson approximable if and only if the corresponding approximable relation
  $\ApproxExt{r} \colon  {(\CovtoLat{S,\fcov},
  \ApproxExt{\ll_{\fcov}})} \to
  (\CovtoLat{S',\fcov'}, \ApproxExt{\ll_{\fcov'}})$ 
  is Lawson approximable.
\end{lemma}
\begin{proof}
  Suppose that $r$ is Lawson approximable, and let $A
  \ApproxExt{\ll_{\fcov}} A'$.
  Then, there exists $A'' \in
  \Fin{S}$ such that $A \ApproxExt{\fcov} A'' \fprox_{L} A'$.  Since
  $r$ is Lawson approximable, there exists $B \in \Fin{S'}$ such that
  $A'' \ApproxExt{r} B$. Then $A \mathrel{\ApproxExt{r}} B$.
  Next, suppose  $A \ApproxExt{\ll_{\fcov}} A'
  \mathrel{\ApproxExt{r}}
  B$ and $A' \ApproxExt{r} C$. Then, there exists $A'' \in
  \Fin{S}$ such that $A \ApproxExt{\fcov} A'' \fprox_{L} A'$.  Since
  $r$ is Lawson approximable, there exists $D \in \Fin{S'}$ such that
  $A'' \ApproxExt{r} D$, $D \ApproxExt{\ll_{\fcov'}} B$, and $D \ApproxExt{\ll_{\fcov'}} C$. 
  Then, $A \mathrel{\ApproxExt{r}} D$.
  
  Conversely, suppose that $\ApproxExt{r}$ is Lawson
  approximable. If $a \fprox a'$, then $\left\{ a \right\}
  \ApproxExt{\ll_{\fcov}} \left\{ a'
  \right\}$. Since $\ApproxExt{r}$ is Lawson approximable, there exists $B \in \Fin{S}$
  such that $\left\{ a \right\} \mathrel{\ApproxExt{r}} B$, and so 
  $a \mathrel{r} B$. Next, suppose $a \fprox a' \mathrel{r} B$ and
  $a' \mathrel{r} C$. Then, $\left\{ a \right\}
  \ApproxExt{\ll_{\fcov}}
  \left\{ a' \right\} \mathrel{\ApproxExt{r}} B$,
  and $\left\{ a' \right\} \mathrel{\ApproxExt{r}} C$. Since
  $\ApproxExt{r}$ is Lawson approximable, there exists $D \in \Fin{S'}$ such that
  $D \ApproxExt{\ll_{\fcov'}} B$, $D \ApproxExt{\ll_{\fcov'}} C$, and $\left\{ a \right\} \mathrel{\ApproxExt{r}} D$.
  Then $a \mathrel{r} D$.
\end{proof}
\begin{definition}
  \label{def:ProxMap}
  A join-preserving Lawson approximable map between
  localized strong continuous finitary covers is called a \emph{proximity map}.
\end{definition}
Let $\ContFCovLoc$ be the category of localized strong continuous
finitary covers and proximity maps.
\begin{proposition}
  The equivalence between $\SContFCov$ and $\SPxJLat$ restricts to an
  equivalence between $\ContFCovLoc$ and $\PxJLatLoc$.
\end{proposition}
\begin{proof}
  By Lemma \ref{prop:ContFCovLocalizedProxSuplatLocalized} and Lemma
  \ref{lem:EquivalenceOfLawson}, the embedding of $\SContFCov$ into
  $\SPxJLat$ restricts to an embedding of $\ContFCovLoc$ into
  $\PxJLatLoc$.  Moreover, for any localized strong proximity
  $\vee$-semilattice $(S,0,\vee, \prox)$, the strong continuous
  finitary cover $(S, \fcov_{\vee}, \prox)$ defined in
  \eqref{eq:CovFromPxSL} is localized by
  Lemma~\ref{lem:EquivalenceLocalization} and
  Lemma~\ref{lem:ContFCovlocalized}.
  As in Lemma \ref{lem:IsoJoinPreserving}, one can also show that every
  isomorphism in $\PxPos$ is Lawson approximable. 
  Thus, the embedding of $\ContFCovLoc$ into $\PxJLatLoc$ is essentially
  surjective.
\end{proof}

\subsection{Formal topologies}\label{sec:FTop}
We relate the theory of strong continuous finitary covers to that of
continuous lattices and locally compact locales in formal
topology~\cite{negri1998continuous}.

\begin{definition}
  \label{def:ContBCov}
  A \emph{basic cover} is a pair $(S, \cov)$ where
  $\cov$ is a relation between $S$ and $\Pow{S}$
  satisfying
  \begin{gather*}
    \frac{a \in U}{a \cov U} \; \text{(reflexivity)}
    \qquad
    \qquad
    \frac{a \cov U \quad U \cov V}{a \cov V} \;\text{(transitivity)}
  \end{gather*}
  where $U \cov V \defeqiv \forall a \in U \left( a \cov V \right)$.
  A basic cover $(S, \cov)$ is \emph{continuous} 
  if it is equipped with a relation $\wb \subseteq S \times S$
  satisfying
  \begin{equation}
    \label{eq:wb}
    a \cov \wb^{-}a
    \qquad
    \qquad
    \frac{b \mathrel{\wb} a \quad a \cov U}{\exists A \in
      \Fin{U}\;  b \cov A}.
    \end{equation}
\end{definition}

Among the various definitions of formal topology described in
\cite{ConvFTop}, we prefer to work with the one with a preorder.
\begin{definition}
  A \emph{formal topology} is a triple $(S,\cov,\leq)$ where
  $(S, \cov)$ is a basic cover and $(S, \leq)$ is a preorder
   satisfying
   \[
     \frac{a \leq b}{a \cov \{b\}} \; \text{($\leq$-left)}
     \qquad
     \qquad
     \frac{a \cov U \quad a \cov V}{a \cov U \downarrow_{\leq} V}
     \;\text{($\downset$-right)}
   \]
  where $U \downarrow_{\leq} V \defeql \downset_{\leq} U \cap
  \downset_{\leq} V$.
  A formal topology $(S, \cov, \leq)$ is \emph{locally compact} if 
  $(S, \cov)$ is a continuous basic cover.
\end{definition} 

A subset $U \subseteq S$ of a basic cover $(S, \cov)$ is  \emph{saturated} if 
 \[
   a \cov U \imp a \in U.
 \]
Negri \cite{negri1998continuous} showed that the collection of
saturated subsets of a continuous basic cover and that of a locally
compact formal topology are continuous lattice and locally compact
locale, respectively. She also showed that every continuous lattice
and locally compact locale can be represented in this way.%
\footnote{In Negri \cite{negri1998continuous}, continuous basic covers
and locally compact formal topologies are called \emph{locally Stone
infinitary preorders} and \emph{locally Stone formal topologies},
respectively.}

\begin{proposition}
  \label{prop:FromContFCovToBasicCover}
  Let $(S, \fcov, \fprox)$
  be a  strong continuous finitary cover. Define
  a relation ${\fcov_{\fprox}} \subseteq S \times \Pow{S}$ by
  \begin{equation}
    \label{eq:CovFromContFCov}
    a \mathrel{\fcov_{\fprox}}  U \defeqiv \forall b \fprox a \, \exists B
    \in \Fin{U}\left( b \ll_{\fcov}B \right).
  \end{equation}

  \begin{enumerate}
    \item \label{prop:FromContFCovToBasicCover1}
      The structure $(S, \fcov_{\fprox})$, 
      is a continuous basic cover with respect to $\fprox$.

    \item \label{prop:FromContFCovToBasicCover2}
      Let $\fproxeq$ be the reflexive closure of $\fprox$.  Then, $(S,
      \fcov,
      \fprox)$ is localized if and only if the triple $(S, \fcov_{\fprox},
      \fproxeq)$ is a locally compact formal topology.
  \end{enumerate}
\end{proposition}
\begin{proof}
  \noindent
  \ref{prop:FromContFCovToBasicCover1}.  This is straightforward to check.
  \smallskip
   
  \noindent
  \ref{prop:FromContFCovToBasicCover2}.  
  Suppose that $(S, \fcov, \fprox)$ is localized.
  Since ${\fprox \circ \fproxeq} \subseteq {\fprox}$, 
  $(S, \fcov_{\fprox}, \fproxeq)$ satisfies ($\leq$-left).
  As for ($\downset$-right),
  suppose $a
  \mathrel{\fcov_{\fprox}} U$ and $a \mathrel{\fcov_{\fprox}} V$.
  Let $b \fprox a$, and choose $c,d \in S$ such that $b \fprox c \fprox
  d \fprox a$.
  Then, there exist $A \in \Fin{U}$ and $B \in \Fin{V}$ such that $d
  \ll_{\fcov} A$ and $d \ll_{\fcov} B$. 
  Thus, there exist $A' \fprox_{L} A$ and $B' \fprox_{L} B$
  such that $d \fcov A'$ and $d \fcov B'$.
  By Lemma \ref{lem:ContFCovlocalized}, there
  exists $C \in \Fin{A' \downarrow_{\fprox} B'} \subseteq 
  \Fin{A \downarrow_\fprox B}$ such that $c \fcov C$. Then $b \ll_{\fcov}
  C$, and so $a \mathrel{\fcov_{\fprox}} U \downarrow_{\fproxeq} V$.
  Conversely, suppose that $\fcov_{\fprox}$ satisfies
  ($\leq$-left) and ($\downset$-right) with respect to $\fproxeq$. Let
  $b \fprox
  a \fcov A$ and $a \fcov B$. Since ${\fcov} \subseteq {\fcov_{\fprox}}$, we
  have $a \mathrel{\fcov_{\fprox}} A
  \downarrow_{\fproxeq} B$ by ($\downset$-right), so there exists $C \in \Fin{A
    \downarrow_{\fproxeq} B}$ such that
  $b \ll_{\fcov} C$. Then, there exists $C' \in  \Fin{A
    \downarrow_{\fprox} B}$ such that $b \fcov C'$.
\end{proof}
Next, we recall the notions of morphism between basic covers and formal
topologies.
\begin{definition}
  Let $(S, \cov)$ and $(S', \cov')$ be basic covers.
  A relation $r \subseteq S \times S'$ is a \emph{basic cover
  map} from $(S, \cov)$ to $(S', \cov')$ if 
  \begin{enumerate}[({BCM}1)]
    \item\label{BCM1} $a \cov r^{-} b \imp a \mathrel{r} b$,
    \item\label{BCM2} $b \cov' V \imp r^{-} b  \cov r^{-}V$.
  \end{enumerate}
Basic covers and basic cover maps form a category $\BCov$. 
The identity $\id_{(S,\cov)}$ on a basic cover
$(S,\cov)$ is defined by 
\[
  a \mathrel{\id_{(S,\cov)}} b \defeqiv a \cov \left\{ b
  \right\},
\]
and the composition $s * r$
of basic cover maps $r \colon S \to S'$ and
$s \colon S' \to S''$  is defined by 
\[
  a \mathrel{(s * r)} c \defeqiv a \cov r^{-}s^{-} c.
\]

  A \emph{formal topology map} between formal topologies
   $(S, \cov, \leq)$ and $(S', \cov', \leq')$ 
  is a basic cover map $r \colon (S, \cov) \to (S', \cov')$
  such that 
  \begin{enumerate}[({FTM}1)]
    \item\label{FTM1} $S \cov r^{-}S'$,
    \item\label{FTM2} $r^{-} a \downarrow_{\leq} r^{-} b \cov r^{-}(a
      \downarrow_{\leq'} b)$.
  \end{enumerate}
Formal topologies and formal topology maps form a subcategory
$\FTop$ of $\BCov$.
\end{definition}
In what follows, we extend
Proposition~\ref{prop:FromContFCovToBasicCover} to morphisms, i.e., we
show that there is a bijective correspondence between
join-approximable maps and basic cover maps, or in the case of localized
strong continuous finitary covers, between
proximity maps and formal topology maps (cf.\
Proposition~\ref{prop:BijectionBCovMapWbrel}).
The correspondence is mediated by an auxiliary notion given in
the next lemma.
\begin{lemma}
  \label{lem:SingleVersionJoinPresAppMap}
  Let $(S, \fcov, \fprox)$ and ${(S'\fcov', \fprox')}$
  be strong continuous finitary covers.
  \begin{enumerate}
    \item \label{lem:SingleVersionJoinPresAppMapA}
    There exists a bijective correspondence between
    join-approximable maps from $S$ to $S'$ and relations $r
    \subseteq S \times S'$ satisfying
    \begin{enumerate}[label=(\alph*)]
      \item\label{lem:SingleVersionJoinPresAppMap1}
        $a \mathrel{r} b \leftrightarrow \exists a'  \in S\left(
        a \fprox a' \mathrel{r} b \right)$,

      \item\label{lem:SingleVersionJoinPresAppMap2}
        $a \fcov A \subseteq r^{-} b  \imp  a \mathrel{r} b$,

      \item\label{lem:SingleVersionJoinPresAppMap3}
        $a \mathrel{r} b \fcov' B \imp \exists A \in \Fin{S} \left(
        a \fcov A \mathrel{r_{L}} B \right)$,

      \item\label{lem:SingleVersionJoinPresAppMap4}
        $a \mathrel{r} b \fprox' b' \imp a \mathrel{r} b'$,

      \item\label{lem:SingleVersionJoinPresAppMap5}
        $a \mathrel{r} b \imp \exists A \in \Fin{S} \, \exists B \in
        \Fin{S'} \left( a \fcov A \mathrel{r_{L}}  B \fprox_{L}' \left\{
        b \right\}\right)$.
    \end{enumerate}

    \item \label{lem:SingleVersionJoinPresAppMapB}
    If $(S, \fcov, \fprox)$ and  $(S', \fcov', \fprox')$
    are localized, the above bijection restricts to a bijection between
    proximity maps and relations between $S$ and $S'$
    satisfying
    \ref{lem:SingleVersionJoinPresAppMap1}--\ref{lem:SingleVersionJoinPresAppMap5}
    above and the following properties:
    \begin{enumerate} [resume*]
      \item\label{lem:SingleVersionJoinPresAppMap6}
         $a \fprox a' \imp \exists A \in \Fin{S} \, \exists B \in
          \Fin{S'} (a \fcov A \mathrel{r_{L}} B)$,

      \item\label{lem:SingleVersionJoinPresAppMap7}
        $a \fprox a' \mathrel{r} b \amp a' \mathrel{r} c \imp
         \exists A \in \Fin{S} \, \exists D \in \Fin{b \downarrow_{\fprox'} c}
         \left( a \fcov A \mathrel{r}_{L} D \right)$.
    \end{enumerate}
  \end{enumerate}
\end{lemma}
\begin{proof}
  \ref{lem:SingleVersionJoinPresAppMapA}.
   If $r \colon (S, \fcov, \fprox) \to (S', \fcov', \fprox')$ is a
   join-approximable  map, then the relation $r^{\dagger}
   \subseteq S \times S'$ defined by
  \begin{equation}
    \label{def:dagger}
    a \mathrel{r^{\dagger}} b \defeqiv a \mathrel{r} \left\{ b
    \right\}
  \end{equation}
  satisfies
  \ref{lem:SingleVersionJoinPresAppMap1}--\ref{lem:SingleVersionJoinPresAppMap5}.
  Conversely, given a relation $r \subseteq S \times S'$ satisfying
  \ref{lem:SingleVersionJoinPresAppMap1}--\ref{lem:SingleVersionJoinPresAppMap5},
  define a relation $r^{*} \subseteq S \times \Fin{S'}$ by
  \begin{equation}
    \label{def:star}
    a \mathrel{r^{*}} B 
    \defeqiv 
    \exists A \in \Fin{S} \left( a \fcov A \mathrel{r_{L}} B \right).
  \end{equation}
  Clearly, $r^{*}$ satisfies
  \eqref{def:approximableJoin}. 
  Moreover, one can easily show 
  ${\mathrel{r^{*}} \cutcomp {\ll_{\fcov}}} = r^{*}= {{\ll_{\fcov'}}
  \cutcomp \mathrel{r^{*}}}$ using 
  \ref{lem:SingleVersionJoinPresAppMap1},
  \ref{lem:SingleVersionJoinPresAppMap3},
  \ref{lem:SingleVersionJoinPresAppMap4},
  and 
  \ref{lem:SingleVersionJoinPresAppMap5}.
  Thus $r^{*}$ is a join-approximable
  map from $(S, \fcov, \fprox)$ to $(S', \fcov', \fprox')$.
  Then, it is straightforward to check that above correspondence is
  bijective (note that the fact $(r^{*})^{\dagger} = r$ requires~\ref{lem:SingleVersionJoinPresAppMap2}).
  \smallskip

  \noindent\ref{lem:SingleVersionJoinPresAppMapB}.
  Suppose that $(S, \fcov, \fprox)$ and $(S',
  \fcov', \fprox')$ are localized. First, if
  $r \colon {(S, \fcov, \fprox)} \to {(S', \fcov', \fprox')}$ is a
  proximity map, then
  the relation $r^{\dagger}$ given by \eqref{def:dagger}
  satisfies~\ref{lem:SingleVersionJoinPresAppMap6}
  and~\ref{lem:SingleVersionJoinPresAppMap7} by the fact that $r$ is
  Lawson and join-approximable (cf.\ Remark~\ref{rem:ConFCovLawson}).
  Conversely, let $r \subseteq S \times S'$ be a relation satisfying
  \ref{lem:SingleVersionJoinPresAppMap1}--\ref{lem:SingleVersionJoinPresAppMap7}.
  Then, the relation~$r^{*}$
  given by \eqref{def:star} satisfies the first property of
  Lawson approximable map by~\ref{lem:SingleVersionJoinPresAppMap6}. For the second property, it
  suffices to show \eqref{eq:LawsonLocalized}. Suppose
  $a \fprox a' \mathrel{r^{*}} \left\{ b \right\}$ and $a'
  \mathrel{r^{*}} \left\{ c \right\}$.
  Then, $a' \mathrel{r} b$ and $a' \mathrel{r} c$ by
  \ref{lem:SingleVersionJoinPresAppMap2}. Thus,
  there exists $D \in {\Fin{b \downarrow_{\fprox'}c}}$ such that
  $a \mathrel{r^{*}} D$ by~\ref{lem:SingleVersionJoinPresAppMap7}.
\end{proof}

\begin{proposition}
  \label{prop:BijectionBCovMapWbrel}
  Let $(S, \fcov, \fprox)$ and ${(S'\fcov', \fprox')}$ be strong
  continuous finitary covers.

  \begin{enumerate}
    \item\label{prop:BijectionBCovMapWbrel1} There exists a bijective
      correspondence between join-approximable maps from
      $(S, \fcov, \fprox)$ and $(S', \fcov', \fprox')$ and basic cover
      maps from $(S, \fcov_{\fprox})$ to $(S', \fcov'_{\fprox'})$.

    \item\label{prop:BijectionBCovMapWbrel2} If $(S, \fcov, \fprox)$ and
      $(S'\fcov', \fprox')$ are localized, then the above bijection restricts
      to a bijection between proximity maps  and formal topology maps.
  \end{enumerate}
\end{proposition}
\begin{proof}
  \noindent \ref{prop:BijectionBCovMapWbrel1}.
  We establish a bijection between basic cover
  maps from $(S, \fcov_{\fprox})$ to $(S', \fcov'_{\fprox'})$
  and relations between $S$ and $S'$ satisfying
  \ref{lem:SingleVersionJoinPresAppMap1}--\ref{lem:SingleVersionJoinPresAppMap5}
  in Lemma \ref{lem:SingleVersionJoinPresAppMap}.
  In what follows, we take Proposition~\ref{prop:FromContFCovToBasicCover} as given.

  First, for a relation $r \subseteq S \times S'$ satisfying 
  \ref{lem:SingleVersionJoinPresAppMap1}--\ref{lem:SingleVersionJoinPresAppMap5}, 
  define $\widebreve{r} \subseteq S \times S'$ by 
  \begin{equation}
    \label{def:breve}
    a \mathrel{\widebreve{r}} b \defeqiv a \mathrel{\fcov_{\fprox}}
    r^{-} b.
  \end{equation}
  We show that $\widebreve{r}$
  satisfies \ref{BCM1} and \ref{BCM2}.
  \smallskip
  
  \noindent\ref{BCM1} This follows from (transitivity) of $\fcov_{\fprox}$.
  \smallskip

  \noindent \ref{BCM2}  Suppose $b \mathrel{\fcov'_{\fprox'}} V$ and 
  $a \mathrel{\widebreve{r}} b$. Since $a \mathrel{\fcov_{\fprox}} r^{-}b$, 
  we may assume $a \mathrel{r} b$ by (transitivity).
  By~\ref{lem:SingleVersionJoinPresAppMap5},
  there exists $A \in \Fin{S}$ and $B \in \Fin{S'}$ such that $a \fcov A
  \mathrel{r_{L}} B \fprox_{L}' \left\{ b \right\}$.
  Since $b \fcov'_{\fprox'} V$, there exist $B',B'' \in \Fin{S'}$
  such that $B \ApproxExt{\fcov'} B' \fprox_{L}' B'' \subseteq V$.
  By \ref{lem:SingleVersionJoinPresAppMap3}, there exists $A' \in
  \Fin{S}$ such that $A \ApproxExt{\fcov} A' \mathrel{r_{L}} B'$,
  and so $A' \mathrel{r_{L}} B''$ by \ref{lem:SingleVersionJoinPresAppMap4}.
  Then, $a \fcov A' \in \Fin{\widebreve{r}^{-}V}$ so that $a
  \fcov_{\fprox} \widebreve{r}^{-}V$. Hence, $\widebreve{r}^{-}b
  \fcov_{\fprox} \widebreve{r}^{-}V$.

  Conversely, for a basic cover map $r \colon (S, \fcov_{\fprox}) \to
  (S', \fcov'_{\fprox'})$, define a relation $r^{\ddagger} \subseteq S
  \times S'$ by
  \begin{equation}
    \label{def:ddagger}
    a \mathrel{r^{\ddagger}} b \defeqiv a \in \downset_{\fprox}
    r^{-}b.
  \end{equation}
  We show that $r^{\ddagger}$ satisfies 
  \ref{lem:SingleVersionJoinPresAppMap1}--\ref{lem:SingleVersionJoinPresAppMap5}
  in Lemma \ref{lem:SingleVersionJoinPresAppMap}. 
  \smallskip

  \noindent\ref{lem:SingleVersionJoinPresAppMap1}
  This follows from the idempotency of $\fprox$.
  \smallskip

  \noindent\ref{lem:SingleVersionJoinPresAppMap2}
  Let $a \fcov A \subseteq (r^{\ddagger})^{-}b$. Then, there exists 
  $A' \in \Fin{r^{-}b}$ such that
  $A \fprox_{L} A'$. Since $S$ is strong, there exists $a'$ such that
  $a \fprox a' \fcov A'$. Then, $a' \fcov_{\fprox} A'$ and so $a'
  \mathrel{r} b$ by \ref{BCM1}. Hence $a \mathrel{r^{\ddagger}} b$. 
  \smallskip

  \noindent\ref{lem:SingleVersionJoinPresAppMap3}
  Let $a \mathrel{r^{\ddagger}} b \fcov' B$.
  Then, there exists $a'$ such that $a \fprox a' \mathrel{r} b$.
  Since $b \fcov'_{\fprox'} B$, we have $r^{-}b \fcov_{\fprox} r^{-}B$
  by \ref{BCM2}, so there exists $A \in \Fin{r^{-}B}$ such that 
  $a \ll_{\fcov} A$. Thus, there exists $A' \fprox_{L} A$
  such that $a \fcov A'$ and $A' \mathrel{(r^{\ddagger})_{L}} B$.
  \smallskip

  \noindent\ref{lem:SingleVersionJoinPresAppMap4}
  Let $a \mathrel{r^{\ddagger}} b \fprox' b'$. Then, there exists $a'$
  such that $a \fprox a' \mathrel{r} b$. Since $b \fprox' b'$ implies
  $b \fcov'_{\fprox'} \left\{ b' \right\}$, we have $r^{-}b \fcov_{\fprox}
  r^{-} b'$ by \ref{BCM2}. Then $a' \mathrel{r} b'$ by \ref{BCM1}
  so that $a \mathrel{r^{\ddagger}} b'$.
  \smallskip

  \noindent\ref{lem:SingleVersionJoinPresAppMap5}
  Let $a \mathrel{r^{\ddagger}} b$. Then, there exists $a'$ such that
  $a \fprox a' \mathrel{r} b$. Since 
  $b \fcov'_{\fprox'} \downset_{\fprox'} b$, we have $a' \fcov_{\fprox}
  r^{-} \downset_{\fprox'} b$ by \ref{BCM2}. Thus, there exists $A \in
  \Fin{r^{-}\downset_{\fprox'} b}$ such that $a \ll_{\fcov} A$, and so
  there exist $A' \in \Fin{S}$ and $B \fprox_{L}' \left\{ b \right\}$
  such that $a \fcov A' \fprox_L A$ and $A \mathrel{r_L} B$.
  Then, $A' \mathrel{(r^{\ddagger})_L} B$.
  \smallskip
 
  Next, we show that the above correspondence is bijective.
  First, note that 
  \begin{equation}
    \label{eq:SaturationDownset}
    a \fcov_{\fprox} U \iff a \fcov_{\fprox} \downset_{\fprox} U 
  \end{equation}
  for each $a \in S$ and $U \subseteq S$. Thus, for any
  basic cover map $r \colon {(S, \fcov_{\fprox})} \to
  {(S', \fcov'_{\fprox'})}$, we have
  \begin{align*}
    a \mathrel{\widebreve{r^{\ddagger}}} b 
    &\iff 
    a \fcov_{\fprox} (r^{\ddagger})^{-} b\\
    &\iff 
    a \fcov_{\fprox} r^{-} b
    && \text{(by \eqref{eq:SaturationDownset})} \\
    &\iff 
    a \mathrel{r} b
    && \text{(by \ref{BCM1})}.
  \end{align*}
  Conversely, for any 
relation $r \subseteq S \times S'$ satisfying 
  \ref{lem:SingleVersionJoinPresAppMap1}--\ref{lem:SingleVersionJoinPresAppMap5},
  we have
  \begin{align*}
    a \mathrel{(\widebreve{r})^{\ddagger}} b 
    &\iff
    \exists a' \in S \left( a \fprox a' \fcov_{\fprox} r^{-}b \right)\\
    &\iff
    \exists A \in \Fin{r^{-}b} \, a \ll_{\fcov} A
    &&\text{(by ${\fcov} \subseteq {\fcov_{\fprox}}$)}\\
    &\iff
    \exists A \in \Fin{r^{-}b} \, a \fcov A
    &&\text{(by \ref{lem:SingleVersionJoinPresAppMap1})}\\
    &\iff
    a \mathrel{r} b
    &&\text{(by \ref{lem:SingleVersionJoinPresAppMap2})}.
  \end{align*}

  \noindent \ref{prop:BijectionBCovMapWbrel2}.
  Suppose that $(S, \fcov, \fprox)$ and $(S', \fcov', \fprox')$ 
  are localized. First, let $r \subseteq S \times S'$ be a relation satisfying
  \ref{lem:SingleVersionJoinPresAppMap1}--\ref{lem:SingleVersionJoinPresAppMap7}
  in Lemma \ref{lem:SingleVersionJoinPresAppMap}.
  We must show that $\widebreve{r}$ given by~\eqref{def:breve} satisfies 
  \ref{FTM1} and \ref{FTM2}.
  \smallskip

  \noindent \ref{FTM1} Let $a \in S$ and $a' \fprox a$. Choose $a''
  \in S$ such that $a' \fprox a'' \fprox a$. 
  By~\ref{lem:SingleVersionJoinPresAppMap6}, there exist $A \in
  \Fin{S}$ and $B \in \Fin{S'}$ such that $a'' \fcov A \mathrel{r_{L}} B$.
  Then, $a' \ll_{\fcov} A$ and $A \subseteq (\widebreve{r})^{-} B$.
  Thus $a \fcov_{\fprox} (\widebreve{r})^{-} B \subseteq
  (\widebreve{r})^{-} S'$, and hence $S \fcov_{\fprox} (\widebreve{r})^{-} S'$.
  \smallskip

  \noindent \ref{FTM2}
  Let $a \in \widebreve{r}^{-}b
  \downarrow_{\fproxeq}\widebreve{r}^{-}c$. By ($\leq$-left), we
  have $a \fcov_{\fprox} r^{-}b$ and $a \fcov_{\fprox} r^{-}c$, and so
  $a \fcov_{\fprox} r^{-}b \downarrow_{\fproxeq}
  r^{-}c$ by ($\downset$-right). By (transitivity), it suffices to show
  $r^{-}b \downarrow_{\fproxeq} r^{-}c 
  \fcov_{\fprox} (\widebreve{r})^{-}\left(  b \downarrow_{\fproxeq'} c  \right)$.
  Let $a' \in r^{-}b \downarrow_{\fproxeq} r^{-}c$, and $a'' \fprox a'$. 
  Choose $a''' \in S$ such that $a'' \fprox a''' \fprox a'$.
  We have 
  $r^{-}b \downarrow_{\fproxeq} r^{-}c = r^{-}b \cap r^{-}c$
  by~\ref{lem:SingleVersionJoinPresAppMap1}, so by applying
  \ref{lem:SingleVersionJoinPresAppMap7}, we find
  $A \in \Fin{S}$ and $D \in \Fin{b \downarrow_{\fprox'} c}$ 
  such that $a''' \fcov A \mathrel{r_{L}} D$. 
  Then, $a'' \ll_{\fcov} A$ and $A \subseteq (\widebreve{r})^{-} D
  \subseteq (\widebreve{r})^{-} (b \downarrow_{\fproxeq'} c)$,
  and so $a' \fcov_{\fprox} (\widebreve{r})^{-}\left( b
  \downarrow_{\fproxeq'} c \right)$.  Hence, $r^{-}b
  \downarrow_{\fproxeq} r^{-}c \fcov_{\fprox}
  (\widebreve{r})^{-}\left( b \downarrow_{\fproxeq'} c \right)$.

  Conversely, let $r$ be a formal topology map from $(S,
  \fcov_{\fprox}, \fproxeq)$ to ${(S', \fcov'_{\fprox'}, \fproxeq')}$.
  We must show that $r^{\ddagger}$ given by \eqref{def:ddagger}
  satisfies~\ref{lem:SingleVersionJoinPresAppMap6}
  and~\ref{lem:SingleVersionJoinPresAppMap7}.
  \smallskip

  \noindent\ref{lem:SingleVersionJoinPresAppMap6}
  Let $a \fprox a'$. Since $S \fcov_{\fprox} r^{-}S'$ by \ref{FTM1}, there exists $A
  \in \Fin{r^{-}S'}$ such that $a \ll_{\fcov} A$, so there exist
  $A' \fprox_{L} A$ and $B \in \Fin{S'}$ such that $a \fcov A'$ and $A
  \mathrel{r_{L}} B$. Then $A' \mathrel{(r^{\ddagger})_{L}} B$.
  \smallskip

  \noindent\ref{lem:SingleVersionJoinPresAppMap7}
  Suppose
  $a \fprox a' \mathrel{r^{\ddagger}} b$ and  $a' \mathrel{r^{\ddagger}} c$.
  Then, $a' \fcov_{\fprox} r^{-}b$
  and $a' \fcov_{\fprox} r^{-}c$.
  Since 
  $b \fcov'_{\fprox'} \downset_{\fprox'} b$
  and $c \fcov'_{\fprox'} \downset_{\fprox'} c$, we have
  $a' \fcov_{\fprox} r^{-}\downset_{\fprox'}b$
  and
  $a' \fcov_{\fprox} r^{-}\downset_{\fprox'}c$ by \ref{BCM2}. Then,
  by ($\downset$-right) and \ref{FTM2},  we have
   $a' \fcov_{\fprox} r^{-} \left( \left(\downset_{\fprox'} b\right) 
  \downset_{\fproxeq'} \left(\downset_{\fprox'} c\right) \right) 
  \subseteq r^{-} \left(b \downarrow_{\fprox'} c \right)$. 
  Thus, there exists $A \in \Fin{r^{-}\left(b \downarrow_{\fprox'} c \right)}$ such that
  $a \ll_{\fcov} A$. Then, as in the proof
  of~\ref{lem:SingleVersionJoinPresAppMap6} above, we find $D \in
  \Fin{b \downarrow_{\fprox'} c}$ and $A' \in \Fin{S}$ such that
  $a \fcov A'$ and $A' \mathrel{(r^{\ddagger})_L} D$.
\end{proof}

Let $\ContBCov$ be the full subcategory of $\BCov$ consisting of
continuous basic covers, and let $\LKFTop$ be the full subcategory of
$\FTop$ consisting of locally compact formal topologies.
Then, it is straightforward to show that the composition of the
assignments $r \mapsto
r^{\dagger}$ and $r \mapsto \widebreve{r}$ given by~\eqref{def:dagger}
and~\eqref{def:breve}, respectively, determines a functor
$F \colon \SContFCov \to \ContBCov$.
By Proposition \ref{prop:BijectionBCovMapWbrel}, $F$ is full and
faithful, and it restricts to a full and faithful functor 
from $\ContFCovLoc$ to $\LKFTop$.

\begin{theorem}
$\SContFCov$ is equivalent to $\ContBCov$.
The equivalence restricts to an equivalence between $\ContFCovLoc$ and $\LKFTop$.
\end{theorem}
\begin{proof}
  First, we show that 
  the functor $F \colon \SContFCov \to \ContBCov$ described above is essentially surjective.
  Let $(S, \cov)$ be a continuous basic cover with a relation
  $\wb \subseteq S \times S$ satisfying \eqref{eq:wb}. Define a finitary cover
  ${\fcov}$ on $\Fin{S}$ by
  \begin{gather}
    \label{eq:wbCover}
    A \mathrel{\fcov} \mathcal{U} \defeqiv A \cov \bigcup
    \mathcal{U}.
   \intertext{Then, define a relation $\fprox$ on $\Fin{S}$ by}
      \label{eq:wbInterpolate}
      A \fprox B \defeqiv \exists C \in \Fin{S} \left( A \cov C
      \mathrel{\wb_{L}} B \right).
  \end{gather}
  One can easily show that $\fprox$ is idempotent by noting that $a
  \cov \wb^{-} a \cov \left\{ a \right\}$. We show that
  $(\Fin{S}, \fcov, \fprox)$ satisfies \eqref{eq:ContFCov}.
  First, suppose $A \fprox B \fcov \mathcal{U}$.
  Then, there exists $C \in \Fin{S}$ such that $A \cov C
  \mathrel{\wb_{L}} B$. Since $B \cov \wb^{-}\bigcup \mathcal{U}$,
  there exists $D \in \Fin{\wb^{-}\bigcup \mathcal{U}}$ such that
  $C \cov D$. Then, there exists $\mathcal{V} \in \FFin{S}$ such that
  $D = \bigcup \mathcal{V}$ and $\mathcal{V}
  \mathrel{(\wb_{L})_{L}} \mathcal{U}$. Thus $A \mathrel{\fcov}
  \mathcal{V} \fprox_{L} \mathcal{U}$. Conversely, if
  $A \mathrel{\fcov} \mathcal{V} \fprox_{L} \mathcal{U}$, then
  $A \fprox \bigcup \mathcal{U} \mathrel{\fcov} \mathcal{U}$.
  Therefore, $(\Fin{S}, \fcov, \fprox)$ is a strong continuous
  finitary cover.
  
  It is straightforward to show that $(S, \cov)$ is isomorphic to
  the basic cover $(\Fin{S}, \fcov_{\fprox})$ determined by 
  $(\Fin{S}, \fcov, \fprox)$ as in \eqref{eq:CovFromContFCov}. Specifically, we have an isomorphism $r
  \colon (S, \cov) \to (\Fin{S}, \fcov_{\fprox})$ defined by 
  \[
    a \mathrel{r} A \defeqiv a \cov A
  \]
  with an inverse $s \colon (\Fin{S}, \fcov_{\fprox}) \to (S, \cov)$
  defined by 
  \[
    A \mathrel{s} a \defeqiv A \cov \left\{ a \right\}.
  \]
  Hence, the functor $F \colon \SContFCov \to \ContBCov$ is essentially
  surjective.
   
  Next, we show that for any locally compact formal
  topology $(S, \cov, \leq)$ equipped with a relation $\wb$ satisfying
  \eqref{eq:wb}, the strong
  continuous finitary cover ${(\Fin{S}, \fcov, \fprox)}$ defined by
  \eqref{eq:wbCover} and \eqref{eq:wbInterpolate} is localized.
  Suppose  $A \fprox B \fcov
  \mathcal{U}$. Then
  $B \cov \wb^{-} \bigcup \mathcal{U} \downarrow_{\leq} \wb^{-} B$,
  so there exists $C \in \Fin{\wb^{-} \bigcup \mathcal{U}
  \downarrow_{\leq} \wb^{-}B}$ such that $A \cov C$. Thus, there
  exists $\mathcal{V} \in \FFin{S}$ such that $C =
  \bigcup{\mathcal{V}}$,
  $\mathcal{V} \fprox_{L} \mathcal{U}$, and
  $\mathcal{V} \fprox_{L} \left\{ B \right\}$.
  Hence $A \fcov \mathcal{V} \in \Fin{B \downarrow_{\fprox} \mathcal{U}}$, and so
  $(\Fin{S}, \fcov, \fprox)$ is localized. 

  As we have shown in the first part, as a basic cover, $(S, \cov, \leq)$ is
  isomorphic to the locally compact formal topology $(\Fin{S}, \fcov_{\fprox}, \fproxeq)$
  determined by ${(\Fin{S}, \fcov, \fprox)}$. 
  Since any isomorphism between the underlying basic covers is a
  formal topology map~\cite[Proposition 5.6]{ConvFTop}, the
  restriction of $F \colon \SContFCov \to \ContBCov$ to
  $\ContFCovLoc$ and $\LKFTop$ is essentially surjective.
\end{proof}

\section*{Acknowledgements}
We thank the following people: anonymous referees for many useful
suggestions; Steve Vickers for valuable feedback on a preliminary
version of the paper; Francesco Ciraulo, Milly Maietti, and Giovanni
Sambin for illuminating discussions on the predicativity problem of
locally compact formal topologies.

\bibliographystyle{alpha}
\bibliography{refs}
\end{document}